\documentclass[acmsmall,screen,noacm]{acmart}

\usepackage{mathpartir}
\usepackage{boxedminipage}
\usepackage{amsmath}
\usepackage{amssymb}
\usepackage{stmaryrd}
\usepackage{semantic}
\usepackage[capitalize,noabbrev]{cleveref}
\usepackage{xcolor}
\usepackage{newunicodechar}
\usepackage{bbold}
\usepackage{scalerel}
\usepackage[cal=boondoxo]{mathalfa}

\renewcommand\footnotetextcopyrightpermission[1]{} 

\date{}

\title{Handling Higher-Order Effects}

\author{Cas van der Rest}
\affiliation{
  \institution{Delft University of Technology}
  \country{The Netherlands}
}
\email{c.r.vanderrest@tudelft.nl}

\author{Jaro Reinders}
\affiliation{
  \institution{Delft University of Technology}
  \country{The Netherlands}
}
\email{j.s.reinders@tudelft.nl}

\author{Casper Bach Poulsen}
\affiliation{
  \institution{Delft University of Technology}
  \country{The Netherlands}
}
\email{c.b.poulsen@tudelft.nl}

\makeatletter
\@ifundefined{lhs2tex.lhs2tex.sty.read}%
  {\@namedef{lhs2tex.lhs2tex.sty.read}{}%
   \newcommand\SkipToFmtEnd{}%
   \newcommand\EndFmtInput{}%
   \long\def\SkipToFmtEnd#1\EndFmtInput{}%
  }\SkipToFmtEnd

\newcommand\ReadOnlyOnce[1]{\@ifundefined{#1}{\@namedef{#1}{}}\SkipToFmtEnd}
\usepackage{amstext}
\usepackage{amssymb}
\usepackage{stmaryrd}
\DeclareFontFamily{OT1}{cmtex}{}
\DeclareFontShape{OT1}{cmtex}{m}{n}
  {<5><6><7><8>cmtex8
   <9>cmtex9
   <10><10.95><12><14.4><17.28><20.74><24.88>cmtex10}{}
\DeclareFontShape{OT1}{cmtex}{m}{it}
  {<-> ssub * cmtt/m/it}{}

\DeclareFontShape{OT1}{cmtt}{bx}{n}
  {<5><6><7><8>cmtt8
   <9>cmbtt9
   <10><10.95><12><14.4><17.28><20.74><24.88>cmbtt10}{}
\DeclareFontShape{OT1}{cmtex}{bx}{n}
  {<-> ssub * cmtt/bx/n}{}

\newcommand{\Conid}[1]{\mathit{#1}}
\newcommand{\Varid}[1]{\mathit{#1}}
\newcommand{\anonymous}{\kern0.06em \vbox{\hrule\@width.5em}}
\newcommand{\plus}{\mathbin{+\!\!\!+}}

\renewcommand{\leq}{\leqslant}

\usepackage{polytable}

\@ifundefined{mathindent}%
  {\newdimen\mathindent\mathindent\leftmargini}%
  {}%

\def\resethooks{%
  \global\let\SaveRestoreHook\empty
  \global\let\ColumnHook\empty}
\newcommand*{\savecolumns}[1][default]%
  {\g@addto@macro\SaveRestoreHook{\savecolumns[#1]}}
\newcommand*{\restorecolumns}[1][default]%
  {\g@addto@macro\SaveRestoreHook{\restorecolumns[#1]}}
\newcommand*{\aligncolumn}[2]%
  {\g@addto@macro\ColumnHook{\column{#1}{#2}}}

\resethooks

\newcommand{\onelinecommentchars}{\quad-{}- }
\newcommand{\commentbeginchars}{\enskip\{-}
\newcommand{\commentendchars}{-\}\enskip}

\newcommand{\visiblecomments}{%
  \let\onelinecomment=\onelinecommentchars
  \let\commentbegin=\commentbeginchars
  \let\commentend=\commentendchars}

\newcommand{\invisiblecomments}{%
  \let\onelinecomment=\empty
  \let\commentbegin=\empty
  \let\commentend=\empty}

\visiblecomments

\newlength{\blanklineskip}
\newcommand{\hsindent}[1]{\quad}
\let\hspre\empty
\let\hspost\empty

\EndFmtInput
\makeatother
\ReadOnlyOnce{polycode.fmt}%
\makeatletter

\newcommand{\hsnewpar}[1]%
  {{\parskip=0pt\parindent=0pt\par\vskip #1\noindent}}

\newcommand{\hscodestyle}{}

\newcommand{\sethscode}[1]%
  {\expandafter\let\expandafter\hscode\csname #1\endcsname
   \expandafter\let\expandafter\endhscode\csname end#1\endcsname}

  {\par\noindent
   \advance\leftskip\mathindent
   \hscodestyle
   \let\\=\@normalcr
   \let\hspre\(\let\hspost\)%
   \pboxed}%
  {\endpboxed\)%
   \par\noindent
   \ignorespacesafterend}

  {\hsnewpar\abovedisplayskip
   \advance\leftskip\mathindent
   \hscodestyle
   \let\hspre\(\let\hspost\)%
   \pboxed}%
  {\endpboxed%
   \hsnewpar\belowdisplayskip
   \ignorespacesafterend}

  {\hsnewpar\abovedisplayskip
   \advance\leftskip\mathindent
   \hscodestyle
   \let\\=\@normalcr
   \(\pboxed}%
  {\endpboxed\)%
   \hsnewpar\belowdisplayskip
   \ignorespacesafterend}

\newcommand{\plainhs}{\sethscode{plainhscode}}

\plainhs

  {\hsnewpar\abovedisplayskip
   \advance\leftskip\mathindent
   \hscodestyle
   \let\\=\@normalcr
   \(\parray}%
  {\endparray\)%
   \hsnewpar\belowdisplayskip
   \ignorespacesafterend}

  {\parray}{\endparray}

  {\(\parray}{\endparray\)}

\def\codeframewidth{\arrayrulewidth}
\RequirePackage{calc}

  {\parskip=\abovedisplayskip\par\noindent
   \hscodestyle
   \arrayrulewidth=\codeframewidth
   \tabular{@{}|p{\linewidth-2\arraycolsep-2\arrayrulewidth-2pt}|@{}}%
   \hline\framedhslinecorrect\\{-1.5ex}%
   \let\endoflinesave=\\
   \let\\=\@normalcr
   \(\pboxed}%
  {\endpboxed\)%
   \framedhslinecorrect\endoflinesave{.5ex}\hline
   \endtabular
   \parskip=\belowdisplayskip\par\noindent
   \ignorespacesafterend}

\newcommand{\framedhslinecorrect}[2]%
  {#1[#2]}

  {\(\def\column##1##2{}%
   \let\>\undefined\let\<\undefined\let\\\undefined
   \newcommand\>[1][]{}\newcommand\<[1][]{}\newcommand\\[1][]{}%
   \def\fromto##1##2##3{##3}%
   }{\) }%

  {\let\orighscode=\hscode
   \let\origendhscode=\endhscode
   \def\endhscode{\def\hscode{\endgroup\def\@currenvir{hscode}\\}\begingroup}
   \orighscode\def\hscode{\endgroup\def\@currenvir{hscode}}}%
  {\origendhscode
   \global\let\hscode=\orighscode
   \global\let\endhscode=\origendhscode}%

\makeatother
\EndFmtInput
\DeclareRobustCommand\longtwoheadmapsto {\vdash\joinrel\!\relbar\joinrel\twoheadrightarrow}

\definecolor{opcolor}{RGB}{0,60,120}
\definecolor{anncolor}{RGB}{75,75,125}
\definecolor{hlcolor}{RGB}{230,230,230}

\newcommand{\susp}[1]{\ensuremath{\{#1\}}}
\newcommand{\ann}[2]{\ensuremath{#2^{{\color{anncolor}#1}}}}

\newcommand{\hl}[1]{\colorbox{hlcolor}{\ensuremath{#1}}}
\newcommand{\inE}[1]{\ensuremath{E^{ℓ}[#1]}}
\newcommand{\bsubst}[2]{\ensuremath{\overline{#1/#2}}}

\definecolor{opcolor}{RGB}{0,60,120}
\definecolor{kwcolor}{RGB}{0,120,90}

\newcommand{\op}[1]{{\color{opcolor}\textsf{#1}}}
\newcommand{\kw}[1]{{\color{opcolor}\textbf{\textsf{#1}}}}

\newunicodechar{◂}{} 

\newcommand\superequiv{\mathrel{\rlap{\raisebox{\fontdimen22\textfont2}{$=$}}\raisebox{-0.5\fontdimen22\textfont2}{$ = $}}}

\newunicodechar{×}{\ensuremath{\mathnormal\times}}
\newunicodechar{→}{\ensuremath{\mathnormal\to}}
\newunicodechar{←}{\ensuremath{\mathnormal\leftarrow}}
\newunicodechar{⟦}{\ensuremath{\mathnormal\llbracket}}
\newunicodechar{⟧}{\ensuremath{\mathnormal\rrbracket}}
\newunicodechar{ℕ}{\ensuremath{\mathbb{N}}}
\newunicodechar{⊕}{\ensuremath{\mathnormal\oplus}}
\newunicodechar{⊞}{\ensuremath{\mathnormal\boxplus}}
\newunicodechar{⇒}{\ensuremath{\mathnormal\Rightarrow}}
\newunicodechar{⟨}{\ensuremath{\mathnormal\langle}}
\newunicodechar{⟩}{\ensuremath{\mathnormal\rangle}}
\newunicodechar{∪}{\ensuremath{\mathnormal\cup}}
\newunicodechar{⦃}{\ensuremath{\mathnormal\{\mskip-4.5mu\mid}}
\newunicodechar{⦄}{\ensuremath{\mathnormal\mid\mskip-4.5mu\}}}
\newunicodechar{⊎}{\ensuremath{\mathnormal\uplus}}
\newunicodechar{✴}{\ensuremath{\mathnormal\ast}}
\newunicodechar{↦}{\ensuremath{\mathnormal\mapsto}}
\newunicodechar{≡}{\ensuremath{\mathnormal\equiv}}
\newunicodechar{∀}{\ensuremath{\mathnormal\forall}}
\newunicodechar{∙}{\ensuremath{\mathnormal\bullet}}
\newunicodechar{≣}{\ensuremath{\mathnormal\superequiv}}
\newunicodechar{▿}{\ensuremath{\mathnormal\triangledown}}
\newunicodechar{≤}{\ensuremath{\mathnormal\leq}}
\newunicodechar{↔}{\ensuremath{\mathnormal\leftrightarrow}}
\newunicodechar{⊂}{\ensuremath{\mathnormal\subset}}
\newunicodechar{∘}{\ensuremath{\mathnormal\circ}}
\newunicodechar{∃}{\ensuremath{\mathnormal\exists}}
\newunicodechar{↓}{\ensuremath{\mathnormal\downarrow}}
\newunicodechar{↑}{\ensuremath{\mathnormal\uparrow}}
\newunicodechar{⊔}{\ensuremath{\mathnormal\sqcup}}
\newunicodechar{⊢}{\ensuremath{\mathnormal\vdash}}
\newunicodechar{⊨}{\ensuremath{\mathnormal{\vDash}}}
\newunicodechar{◇}{\ensuremath{\mathnormal\diamond}}
\newunicodechar{⊙}{\ensuremath{\mathnormal\odot}}
\newunicodechar{⊤}{\ensuremath{\mathnormal\top}}
\newunicodechar{⊥}{\ensuremath{\mathnormal\bot}}
\newunicodechar{∣}{\ensuremath{\mathnormal\mid}}
\newunicodechar{′}{\ensuremath{^\prime}}
\newunicodechar{∅}{\ensuremath{\emptyset}}
\newunicodechar{≺}{\ensuremath{\mathnormal{\prec}}}
\newunicodechar{≼}{\ensuremath{\mathnormal{\preceq}}}
\newunicodechar{∩}{\ensuremath{\mathnormal{\cap}}}
\newunicodechar{⟪}{\ensuremath{\langle\kern-.2em\langle}}
\newunicodechar{⟫}{\ensuremath{\rangle\kern-.2em\rangle}}
\newunicodechar{⊆}{\ensuremath{\mathnormal{\subseteq}}}
\newunicodechar{▻}{\ensuremath{\mathnormal\vartriangleright}}
\newunicodechar{∷}{\ensuremath{::}}
\newunicodechar{►}{\ensuremath{\mathnormal{\blacktriangleright}}}
\newunicodechar{▹}{\ensuremath{\mathnormal{\triangleright}}}
\newunicodechar{□}{\ensuremath{\mathnormal{\square}}}
\newunicodechar{⋯}{\ensuremath{\mathnormal{\ldots}}}
\newunicodechar{▣}{\ensuremath{\mathnormal{\ldots}}}
\newunicodechar{⋮}{\ensuremath{\mathnormal{\quad\quad\vdots}}}
\newunicodechar{∈}{\ensuremath{\mathnormal{\in}}}
\newunicodechar{⊑}{\ensuremath{\mathbin{\sqsubseteq}}}
\newunicodechar{𝓑}{\ensuremath{\mathnormal{\gg\!\!=}}}
\newunicodechar{𝒜}{\framebox{\ensuremath{\mathnormal{\mathcal{A}}}}}
\newunicodechar{·}{\ensuremath{\mathnormal{.}}}
\newunicodechar{∉}{\ensuremath{\mathnormal{\notin}}}
\newunicodechar{𝔭}{\ensuremath{\mathnormal{\mathfrak{p}}}}
\newunicodechar{𝔮}{\ensuremath{\mathnormal{\mathfrak{q}}}}
\newunicodechar{ℓ}{\ensuremath{\mathnormal{\ell}}}
\newunicodechar{⊙}{\ensuremath{\mathnormal{\odot}}}
\newunicodechar{★}{\ensuremath{\mathnormal{\star}}}
\newunicodechar{⋐}{\ensuremath{\mathnormal{\inc}}}

\newunicodechar{φ}{\ensuremath{\phi}}
\newunicodechar{Φ}{\ensuremath{\Phi}}
\newunicodechar{μ}{\ensuremath{\mu}}
\newunicodechar{σ}{\ensuremath{\sigma}}
\newunicodechar{ξ}{\ensuremath{\xi}}
\newunicodechar{Ξ}{\ensuremath{\Xi}}
\newunicodechar{λ}{\ensuremath{\lambda}}
\newunicodechar{ε}{\ensuremath{\epsilon}}
\newunicodechar{Σ}{\ensuremath{\Sigma}}
\newunicodechar{Δ}{\ensuremath{\Delta}}
\newunicodechar{Π}{\ensuremath{\Pi}}
\newunicodechar{Γ}{\ensuremath{\Gamma}}
\newunicodechar{η}{\ensuremath{\eta}}
\newunicodechar{ζ}{\ensuremath{\zeta}}
\newunicodechar{Λ}{\ensuremath{\Lambda}}
\newunicodechar{τ}{\ensuremath{\tau}}
\newunicodechar{κ}{\ensuremath{\kappa}}
\newunicodechar{Ψ}{\ensuremath{\Psi}}
\newunicodechar{υ}{\ensuremath{\upsilon}}
\newunicodechar{ρ}{\ensuremath{\rho}}
\newunicodechar{α}{\ensuremath{\alpha}}
\newunicodechar{β}{\ensuremath{\beta}}
\newunicodechar{δ}{\ensuremath{\delta}}

\newunicodechar{₀}{$_{0}$}
\newunicodechar{₁}{$_{1}$}
\newunicodechar{₂}{$_{2}$}
\newunicodechar{₃}{$_{3}$}
\newunicodechar{₄}{$_{4}$}
\newunicodechar{₅}{$_{5}$}
\newunicodechar{₆}{$_{6}$}
\newunicodechar{₇}{$_{7}$}
\newunicodechar{₈}{$_{8}$}
\newunicodechar{₉}{$_{9}$}
\newunicodechar{¹}{$^{1}$}
\newunicodechar{⁻}{$^{-}$}
\newunicodechar{ᴬ}{$^{A}$}
\newunicodechar{ᴮ}{$^{B}$}
\newunicodechar{ˣ}{$^{×}$}
\newunicodechar{ᶠ}{$^{F}$}
\newunicodechar{ˡ}{$^l$}
\newunicodechar{ʳ}{$^r$}
\newunicodechar{ᴰ}{$^D$}
\newunicodechar{ᵇ}{$^b$}
\newunicodechar{ᵐ}{$_m$}
\newunicodechar{ⁿ}{$_n$}
\newunicodechar{ₙ}{$_n$}
\newunicodechar{ₘ}{$_m$}
\newunicodechar{ᵛ}{$_v$}
\newunicodechar{ᵍ}{$_f$}
\newunicodechar{ᵢ}{$_i$}
\newunicodechar{ₗ}{$_l$}
\newunicodechar{ᵣ}{$_r$}
\newunicodechar{ᶜ}{$^c$}
\newunicodechar{∶}{$:$}

\renewcommand\hscodestyle{%
   \setlength\leftskip{0.5cm}%
}

\begin{document}

\begin{abstract}
  Algebraic effect handlers is a programming paradigm where programmers can declare their own syntactic operations, and modularly define the semantics of these using effect handlers.
  However, we cannot directly define algebraic effect handlers for many higher-order operations (or higher-order effects)---i.e., operations that have computations as parameters.
  Examples of such higher-order effects include common programming features, such as try-catch exception handlers, function abstraction, and more.
  In this paper we present a new kind of effect handler that addresses this shortcoming.
  Our effect handler approach is closely related to previous work on scoped effect handlers, which also supports higher-order effects.
  A key difference is that our effect handlers make it easy to understand separate (higher-order) effects as separate concerns, since effects do not interact.
  In contrast, effect interaction is the default with scoped effect handlers.
  While separate concerns is the default with our handlers, it is also possible to define handlers where effects interact.
\end{abstract}

\maketitle

\section{Introduction}


Algebraic effects and handlers~\cite{DBLP:conf/esop/PlotkinP09} is an approach to separately define and modularly compose the syntax and semantics of effectful operations.
Using the approach, programmers can declare their own syntax of operations and use these to write concise and effectful programs.
For example, by declaring a \ensuremath{\op{rand}} operation that represents generating a random number, and \ensuremath{\op{push}} and \ensuremath{\op{pop}} operations that represent pushing and popping elements from a stack, we can write a small program that represents mutating the state of an underlying stack:
\begin{hscode}\SaveRestoreHook
\column{B}{@{}>{\hspre}l<{\hspost}@{}}%
\column{3}{@{}>{\hspre}l<{\hspost}@{}}%
\column{E}{@{}>{\hspre}l<{\hspost}@{}}%
\>[3]{}\Varid{prog}_{0}\mathrel{=}\mathbf{let}\;\Varid{x}\mathrel{=}\op{rand}\;\mathbf{in}\;\op{push}\;\Varid{x};\op{pop};\op{pop}{}\<[E]%
\ColumnHook
\end{hscode}\resethooks
We can give meaning to this program by writing two separate algebraic effect handlers; e.g., a handler \ensuremath{\Varid{hRand}} that defines the effects of the \ensuremath{\op{rand}} operation, and a handler \ensuremath{\Varid{hStack}} that defines the effects of the \ensuremath{\op{push}} and \ensuremath{\op{pop}} operations.
By applying these handlers one after the other to our small program, we get as output the value that is on top of the the initial stack of the program:
\begin{hscode}\SaveRestoreHook
\column{B}{@{}>{\hspre}l<{\hspost}@{}}%
\column{3}{@{}>{\hspre}l<{\hspost}@{}}%
\column{E}{@{}>{\hspre}l<{\hspost}@{}}%
\>[3]{}\Varid{hStack}\;[\mskip1.5mu \mathrm{1}\mskip1.5mu]\;(\Varid{hRand}\;\Varid{prog}_{0}) \longtwoheadmapsto \mathrm{1}{}\<[E]%
\ColumnHook
\end{hscode}\resethooks
Many algebraic effect handlers, such as \ensuremath{\Varid{hStack}} and \ensuremath{\Varid{hRand}}, can be defined in a way that the order we apply handlers in is irrelevant.
For example, we can safely swap the order of \ensuremath{\Varid{hStack}} and \ensuremath{\Varid{hRand}} in above to give the same final result, since \ensuremath{\Varid{hStack}} and \ensuremath{\Varid{hRand}} have separate concerns.
This support for separation of concerns makes algebraic effect handlers attractive as a programming paradigm, and as a paradigm for building (domain-specific) languages from modular language components whose syntax and semantics is modularly defined in terms of effect handlers.
However, using algebraic effect handlers alone, it is challenging to define language components for some common effects found in programming languages, such as exception handling.

Consider again our program above: what happens if we \ensuremath{\op{pop}} from an empty stack?
One reasonable expectation is that \ensuremath{\op{pop}} should raise an exception that we can catch in our program.
To this end, we wish to introduce a new \ensuremath{\op{catch}} operation that we use as follows to catch a potential exception raised by \ensuremath{\op{pop};\op{pop}}:
\begin{hscode}\SaveRestoreHook
\column{B}{@{}>{\hspre}l<{\hspost}@{}}%
\column{3}{@{}>{\hspre}l<{\hspost}@{}}%
\column{E}{@{}>{\hspre}l<{\hspost}@{}}%
\>[3]{}\Varid{prog}_{1}\mathrel{=}\mathbf{let}\;\Varid{x}\mathrel{=}\op{rand}\;\mathbf{in}\;\op{push}\;\Varid{x};\op{catch}\;\susp{\op{pop};\op{pop}}\;\susp{\mathrm{0}}{}\<[E]%
\ColumnHook
\end{hscode}\resethooks
However, \ensuremath{\op{catch}} is a \emph{higher-order operation}; i.e., an operation that has computations as parameters.
Many higher-order operations, such as \ensuremath{\op{catch}}, are not \emph{algebraic}, so we cannot directly declare such operations and define them in terms of an algebraic effect handlers.

An alternative way to declare and define \ensuremath{\op{catch}} is to use \emph{scoped effect handlers}~\cite{DBLP:conf/haskell/WuSH14}.
With scoped effects, we can define a handler \ensuremath{\Varid{hCatch}} for \ensuremath{\op{catch}}.
Using that effect handler, what should be the result of running the following program if the initial stack is empty?
\begin{hscode}\SaveRestoreHook
\column{B}{@{}>{\hspre}l<{\hspost}@{}}%
\column{3}{@{}>{\hspre}l<{\hspost}@{}}%
\column{E}{@{}>{\hspre}l<{\hspost}@{}}%
\>[3]{}\Varid{prog}_{2}\mathrel{=}\Varid{prog}_{1};\op{pop}{}\<[E]%
\ColumnHook
\end{hscode}\resethooks
With scoped effects, the answer depends on what order you run handlers in.
Assuming \ensuremath{\Varid{hCatch}} returns a \ensuremath{\Conid{Just}} wrapped value if no exception propagates to the top-level and \ensuremath{\Conid{Nothing}} otherwise, and assuming that \ensuremath{\Varid{hRand}} generates the number \ensuremath{\mathrm{1337}}, these are the results of nesting the \ensuremath{\Varid{hCatch}} and \ensuremath{\Varid{hStack}} handlers in different orders:
\begin{hscode}\SaveRestoreHook
\column{B}{@{}>{\hspre}l<{\hspost}@{}}%
\column{3}{@{}>{\hspre}l<{\hspost}@{}}%
\column{E}{@{}>{\hspre}l<{\hspost}@{}}%
\>[3]{}\Varid{hCatch}\;(\Varid{hStack}\;[\mskip1.5mu \mskip1.5mu]\;(\Varid{hRand}\;\Varid{prog}_{2})) \longtwoheadmapsto \Conid{Just}\;\mathrm{1337}{}\<[E]%
\\
\>[3]{}\Varid{hStack}\;[\mskip1.5mu \mskip1.5mu]\;(\Varid{hCatch}\;(\Varid{hRand}\;\Varid{prog}_{2})) \longtwoheadmapsto \Conid{Nothing}{}\<[E]%
\ColumnHook
\end{hscode}\resethooks
The reason for this difference is that there is \emph{effect interaction} between the \ensuremath{\Varid{hCatch}} and \ensuremath{\Varid{hStack}} handlers.
By running \ensuremath{\Varid{hStack}} before \ensuremath{\Varid{hCatch}}, causes an exception being raised during evaluation of the first sub-computation in a \ensuremath{\op{catch}} operation to ``snap back'' to the state it was in before the \ensuremath{\op{catch}}.

The example above illustrates how scoped effects provide integrated support for \emph{effect interaction}, such as the interaction between the stateful stack and the exception catching effects.
This integrated support for effect interaction makes it possible to concisely define different semantics.
But it also precludes defining and understanding exception handling and other non-algebraic higher-order effects as separate concerns.

In this paper, we introduce \ensuremath{\Varid{λ}^{\Varid{hop}}}, a core calculus for \emph{handlers of higher-order effects}, that makes it possible to declare and handle higher-order operations, such as \ensuremath{\op{catch}}, in a way that we can understand their semantics as a concern that is separate from other effects.
That is, in a way that operations, such as \ensuremath{\op{catch}}, does not interact with other effects.

The key difference between scoped effect handlers and the handlers of higher-order effects in \ensuremath{\Varid{λ}^{\Varid{hop}}}, is that our handlers treat sub-computations, such as \ensuremath{\Varid{m}_{1}} and \ensuremath{\Varid{m}_{2}} in \ensuremath{\op{catch}\;\Varid{m}_{1}\;\Varid{m}_{2}} as ``delayed'' computations whose operations are only handled once the handler of \ensuremath{\op{catch}} forces the sub-computations to run.
In contrast, scoped effects treat \ensuremath{\Varid{m}_{1}} and \ensuremath{\Varid{m}_{2}} as ``scoped'' computations that we may pre-apply handlers to, before the handler of \ensuremath{\op{catch}} forces the sub-computations to run.
Thus, with our handlers, effects generally have a ``global'' interpretation, in contrast to the ``scoped'' interpretation of scoped effects.
However, as we illustrate in \cref{sec:approach}, handlers in \ensuremath{\Varid{λ}^{\Varid{hop}}} can also be used to encode a ``scoped'' effect semantics.

We make the following technical contributions:

\begin{itemize}
\item We illustrate in \cref{sec:approach} how common algebraic and higher-order effects can be defined using handlers of higher-order effects in \ensuremath{\Varid{λ}^{\Varid{hop}}}, and how to encode handlers with a scoped effect interpretation using handlers of higher-order effects.
\item We present (in \cref{sec:calculus}) \ensuremath{\Varid{λ}^{\Varid{hop}}}---a core calculus and type discipline for handlers of higher-order effects.  The paper is also accompanied by a prototype interpreter (attached as anonymous supplemental material) that encodes the reduction rules of the core calculus.
\item We formalize in \cref{sec:meta-theory} what it means for handlers to have separate concerns, and provide sufficient criteria for proving that a given handler provides separation of concerns.
\item We prove in \cref{sec:soc-handlers} that handlers for common algebraic and higher-order effects provides separation of concerns, and illustrate how certain handlers inherently interact with other effects.
\end{itemize}

\section{Programming with Handlers of (Higher-Order) Effects}
\label{sec:approach}

This section gives an informal overview of our approach and shows how it compares with closely related existing solutions.
We postpone a formal discussion of dynamic semantics and types to Section~\ref{sec:calculus}.


%
%

\subsection{Handling State}



As a first example we consider how to handle the state effect.
The example serves two purposes: (1) it acts as an introduction to effect handlers for the unfamiliar reader; and (2) it acts as an introduction to the \ensuremath{\Varid{λ}^{\Varid{hop}}} language whose core calculus we introduce in the next section.
The state effect has two operations, \ensuremath{\op{get}} and \ensuremath{\op{put}}, which retrieve and store a value respectively.
We can use these operations to define a computation which increments the current state by one:
\begin{hscode}\SaveRestoreHook
\column{B}{@{}>{\hspre}l<{\hspost}@{}}%
\column{E}{@{}>{\hspre}l<{\hspost}@{}}%
\>[B]{}\Varid{incr}\mathrel{=}\{\mskip1.5mu \mathbf{let}\;\Varid{x}\mathrel{=}\op{get}\kw{!}\;\mathbf{in}\;(\op{put}\;(\Varid{x}\mathbin{+}\mathrm{1}))\kw{!}\mskip1.5mu\}{}\<[E]%
\ColumnHook
\end{hscode}\resethooks
Here, \ensuremath{\{\mskip1.5mu } and \ensuremath{\mskip1.5mu\}} delimit a \emph{suspension}; i.e., a block of code that is treated as a value, akin to a 0-argument function value.\footnote{Our syntax and semantics of suspensions is inspired by the corresponding feature and syntax used in Frank~\cite{DBLP:journals/jfp/ConventLMM20}}
We use \ensuremath{\kw{!}} to \emph{enact} a suspension, akin to calling a 0-argument function value.
In \ensuremath{\Varid{λ}^{\Varid{hop}}}, computations are pure by default, and only suspensions are allowed to invoke effectful operations.
(This purity restriction is enforced by the type discipline we introduce in \cref{sec:calculus}.)
Since the \ensuremath{\Varid{incr}} program calls the effectful operations \ensuremath{\op{put}} and \ensuremath{\op{get}}, it is wrapped in a suspension.
Each of the calls to \ensuremath{\op{get}} and \ensuremath{\op{put}} operations are postfixed by \ensuremath{\kw{!}} since operations in \ensuremath{\Varid{λ}^{\Varid{hop}}} are represented as suspensions.

In order to run the \ensuremath{\Varid{incr}} program, we need to define a handler for the \ensuremath{\op{get}} and \ensuremath{\op{put}} operations.
Below we define a function \ensuremath{\Varid{hState}} which takes as input an initial state \ensuremath{\Varid{s}_{0}} and a suspension \ensuremath{\Varid{prog}} which may call \ensuremath{\op{put}} and \ensuremath{\op{get}}:
\begin{hscode}\SaveRestoreHook
\column{B}{@{}>{\hspre}l<{\hspost}@{}}%
\column{3}{@{}>{\hspre}l<{\hspost}@{}}%
\column{5}{@{}>{\hspre}l<{\hspost}@{}}%
\column{13}{@{}>{\hspre}l<{\hspost}@{}}%
\column{18}{@{}>{\hspre}l<{\hspost}@{}}%
\column{E}{@{}>{\hspre}l<{\hspost}@{}}%
\>[B]{}\Varid{hState}\;\Varid{s}_{0}\;\Varid{prog}\mathrel{=}{}\<[E]%
\\
\>[B]{}\hsindent{3}{}\<[3]%
\>[3]{}\kw{handle}\;\{\mskip1.5mu {}\<[E]%
\\
\>[3]{}\hsindent{2}{}\<[5]%
\>[5]{}\op{get}\;{}\<[13]%
\>[13]{}\Varid{s}\;\Varid{k}{}\<[18]%
\>[18]{}\mathbin{↦}\Varid{k}\;\Varid{s}\;\{\mskip1.5mu \Varid{s}\mskip1.5mu\},{}\<[E]%
\\
\>[3]{}\hsindent{2}{}\<[5]%
\>[5]{}\op{put}\;\Varid{s'}\;{}\<[13]%
\>[13]{}\Varid{s}\;\Varid{k}{}\<[18]%
\>[18]{}\mathbin{↦}\Varid{k}\;\Varid{s'}\;\{\mskip1.5mu ()\mskip1.5mu\},{}\<[E]%
\\
\>[3]{}\hsindent{2}{}\<[5]%
\>[5]{}\kw{return}\;\Varid{x}\;\Varid{s}{}\<[18]%
\>[18]{}\mathbin{↦}\{\mskip1.5mu (\Varid{x},\Varid{s})\mskip1.5mu\}{}\<[E]%
\\
\>[B]{}\hsindent{3}{}\<[3]%
\>[3]{}\mskip1.5mu\}\;\Varid{s}_{0}\;\Varid{prog}\kw{!}{}\<[E]%
\ColumnHook
\end{hscode}\resethooks
When \ensuremath{\Varid{hState}} is called, it enacts the suspended \ensuremath{\Varid{prog}} argument in the context of a \ensuremath{\kw{handle}} expression, which installs a handler of the \ensuremath{\op{get}} and \ensuremath{\op{put}} operations (the \ensuremath{\kw{return}} clause does not represent an operation; we describe its meaning shortly).
The installed handler is \emph{parameterized} by a value \ensuremath{\Varid{s}_{0}}, which represents the ``current'' state at the point where we start to evaluate \ensuremath{\Varid{prog}}.

When we encounter a \ensuremath{\op{get}} or \ensuremath{\op{put}} operation during evaluation of \ensuremath{\Varid{prog}}, we ``jump'' to the handler installed by \ensuremath{\Varid{hState}}, and evaluate the clause corresponding to the performed operation.
Each clause binds variables that represent the following: (1) the parameters that the operation being performed was applied to (e.g., the \ensuremath{\Varid{s'}} parameter of the \ensuremath{\op{put}} clause); (2) the ``current'' parameter of the handler (e.g., the \ensuremath{\Varid{s}} bound in the \ensuremath{\op{get}} and \ensuremath{\op{put}} clauses); and (3) a \emph{continuation} (the \ensuremath{\Varid{k}} bound in the \ensuremath{\op{get}} and \ensuremath{\op{put}} clauses).
The continuation represents the program context in which an operation is performed, and expects two arguments: an argument to be used instead of the ``current'' effect handler parameter, and a computation to be run in the ``resumed'' context.

We can think of an effect handler as behaving similarly to an exception handler: performing an operation is similar to throwing an exception, and handling an operation is similar to handling an exception.
The important difference between exception handlers and effect handlers is that effect handlers provide access to the continuation representing the program context where the operation is being performed.
In contrast, when we raise an exception, the program context that the exception was raised in is discarded, and can never be resumed.
By calling a continuation from an effect handler clause, effect handlers can ``resume'' the program in the context of a handler with an updated parameter.

For example, during evaluation of \ensuremath{(\Varid{hState}\;\mathrm{0}\;\Varid{incr})\kw{!}}, when we encounter the \ensuremath{\op{get}\kw{!}} operation in \ensuremath{\Varid{incr}}, we evaluate the right hand side expression of the \ensuremath{\op{get}} clause of the installed \ensuremath{\kw{handle}} expression.
In this right hand side expression, \ensuremath{\Varid{s}} is bound to \ensuremath{\mathrm{0}}, and \ensuremath{\Varid{k}} is bound to a continuation \ensuremath{\Varid{λ}\;\Varid{m}\;\!.\!\;\kw{handle}\;\{\mskip1.5mu \mathbin{…}\mskip1.5mu\}\;\mathrm{0}\;\{\mskip1.5mu \mathbf{let}\;\Varid{x}\mathrel{=}\Varid{m}\kw{!}\;\mathbf{in}\;(\op{put}\;(\Varid{x}\mathbin{+}\mathrm{1}))\kw{!}\mskip1.5mu\}}.

A key difference between \ensuremath{\Varid{λ}^{\Varid{hop}}} and existing effect handlers that we are aware of, is that the continuation bound by handler clauses in existing effect handlers expect a \emph{value}.
Effect handlers in \ensuremath{\Varid{λ}^{\Varid{hop}}} differ: continuations bound by handler clauses in \ensuremath{\Varid{λ}^{\Varid{hop}}} expect a \emph{suspension}.
This suspension will be enacted in the context where the handled operation occurred.
The \ensuremath{\Varid{hState}} handler does not make essential use of this added generality.
In \cref{sec:ho-effects} we consider a different example where we do make use of it.



The \ensuremath{\kw{return}} clause of \ensuremath{\Varid{hState}} defines what should happen when a computation returns a pure value.
In this case we pack the return value \ensuremath{\Varid{x}} and the final state \ensuremath{\Varid{s}} together in a tuple.
This tuple is wrapped in a suspension, because in general all handler clauses return suspended computations (in the case of the \ensuremath{\op{get}} and \ensuremath{\op{put}} clauses, the continuation \ensuremath{\Varid{k}} returns a suspended computation).




Any algebraic effect handler can be expressed in our system as a handler that passes a suspended value to the continuation. 
In the next section we discuss a handler for a more interesting effect where we pass a computation to a continuation instead.

%
%


\subsection{Handling Exception Catching}
\label{sec:ho-effects}


A common higher-order effect is exception catching. The effect consists of two operations, \ensuremath{\op{throw}} and \ensuremath{\op{catch}}. The \ensuremath{\op{throw}} operation has no arguments and just raises an exception which ``jumps'' to the nearest exception handler, or aborts if there is none.
We say that the exception catching effect is a higher-order effect because the \ensuremath{\op{catch}} operation has two higher-order arguments.
In Haskell, we might type this higher-order effect using the following type signature:\footnote{Haskell's popular \emph{monad transformer library} defines similar type classes for an exception operation that can carry a value: \url{https://hackage.haskell.org/package/mtl/docs/Control-Monad-Except.html}}
\begin{hscode}\SaveRestoreHook
\column{B}{@{}>{\hspre}l<{\hspost}@{}}%
\column{3}{@{}>{\hspre}l<{\hspost}@{}}%
\column{11}{@{}>{\hspre}l<{\hspost}@{}}%
\column{E}{@{}>{\hspre}l<{\hspost}@{}}%
\>[B]{}\mathbf{class}\;\Conid{MonadCatch}\;\Varid{m}\;\mathbf{where}{}\<[E]%
\\
\>[B]{}\hsindent{3}{}\<[3]%
\>[3]{}\Varid{throwM}{}\<[11]%
\>[11]{}\mathbin{::}\Varid{m}\;\Varid{a}{}\<[E]%
\\
\>[B]{}\hsindent{3}{}\<[3]%
\>[3]{}\Varid{catchM}{}\<[11]%
\>[11]{}\mathbin{::}\Varid{m}\;\Varid{a}\to \Varid{m}\;\Varid{a}\to \Varid{m}\;\Varid{a}{}\<[E]%
\ColumnHook
\end{hscode}\resethooks
The type of \ensuremath{\Conid{MonadCatch}} tells us that the \ensuremath{\Varid{catchM}} operation is parameterized by computations that have the same type as the \ensuremath{\Varid{catchM}} operation itself, which is why we call it a ``higher-order'' operation.
Our \ensuremath{\op{catch}} operation is higher-order in the same sense.
However, we save the discussion of how operations are typed in \ensuremath{\Varid{λ}^{\Varid{hop}}} to \cref{sec:calculus}.
Here, we illustrate how we can define a handler for the exception catching operation.
Consider the following:
\begin{hscode}\SaveRestoreHook
\column{B}{@{}>{\hspre}l<{\hspost}@{}}%
\column{E}{@{}>{\hspre}l<{\hspost}@{}}%
\>[B]{}\op{catch}\;\{\mskip1.5mu \mathbf{let}\;\Varid{x}\mathrel{=}\op{throw}\kw{!}\;\mathbf{in}\;\Varid{x}\mathbin{+}\mathrm{1}\mskip1.5mu\}\;\{\mskip1.5mu \mathrm{1}\mskip1.5mu\}{}\<[E]%
\ColumnHook
\end{hscode}\resethooks
Because the \ensuremath{\op{throw}} operation is inside the \ensuremath{\op{catch}} block, the program should produce the value 1 yielded by the second suspension argument of \ensuremath{\op{catch}} (which we will call the fallback computation of \ensuremath{\op{catch}}).
To match this behaviour, the \ensuremath{\op{catch}} operation should be handled by running the computation in the first argument; and if that computation raises an exception, the fallback computation is run.
The following handler implements this behavior:
\begin{hscode}\SaveRestoreHook
\column{B}{@{}>{\hspre}l<{\hspost}@{}}%
\column{3}{@{}>{\hspre}l<{\hspost}@{}}%
\column{5}{@{}>{\hspre}l<{\hspost}@{}}%
\column{7}{@{}>{\hspre}l<{\hspost}@{}}%
\column{18}{@{}>{\hspre}l<{\hspost}@{}}%
\column{E}{@{}>{\hspre}l<{\hspost}@{}}%
\>[B]{}\Varid{hCatch}\;\Varid{prog}\mathrel{=}{}\<[E]%
\\
\>[B]{}\hsindent{3}{}\<[3]%
\>[3]{}\kw{handle}\;\{\mskip1.5mu {}\<[E]%
\\
\>[3]{}\hsindent{2}{}\<[5]%
\>[5]{}\op{throw}\;{}\<[18]%
\>[18]{}\Varid{k}\mathbin{↦}\{\mskip1.5mu \Conid{Nothing}\mskip1.5mu\},{}\<[E]%
\\
\>[3]{}\hsindent{2}{}\<[5]%
\>[5]{}\op{catch}\;\Varid{m}_{1}\;\Varid{m}_{2}\;{}\<[18]%
\>[18]{}\Varid{k}\mathbin{↦}\Varid{k}\;\{\mskip1.5mu {}\<[E]%
\\
\>[5]{}\hsindent{2}{}\<[7]%
\>[7]{}\kw{match}\;(\Varid{hCatch}\;\Varid{m}_{1})\kw{!}{}\<[E]%
\\
\>[5]{}\hsindent{2}{}\<[7]%
\>[7]{}\mid \Conid{Nothing}\to \Varid{m}_{2}\kw{!}{}\<[E]%
\\
\>[5]{}\hsindent{2}{}\<[7]%
\>[7]{}\mid \Conid{Just}\;\Varid{x}\to \Varid{x}{}\<[E]%
\\
\>[3]{}\hsindent{2}{}\<[5]%
\>[5]{}\mskip1.5mu\},{}\<[E]%
\\
\>[3]{}\hsindent{2}{}\<[5]%
\>[5]{}\kw{return}\;\Varid{x}\mathbin{↦}\{\mskip1.5mu \Conid{Just}\;\Varid{x}\mskip1.5mu\}{}\<[E]%
\\
\>[B]{}\hsindent{3}{}\<[3]%
\>[3]{}\mskip1.5mu\}\;\Varid{prog}\kw{!}{}\<[E]%
\ColumnHook
\end{hscode}\resethooks
The first clause of the handler discards the continuation \ensuremath{\Varid{k}} and simply returns a suspended \ensuremath{\Conid{Nothing}}.
As with the \ensuremath{\Varid{hState}} handler, the \ensuremath{\Varid{hCatch}} handler may not handle all the effects in \ensuremath{\Varid{prog}}, and since \ensuremath{\Varid{prog}} may have more effects than the exception catching effect, the \ensuremath{\Varid{hCatch}} handler yields a suspension that represents the computation resulting from handling all exception catching effects.

The second clause of \ensuremath{\Varid{hCatch}} passes a suspension to \ensuremath{\Varid{k}} which matches on the result of calling \ensuremath{\Varid{hCatch}} recursively on \ensuremath{\Varid{m}_{1}}.
If \ensuremath{\Varid{m}_{1}} throws an exception, then the result of \ensuremath{(\Varid{hCatch}\;\Varid{m}_{1})\kw{!}} will be \ensuremath{\Conid{Nothing}}.  In this case, we catch the exception and run \ensuremath{\Varid{m}_{2}} in the context given by the continuation \ensuremath{\Varid{k}}.
Note that the continuation \ensuremath{\Varid{k}} may be a program context that itself contains handlers.
By passing the suspension to \ensuremath{\Varid{k}} which runs \ensuremath{\Varid{m}_{1}} and \ensuremath{\Varid{m}_{2}}, these computations will be run in the context of the continuation \ensuremath{\Varid{k}}.
This ensures that those operations performed during evaluation of \ensuremath{\Varid{m}_{1}} and \ensuremath{\Varid{m}_{2}} which \ensuremath{\Varid{k}} has handlers installed for, will be handled by the installed handlers in \ensuremath{\Varid{k}}.

To see why it is necessary to run \ensuremath{\Varid{m}_{1}} and \ensuremath{\Varid{m}_{2}} in the context of \ensuremath{\Varid{k}}, consider the following program:
\begin{hscode}\SaveRestoreHook
\column{B}{@{}>{\hspre}l<{\hspost}@{}}%
\column{E}{@{}>{\hspre}l<{\hspost}@{}}%
\>[B]{}(\Varid{hCatch}\;(\Varid{hState}\;\mathrm{0}\;(\op{catch}\;(\op{put}\;\mathrm{1})\;\{\mskip1.5mu ()\mskip1.5mu\})))\kw{!}{}\<[E]%
\ColumnHook
\end{hscode}\resethooks
After installing the \ensuremath{\Varid{hCatch}} and \ensuremath{\Varid{hState}} handlers, the \ensuremath{\op{catch}} operation will be handled.
This causes the \ensuremath{\op{catch}} clause of \ensuremath{\Varid{hCatch}} to be run where \ensuremath{\Varid{m}_{1}} is bound to \ensuremath{\op{put}\;\mathrm{1}}, \ensuremath{\Varid{m}_{2}} is bound to \ensuremath{\{\mskip1.5mu ()\mskip1.5mu\}}, and \ensuremath{\Varid{k}} is bound to \ensuremath{\Varid{λ}\;\Varid{m}\;\!.\!\;\Varid{hCatch}\;(\Varid{hState}\;\mathrm{0}\;\Varid{m})}.
Now, if we were to run \ensuremath{\Varid{m}_{1}} \emph{directly}, rather than pass it in a suspension to \ensuremath{\Varid{k}}, the \ensuremath{\op{put}\;\mathrm{1}} operation would escape the \ensuremath{\Varid{hState}} handler!
The static type discipline that we introduce in \cref{sec:calculus} ensures that it is not possible for effects to escape in this way, and that the \ensuremath{\Varid{hCatch}} handler above is well-typed and guaranteed to not go wrong.

If no exceptions are raised while evaluating the \ensuremath{\Varid{prog}} that \ensuremath{\Varid{hCatch}} accepts as argument, we wrap the return value of \ensuremath{\Varid{prog}} in a \ensuremath{\Conid{Just}}.

%
%



\subsection{Effect Interaction and Separation of Concerns}
\label{sec:effect-interaction}

The \ensuremath{\Varid{hCatch}} handler in the previous section treats exception catching as a separate concern, in the sense that the order in which we apply handlers is irrelevant for the \ensuremath{\Varid{hCatch}} effect and other handlers that treat their effects as separate concerns.
In \cref{sec:meta-theory} we define formally what it means for a handler to treat its effects as a separate concern.
Here, we illustrate it by example.
Consider the following program:
%
\begin{hscode}\SaveRestoreHook
\column{B}{@{}>{\hspre}l<{\hspost}@{}}%
\column{3}{@{}>{\hspre}l<{\hspost}@{}}%
\column{E}{@{}>{\hspre}l<{\hspost}@{}}%
\>[B]{}\Varid{transact}\mathrel{=}\{\mskip1.5mu {}\<[E]%
\\
\>[B]{}\hsindent{3}{}\<[3]%
\>[3]{}\mathbf{let}\;()\mathrel{=}(\op{put}\;\mathrm{1})\kw{!}\;\mathbf{in}{}\<[E]%
\\
\>[B]{}\hsindent{3}{}\<[3]%
\>[3]{}\mathbf{let}\;()\mathrel{=}(\op{catch}\;\{\mskip1.5mu \mathbf{let}\;()\mathrel{=}(\op{put}\;\mathrm{2})\kw{!}\;\mathbf{in}\;\op{throw}\kw{!}\mskip1.5mu\}\;\{\mskip1.5mu ()\mskip1.5mu\})\kw{!}\;\mathbf{in}{}\<[E]%
\\
\>[B]{}\hsindent{3}{}\<[3]%
\>[3]{}\op{get}\kw{!}{}\<[E]%
\\
\>[B]{}\mskip1.5mu\}{}\<[E]%
\ColumnHook
\end{hscode}\resethooks
%
The program produces similar results, independently of how we order \ensuremath{\Varid{hCatch}} and \ensuremath{\Varid{hState}}:
\begin{hscode}\SaveRestoreHook
\column{B}{@{}>{\hspre}l<{\hspost}@{}}%
\column{E}{@{}>{\hspre}l<{\hspost}@{}}%
\>[B]{}(\Varid{hCatch}\;(\Varid{hState}\;\mathrm{0}\;\Varid{transact}))\kw{!} \longtwoheadmapsto \Conid{Just}\;(\mathrm{2},\mathrm{2}){}\<[E]%
\\
\>[B]{}(\Varid{hState}\;\mathrm{0}\;(\Varid{hCatch}\;\Varid{transact}))\kw{!} \longtwoheadmapsto (\Conid{Just}\;\mathrm{2},\mathrm{2}){}\<[E]%
\ColumnHook
\end{hscode}\resethooks
While changing the order of the handlers changes the type of the result, it does not change the essential final return value of the computation.
To put it differently, the \ensuremath{\Varid{hCatch}} handler has no \emph{effect interaction}.


The lack of effect interaction is both a blessing and a curse.
On one hand we do not have to worry about how we order our handlers in order to understand the semantics of \ensuremath{\Varid{hCatch}}.
On the other hand we seem to lose some expressiveness.
Concretely, for some applications, we may wish to treat exceptions in a ``transactional'' manner, where raising an exception represents aborting a transaction and snapping back the state to what it was before entering a \ensuremath{\op{catch}}.
With scoped effects and with Haskell's monad transformers~\cite{DBLP:conf/popl/LiangHJ95}, we can get this transactional semantics by reordering handlers or monad transformers.
Thus, with scoped effect handlers, reordering handlers would give us two different semantics:
\begin{hscode}\SaveRestoreHook
\column{B}{@{}>{\hspre}l<{\hspost}@{}}%
\column{E}{@{}>{\hspre}l<{\hspost}@{}}%
\>[B]{}(\Varid{hCatch}\;(\Varid{hState}\;\mathrm{0}\;\Varid{transact}))\kw{!} \longtwoheadmapsto \Conid{Just}\;(\mathrm{1},\mathrm{1}){}\<[E]%
\\
\>[B]{}(\Varid{hState}\;\mathrm{0}\;(\Varid{hCatch}\;\Varid{transact}))\kw{!} \longtwoheadmapsto (\Conid{Just}\;\mathrm{2},\mathrm{2}){}\<[E]%
\ColumnHook
\end{hscode}\resethooks
%
In this subsection we show that we can define handlers in a way that we get a similar effect interaction by default, mimicking the semantics of scoped effect handlers and monad transformers.

A naive way to achieve interaction between the state effect and the exception catching effect, is to define an exception handler that explicitly interacts with state:
\begin{hscode}\SaveRestoreHook
\column{B}{@{}>{\hspre}l<{\hspost}@{}}%
\column{3}{@{}>{\hspre}l<{\hspost}@{}}%
\column{5}{@{}>{\hspre}l<{\hspost}@{}}%
\column{7}{@{}>{\hspre}l<{\hspost}@{}}%
\column{18}{@{}>{\hspre}l<{\hspost}@{}}%
\column{E}{@{}>{\hspre}l<{\hspost}@{}}%
\>[B]{}\Varid{hCatch}_{1}\;\Varid{prog}\mathrel{=}{}\<[E]%
\\
\>[B]{}\hsindent{3}{}\<[3]%
\>[3]{}\kw{handle}\;\{\mskip1.5mu {}\<[E]%
\\
\>[3]{}\hsindent{2}{}\<[5]%
\>[5]{}\op{throw}\;{}\<[18]%
\>[18]{}\Varid{k}\mathbin{↦}\{\mskip1.5mu \Conid{Nothing}\mskip1.5mu\},{}\<[E]%
\\
\>[3]{}\hsindent{2}{}\<[5]%
\>[5]{}\op{catch}\;\Varid{m}_{1}\;\Varid{m}_{2}\;{}\<[18]%
\>[18]{}\Varid{k}\mathbin{↦}\Varid{k}\;\{\mskip1.5mu {}\<[E]%
\\
\>[5]{}\hsindent{2}{}\<[7]%
\>[7]{}\hl{\;\mathbf{let}\;\Varid{s}\mathrel{=}\op{get}\kw{!}\;\mathbf{in}\;}{}\<[E]%
\\
\>[5]{}\hsindent{2}{}\<[7]%
\>[7]{}\kw{match}\;(\Varid{hCatch}_{1}\;\Varid{m}_{1})\kw{!}{}\<[E]%
\\
\>[5]{}\hsindent{2}{}\<[7]%
\>[7]{}\mid \Conid{Nothing}\to \hl{\;\mathbf{let}\;()\mathrel{=}(\op{put}\;\Varid{s})\kw{!}\;\mathbf{in}\;}\;\Varid{m}_{2}\kw{!}{}\<[E]%
\\
\>[5]{}\hsindent{2}{}\<[7]%
\>[7]{}\mid \Conid{Just}\;\Varid{x}\to \Varid{x}{}\<[E]%
\\
\>[3]{}\hsindent{2}{}\<[5]%
\>[5]{}\mskip1.5mu\},{}\<[E]%
\\
\>[3]{}\hsindent{2}{}\<[5]%
\>[5]{}\kw{return}\;\Varid{x}\mathbin{↦}\{\mskip1.5mu \Conid{Just}\;\Varid{x}\mskip1.5mu\}{}\<[E]%
\\
\>[B]{}\hsindent{3}{}\<[3]%
\>[3]{}\mskip1.5mu\}\;\Varid{prog}\kw{!}{}\<[E]%
\ColumnHook
\end{hscode}\resethooks
This explicit interaction between exception and state has two drawbacks. Firstly, It forces us to always run \ensuremath{\Varid{hState}} after \ensuremath{\Varid{hCatch}}. Secondly, it only provides interaction between these two effects. 
We can, however, define a handler that provides a effect interaction in general.

The idea is to define the handler clause of a \ensuremath{\op{catch}\;\Varid{m}_{1}\;\Varid{m}_{2}} operation in a way that it explicitly ``restarts'' the continuation in case \ensuremath{\Varid{m}_{1}} throws an exception.
To do so we add an additional \ensuremath{\op{abort}} operation to our exception catching effect, which mediates this restarting.
\begin{hscode}\SaveRestoreHook
\column{B}{@{}>{\hspre}l<{\hspost}@{}}%
\column{3}{@{}>{\hspre}l<{\hspost}@{}}%
\column{5}{@{}>{\hspre}l<{\hspost}@{}}%
\column{7}{@{}>{\hspre}l<{\hspost}@{}}%
\column{18}{@{}>{\hspre}l<{\hspost}@{}}%
\column{E}{@{}>{\hspre}l<{\hspost}@{}}%
\>[B]{}\Varid{hCatch}_{2}\;\Varid{prog}\mathrel{=}{}\<[E]%
\\
\>[B]{}\hsindent{3}{}\<[3]%
\>[3]{}\kw{handle}\;\{\mskip1.5mu {}\<[E]%
\\
\>[3]{}\hsindent{2}{}\<[5]%
\>[5]{}\op{throw}\;{}\<[18]%
\>[18]{}\Varid{k}\mathbin{↦}\{\mskip1.5mu \Conid{Nothing}\mskip1.5mu\},{}\<[E]%
\\
\>[3]{}\hsindent{2}{}\<[5]%
\>[5]{}\op{catch}\;\Varid{m}_{1}\;\Varid{m}_{2}\;{}\<[18]%
\>[18]{}\Varid{k}\mathbin{↦}\Varid{k}\;\{\mskip1.5mu {}\<[E]%
\\
\>[5]{}\hsindent{2}{}\<[7]%
\>[7]{}\kw{match}\;(\Varid{hCatch}_{2}\;\Varid{m}_{1})\kw{!}{}\<[E]%
\\
\>[5]{}\hsindent{2}{}\<[7]%
\>[7]{}\mid \Conid{Nothing}\to \hl{(\op{abort}\;(\Varid{k}\;\Varid{m}_{2}))\kw{!}}{}\<[E]%
\\
\>[5]{}\hsindent{2}{}\<[7]%
\>[7]{}\mid \Conid{Just}\;\Varid{x}\to \Varid{x}{}\<[E]%
\\
\>[3]{}\hsindent{2}{}\<[5]%
\>[5]{}\mskip1.5mu\},{}\<[E]%
\\
\>[3]{}\hsindent{2}{}\<[5]%
\>[5]{}\hl{\op{abort}\;\Varid{m}\;\Varid{k}\mathbin{↦}\Varid{m},}{}\<[E]%
\\
\>[3]{}\hsindent{2}{}\<[5]%
\>[5]{}\kw{return}\;\Varid{x}\mathbin{↦}\{\mskip1.5mu \Conid{Just}\;\Varid{x}\mskip1.5mu\}{}\<[E]%
\\
\>[B]{}\hsindent{3}{}\<[3]%
\>[3]{}\mskip1.5mu\}\;\Varid{prog}\kw{!}{}\<[E]%
\ColumnHook
\end{hscode}\resethooks
Using this handler, evaluation of our \ensuremath{\Varid{transact}} program first handles the \ensuremath{\op{put}\;\mathrm{1}} effect, to yield the following program state:
\begin{hscode}\SaveRestoreHook
\column{B}{@{}>{\hspre}l<{\hspost}@{}}%
\column{5}{@{}>{\hspre}l<{\hspost}@{}}%
\column{E}{@{}>{\hspre}l<{\hspost}@{}}%
\>[B]{}\Varid{hCatch}_{2}\;\{\mskip1.5mu \Varid{hState}\;\mathrm{1}\;\{\mskip1.5mu {}\<[E]%
\\
\>[B]{}\hsindent{5}{}\<[5]%
\>[5]{}\mathbf{let}\;()\mathrel{=}(\op{catch}\;\{\mskip1.5mu \mathbf{let}\;()\mathrel{=}(\op{put}\;\mathrm{2})\kw{!}\;\mathbf{in}\;\op{throw}\kw{!}\mskip1.5mu\}\;\{\mskip1.5mu ()\mskip1.5mu\})\kw{!}\;\mathbf{in}{}\<[E]%
\\
\>[B]{}\hsindent{5}{}\<[5]%
\>[5]{}\op{get}\kw{!}\mskip1.5mu\}\mskip1.5mu\}{}\<[E]%
\ColumnHook
\end{hscode}\resethooks
Next, the \ensuremath{\Varid{hCatch}_{2}} handler starts handling the \ensuremath{\op{catch}}, to give this sequence of steps, where the highlights indicate which operation synchronizes with which handler in each step.
\begin{hscode}\SaveRestoreHook
\column{B}{@{}>{\hspre}l<{\hspost}@{}}%
\column{5}{@{}>{\hspre}l<{\hspost}@{}}%
\column{7}{@{}>{\hspre}l<{\hspost}@{}}%
\column{E}{@{}>{\hspre}l<{\hspost}@{}}%
\>[B]{}\Varid{hCatch}_{2}\;\{\mskip1.5mu \hl{\Varid{hState}\;\mathrm{1}}\;\{\mskip1.5mu {}\<[E]%
\\
\>[B]{}\hsindent{5}{}\<[5]%
\>[5]{}\mathbf{let}\;()\mathrel{=}\kw{match}\;(\Varid{hCatch}_{2}\;\{\mskip1.5mu \mathbf{let}\;()\mathrel{=}(\hl{\op{put}}\;\mathrm{2})\kw{!}\;\mathbf{in}\;\op{throw}\kw{!}\mskip1.5mu\})\kw{!}{}\<[E]%
\\
\>[5]{}\hsindent{2}{}\<[7]%
\>[7]{}\mid \Conid{Nothing}\to (\op{abort}\;(\Varid{hCatch}_{2}\;\{\mskip1.5mu \Varid{hState}\;\mathrm{1}\;\{\mskip1.5mu \mathbf{let}\;()\mathrel{=}\{\mskip1.5mu ()\mskip1.5mu\}\kw{!}\;\mathbf{in}\;\op{get}\kw{!}\mskip1.5mu\}\mskip1.5mu\}))\kw{!}{}\<[E]%
\\
\>[5]{}\hsindent{2}{}\<[7]%
\>[7]{}\mid \Conid{Just}\;\Varid{x}\to \Varid{x}{}\<[E]%
\\
\>[B]{}\hsindent{5}{}\<[5]%
\>[5]{}\mathbf{in}\;\op{get}\kw{!}\mskip1.5mu\}\mskip1.5mu\}{}\<[E]%
\\
\>[B]{} \longmapsto {}\<[E]%
\\
\>[B]{}\Varid{hCatch}_{2}\;\{\mskip1.5mu \Varid{hState}\;\mathrm{2}\;\{\mskip1.5mu {}\<[E]%
\\
\>[B]{}\hsindent{5}{}\<[5]%
\>[5]{}\mathbf{let}\;()\mathrel{=}\kw{match}\;\hl{\Varid{hCatch}_{2}}\;\{\mskip1.5mu \hl{\op{throw}}\mathbin{!}\mskip1.5mu\}\kw{!}{}\<[E]%
\\
\>[5]{}\hsindent{2}{}\<[7]%
\>[7]{}\mid \Conid{Nothing}\to (\op{abort}\;(\Varid{hCatch}_{2}\;\{\mskip1.5mu \Varid{hState}\;\mathrm{1}\;\{\mskip1.5mu \mathbf{let}\;()\mathrel{=}\{\mskip1.5mu ()\mskip1.5mu\}\kw{!}\;\mathbf{in}\;\op{get}\kw{!}\mskip1.5mu\}\mskip1.5mu\}))\kw{!}{}\<[E]%
\\
\>[5]{}\hsindent{2}{}\<[7]%
\>[7]{}\mid \Conid{Just}\;\Varid{x}\to \Varid{x}{}\<[E]%
\\
\>[B]{}\hsindent{5}{}\<[5]%
\>[5]{}\mathbf{in}\;\op{get}\kw{!}\mskip1.5mu\}\mskip1.5mu\}{}\<[E]%
\\
\>[B]{} \longmapsto {}\<[E]%
\\
\>[B]{}\Varid{hCatch}_{2}\;\{\mskip1.5mu \Varid{hState}\;\mathrm{2}\;\{\mskip1.5mu {}\<[E]%
\\
\>[B]{}\hsindent{5}{}\<[5]%
\>[5]{}\mathbf{let}\;()\mathrel{=}\kw{match}\;\Conid{Nothing}{}\<[E]%
\\
\>[5]{}\hsindent{2}{}\<[7]%
\>[7]{}\mid \Conid{Nothing}\to (\op{abort}\;(\Varid{hCatch}_{2}\;\{\mskip1.5mu \Varid{hState}\;\mathrm{1}\;\{\mskip1.5mu \mathbf{let}\;()\mathrel{=}\{\mskip1.5mu ()\mskip1.5mu\}\kw{!}\;\mathbf{in}\;\op{get}\kw{!}\mskip1.5mu\}\mskip1.5mu\}))\kw{!}{}\<[E]%
\\
\>[5]{}\hsindent{2}{}\<[7]%
\>[7]{}\mid \Conid{Just}\;\Varid{x}\to \Varid{x}{}\<[E]%
\\
\>[B]{}\hsindent{5}{}\<[5]%
\>[5]{}\mathbf{in}\;\op{get}\kw{!}\mskip1.5mu\}\mskip1.5mu\}{}\<[E]%
\\
\>[B]{} \longmapsto {}\<[E]%
\\
\>[B]{}\hl{\Varid{hCatch}_{2}}\;\{\mskip1.5mu \Varid{hState}\;\mathrm{2}\;\{\mskip1.5mu \mathbf{let}\;()\mathrel{=}(\hl{\op{abort}}\;(\Varid{hCatch}_{2}\;\{\mskip1.5mu \Varid{hState}\;\mathrm{1}\;\{\mskip1.5mu \mathbf{let}\;()\mathrel{=}\{\mskip1.5mu ()\mskip1.5mu\}\kw{!}\;\mathbf{in}\;\op{get}\kw{!}\mskip1.5mu\}\mskip1.5mu\}))\kw{!}\;\mathbf{in}\;\op{get}\kw{!}\mskip1.5mu\}\mskip1.5mu\}{}\<[E]%
\\
\>[B]{} \longmapsto {}\<[E]%
\\
\>[B]{}(\Varid{hCatch}_{2}\;\{\mskip1.5mu \hl{\Varid{hState}\;\mathrm{1}}\;\{\mskip1.5mu \mathbf{let}\;()\mathrel{=}\{\mskip1.5mu ()\mskip1.5mu\}\kw{!}\;\mathbf{in}\;\hl{\op{get}}\kw{!}\mskip1.5mu\}\mskip1.5mu\})\kw{!}{}\<[E]%
\\
\>[B]{} \longmapsto {}\<[E]%
\\
\>[B]{}(\Varid{hCatch}_{2}\;\{\mskip1.5mu \Varid{hState}\;\mathrm{1}\;\{\mskip1.5mu \mathrm{1}\mskip1.5mu\}\mskip1.5mu\})\kw{!} \longmapsto (\Varid{hCatch}_{2}\;\{\mskip1.5mu (\mathrm{1},\mathrm{1})\mskip1.5mu\})\kw{!} \longmapsto \Conid{Just}\;(\mathrm{1},\mathrm{1}){}\<[E]%
\ColumnHook
\end{hscode}\resethooks
%
%
This demonstrates that \ensuremath{\Varid{λ}^{\Varid{hop}}} supports effect handling both with and without effect interaction.

\subsection{Non-Determinism is not a Separate Concern}
\label{sec:nondet-nonsep}



Not all effect handlers have separate concerns.
For example, the \emph{non-determinism} effect, which has a single operation \ensuremath{\op{flip}} which represents flipping a coin and returning a \ensuremath{\Conid{Boolean}}, non-deterministically.
For example:
\begin{hscode}\SaveRestoreHook
\column{B}{@{}>{\hspre}l<{\hspost}@{}}%
\column{13}{@{}>{\hspre}l<{\hspost}@{}}%
\column{E}{@{}>{\hspre}l<{\hspost}@{}}%
\>[B]{}\Varid{flippy}\mathrel{=}\{\mskip1.5mu {}\<[13]%
\>[13]{}\kw{match}\;\op{flip}\kw{!}{}\<[E]%
\\
\>[13]{}\mid \Conid{True}\to \op{throw}\kw{!}{}\<[E]%
\\
\>[13]{}\mid \Conid{False}\to \mathrm{1}\mskip1.5mu\}{}\<[E]%
\ColumnHook
\end{hscode}\resethooks
Following \citet{DBLP:conf/icfp/KammarLO13}, the \ensuremath{\op{flip}} effect is handled by running \emph{both} the program where the \ensuremath{\op{flip}} operation returns \ensuremath{\Conid{True}} and where it returns \ensuremath{\Conid{False}}, to accumulate a list of all possible program results; i.e.:
\begin{hscode}\SaveRestoreHook
\column{B}{@{}>{\hspre}l<{\hspost}@{}}%
\column{3}{@{}>{\hspre}l<{\hspost}@{}}%
\column{5}{@{}>{\hspre}l<{\hspost}@{}}%
\column{E}{@{}>{\hspre}l<{\hspost}@{}}%
\>[B]{}\Varid{hND}\;\Varid{prog}\mathrel{=}{}\<[E]%
\\
\>[B]{}\hsindent{3}{}\<[3]%
\>[3]{}\kw{handle}\;\{\mskip1.5mu {}\<[E]%
\\
\>[3]{}\hsindent{2}{}\<[5]%
\>[5]{}\op{flip}\;\Varid{k}\mathbin{↦}\{\mskip1.5mu (\Varid{k}\;\{\mskip1.5mu \Conid{True}\mskip1.5mu\})\kw{!}\plus (\Varid{k}\;\{\mskip1.5mu \Conid{False}\mskip1.5mu\})\kw{!}\mskip1.5mu\},{}\<[E]%
\\
\>[3]{}\hsindent{2}{}\<[5]%
\>[5]{}\kw{return}\;\Varid{x}\mathbin{↦}\{\mskip1.5mu [\mskip1.5mu \Varid{x}\mskip1.5mu]\mskip1.5mu\}{}\<[E]%
\\
\>[B]{}\hsindent{3}{}\<[3]%
\>[3]{}\mskip1.5mu\}\;\Varid{prog}\kw{!}{}\<[E]%
\ColumnHook
\end{hscode}\resethooks
The return clause returns a singleton list representing a single possible program result.

We illustrated in \cref{sec:effect-interaction} that the \ensuremath{\Varid{hCatch}} handler treats the exception effect as a separate concern.
In this section we illustrate that the non-determinism effect is not treated as a separate concern by the \ensuremath{\Varid{hND}} handler.
Concretely, by handling \ensuremath{\Varid{hND}} and \ensuremath{\Varid{hCatch}} in different orders, we get different results:
\begin{hscode}\SaveRestoreHook
\column{B}{@{}>{\hspre}l<{\hspost}@{}}%
\column{E}{@{}>{\hspre}l<{\hspost}@{}}%
\>[B]{}(\Varid{hND}\;(\Varid{hCatch}\;\Varid{flippy}))\kw{!} \longtwoheadmapsto [\mskip1.5mu \Conid{Nothing},\Conid{Just}\;\mathrm{1}\mskip1.5mu]{}\<[E]%
\\
\>[B]{}(\Varid{hCatch}\;(\Varid{hND}\;\Varid{flippy}))\kw{!} \longtwoheadmapsto \Conid{Nothing}{}\<[E]%
\ColumnHook
\end{hscode}\resethooks
We also get interaction if we handle \ensuremath{\Varid{hND}} and \ensuremath{\Varid{hState}} in different orders:
\begin{hscode}\SaveRestoreHook
\column{B}{@{}>{\hspre}l<{\hspost}@{}}%
\column{14}{@{}>{\hspre}l<{\hspost}@{}}%
\column{E}{@{}>{\hspre}l<{\hspost}@{}}%
\>[B]{}\Varid{flippy}_{1}\mathrel{=}\{\mskip1.5mu {}\<[14]%
\>[14]{}\kw{match}\;\op{flip}\kw{!}{}\<[E]%
\\
\>[14]{}\mid \Conid{True}\to \mathbf{let}\;()\mathrel{=}\op{put}\;\mathrm{1}\;\mathbf{in}\;\mathrm{0}{}\<[E]%
\\
\>[14]{}\mid \Conid{False}\to \op{get}\kw{!}\mskip1.5mu\}{}\<[E]%
\\[\blanklineskip]%
\>[B]{}(\Varid{hND}\;(\Varid{hState}\;\Varid{flippy}_{1}))\kw{!} \longtwoheadmapsto [\mskip1.5mu (\mathrm{0},\mathrm{1}),(\mathrm{0},\mathrm{0})\mskip1.5mu]{}\<[E]%
\\
\>[B]{}(\Varid{hState}\;(\Varid{hND}\;\Varid{flippy}_{1}))\kw{!} \longtwoheadmapsto ([\mskip1.5mu \mathrm{0},\mathrm{1}\mskip1.5mu],\mathrm{1}){}\<[E]%
\ColumnHook
\end{hscode}\resethooks
These examples demonstrate that the \ensuremath{\Varid{hND}} handler does not have separate concerns: the \ensuremath{\Conid{Just}\;\mathrm{1}} value produced by the first ordering of \ensuremath{\Varid{hND}} and \ensuremath{\Varid{hCatch}} above represents a result that is not produced by the second ordering; and the \ensuremath{(\mathrm{0},\mathrm{0})} result produced by the first ordering of \ensuremath{\Varid{hND}} and \ensuremath{\Varid{hState}} represents a result that is not produced by the second ordering of those handlers.
The \ensuremath{\Varid{hND}} handler generally affects the final results in a way that is fundamentally different from, e.g., the \ensuremath{\Varid{hState}}, \ensuremath{\Varid{hCatch}}, and many other effects.
We characterize this fundamental difference in \cref{sec:meta-theory} where we give a formal definition of what it means to have separate concerns, and in \cref{sec:soc-ho-eff} we show that \ensuremath{\Varid{hND}} does not have this property.






\subsection{Lambda Abstraction as an Effect}
\label{sec:lambda-effect}


In this section we further demonstrate the expressiveness of \ensuremath{\Varid{λ}^{\Varid{hop}}}, by showing how to handle a lambda abstraction effect.
The effect has three operations: (1) a function abstraction operation \ensuremath{\op{abs}\;\Varid{x}\;\Varid{m}} where \ensuremath{\Varid{x}} is a named parameter, and \ensuremath{\Varid{m}} is the suspended body of the function; (2) an operation representing a variable, \ensuremath{\op{var}\;\Varid{x}} where \ensuremath{\Varid{x}} is a named identifier; and (3) a function application operation \ensuremath{\op{app}\;\Varid{v}_{1}\;\Varid{v}_{2}} where \ensuremath{\Varid{v}_{1}} and \ensuremath{\Varid{v}_{2}} are values.
Recent work by \citet{DBLP:conf/aplas/BergSPW21} argues that these higher-order operations are beyond what both algebraic effect handlers and scoped effect handlers can define, and they proposed a Haskell-based library called \emph{latent effects}.
We relate to their work in \cref{sec:related}.
Here, we show how to handle the higher-order operations for the lambda abstraction effect using \ensuremath{\Varid{λ}^{\Varid{hop}}}.
%
%
%

Consider the following program which uses the state effect and the lambda abstraction effect:
\begin{hscode}\SaveRestoreHook
\column{B}{@{}>{\hspre}l<{\hspost}@{}}%
\column{12}{@{}>{\hspre}l<{\hspost}@{}}%
\column{20}{@{}>{\hspre}l<{\hspost}@{}}%
\column{27}{@{}>{\hspre}c<{\hspost}@{}}%
\column{27E}{@{}l@{}}%
\column{E}{@{}>{\hspre}l<{\hspost}@{}}%
\>[B]{}\Varid{lammy}\mathrel{=}\{\mskip1.5mu {}\<[12]%
\>[12]{}\mathbf{let}\;(){}\<[20]%
\>[20]{}\mathrel{=}(\op{put}\;\mathrm{1})\kw{!}\;\mathbf{in}{}\<[E]%
\\
\>[12]{}\mathbf{let}\;\Varid{f}{}\<[20]%
\>[20]{}\mathrel{=}(\op{abs}\;\text{\ttfamily \char34 x\char34}\;\{\mskip1.5mu (\op{var}\;\text{\ttfamily \char34 x\char34})\kw{!}\mathbin{+}\op{get}\kw{!}\mskip1.5mu\})\kw{!}\;\mathbf{in}{}\<[E]%
\\
\>[12]{}\mathbf{let}\;(){}\<[20]%
\>[20]{}\mathrel{=}(\op{put}\;\mathrm{2})\kw{!}\;\mathbf{in}{}\<[E]%
\\
\>[12]{}(\op{app}\;\Varid{f}\;\mathrm{3})\kw{!}{}\<[27]%
\>[27]{}\mskip1.5mu\}{}\<[27E]%
\ColumnHook
\end{hscode}\resethooks
We can interpret \ensuremath{\Varid{lammy}} in several ways, depending on how we handle the abstraction effect.
By handling the lambda effect in a way that provides separation of concerns, we can postpone the evaluation of the body of the lambda until it is called.
This treats the state effect as a \emph{dynamic} effect, meaning the program returns \ensuremath{\mathrm{3}\mathbin{+}\mathrm{2}\mathrel{=}\mathrm{5}}.
The \ensuremath{\Varid{hLam}_{0}} handler below assumes that (1) there is a \ensuremath{\Varid{lookup}\;\Varid{x}\;\Varid{r}} function for looking up a name \ensuremath{\Varid{x}} in an environment \ensuremath{\Varid{r}}; and (2) the value type used by the abstraction effect has a data type constructor \ensuremath{\Conid{Clo}\;\Varid{x}\;\Varid{m}\;\Varid{r}} where \ensuremath{\Varid{x}} is a named parameter, \ensuremath{\Varid{m}} is a suspended function body, and \ensuremath{\Varid{r}} is a closure environment:\footnote{The pattern match in the \ensuremath{\op{app}} clause of the handler is incomplete.  For the purpose of this subsection we assume a type system that would warn about this incomplete pattern, but accept the program as well-typed; similarly to how Haskell would type an incomplete \ensuremath{\mathbf{case}\mathbin{…}\mathbf{of}} expression.}
\begin{hscode}\SaveRestoreHook
\column{B}{@{}>{\hspre}l<{\hspost}@{}}%
\column{3}{@{}>{\hspre}l<{\hspost}@{}}%
\column{5}{@{}>{\hspre}l<{\hspost}@{}}%
\column{18}{@{}>{\hspre}l<{\hspost}@{}}%
\column{23}{@{}>{\hspre}l<{\hspost}@{}}%
\column{32}{@{}>{\hspre}l<{\hspost}@{}}%
\column{E}{@{}>{\hspre}l<{\hspost}@{}}%
\>[B]{}\Varid{hLam}_{0}\;\Varid{r}_{0}\;\Varid{prog}\mathrel{=}{}\<[E]%
\\
\>[B]{}\hsindent{3}{}\<[3]%
\>[3]{}\kw{handle}\;\{\mskip1.5mu {}\<[E]%
\\
\>[3]{}\hsindent{2}{}\<[5]%
\>[5]{}\op{abs}\;\Varid{x}\;\Varid{m}\;{}\<[18]%
\>[18]{}\Varid{r}\;\Varid{k}{}\<[23]%
\>[23]{}\mathbin{↦}\Varid{k}\;\Varid{r}\;\{\mskip1.5mu \Conid{Clo}\;\Varid{x}\;\Varid{m}\;\Varid{r}\mskip1.5mu\},{}\<[E]%
\\
\>[3]{}\hsindent{2}{}\<[5]%
\>[5]{}\op{app}\;\Varid{v}_{1}\;\Varid{v}_{2}\;{}\<[18]%
\>[18]{}\Varid{r}\;\Varid{k}{}\<[23]%
\>[23]{}\mathbin{↦}\Varid{k}\;\Varid{r}\;\{\mskip1.5mu {}\<[32]%
\>[32]{}\kw{match}\;\Varid{v}_{1}{}\<[E]%
\\
\>[32]{}\mid \Conid{Clo}\;\Varid{x}\;\Varid{m}\;\Varid{r'}\to \Varid{hLam}_{0}\;((\Varid{x},\Varid{v}_{2})\mathbin{::}\Varid{r'})\;\Varid{m}\mskip1.5mu\},{}\<[E]%
\\
\>[3]{}\hsindent{2}{}\<[5]%
\>[5]{}\op{var}\;\Varid{x}\;{}\<[18]%
\>[18]{}\Varid{r}\;\Varid{k}{}\<[23]%
\>[23]{}\mathbin{↦}\mathbf{let}\;\Varid{v}\mathrel{=}\Varid{lookup}\;\Varid{x}\;\Varid{r}\;\mathbf{in}\;\Varid{k}\;\Varid{r}\;\{\mskip1.5mu \Varid{v}\mskip1.5mu\},{}\<[E]%
\\
\>[3]{}\hsindent{2}{}\<[5]%
\>[5]{}\kw{return}\;\Varid{x}\;\Varid{r}{}\<[23]%
\>[23]{}\mathbin{↦}\{\mskip1.5mu \Varid{x}\mskip1.5mu\}{}\<[E]%
\\
\>[B]{}\hsindent{3}{}\<[3]%
\>[3]{}\mskip1.5mu\}\;\Varid{r}_{0}\;\Varid{prog}\kw{!}{}\<[E]%
\ColumnHook
\end{hscode}\resethooks
As we argue in \cref{sec:soc-ho-eff}, this handler has separate concerns.

An alternative handler for the lambda effect could use a ``scoped'' semantics, where we can use effect interaction between the abstraction effect handler and the state handler to treat the state effect as a \emph{static} effect, meaning the \ensuremath{\op{get}} inside of the \ensuremath{\op{abs}} bound to \ensuremath{\Varid{f}} is dereferenced under the state that the \ensuremath{\op{abs}} operation is first evaluated under.
This will cause the program above to return \ensuremath{\mathrm{3}\mathbin{+}\mathrm{1}\mathrel{=}\mathrm{4}}.
The handler below implements this semantics, using an \ensuremath{\Varid{abort'}} operation to reset a continuation, similarly to how we used \ensuremath{\op{abort}} to give a scoped semantics of exception catching.
\begin{hscode}\SaveRestoreHook
\column{B}{@{}>{\hspre}l<{\hspost}@{}}%
\column{3}{@{}>{\hspre}l<{\hspost}@{}}%
\column{5}{@{}>{\hspre}l<{\hspost}@{}}%
\column{18}{@{}>{\hspre}l<{\hspost}@{}}%
\column{23}{@{}>{\hspre}l<{\hspost}@{}}%
\column{32}{@{}>{\hspre}l<{\hspost}@{}}%
\column{E}{@{}>{\hspre}l<{\hspost}@{}}%
\>[B]{}\Varid{hLam}_{1}\;\Varid{r}_{0}\;\Varid{prog}\mathrel{=}{}\<[E]%
\\
\>[B]{}\hsindent{3}{}\<[3]%
\>[3]{}\kw{handle}\;\{\mskip1.5mu {}\<[E]%
\\
\>[3]{}\hsindent{2}{}\<[5]%
\>[5]{}\op{abs}\;\Varid{x}\;\Varid{m}\;{}\<[18]%
\>[18]{}\Varid{r}\;\Varid{k}{}\<[23]%
\>[23]{}\mathbin{↦}\Varid{k}\;\Varid{r}\;\{\mskip1.5mu \Conid{Clo}\;\Varid{x}\;\hl{(\Varid{λ}\;\Varid{r'}\;\!.\!\;\op{abort$^\prime$}\;(\Varid{k}\;(\Varid{hLam}_{1}\;\Varid{r'}\;\Varid{m})))}\;\Varid{r}\mskip1.5mu\},{}\<[E]%
\\
\>[3]{}\hsindent{2}{}\<[5]%
\>[5]{}\op{app}\;\Varid{v}_{1}\;\Varid{v}_{2}\;{}\<[18]%
\>[18]{}\Varid{r}\;\Varid{k}{}\<[23]%
\>[23]{}\mathbin{↦}\Varid{k}\;\Varid{r}\;\{\mskip1.5mu {}\<[32]%
\>[32]{}\kw{match}\;\Varid{v}_{1}{}\<[E]%
\\
\>[32]{}\mid \Conid{Clo}\;\Varid{x}\;\Varid{m}\;\Varid{r'}\to \hl{\Varid{m}\;((\Varid{x},\Varid{v}_{2})\mathbin{::}\Varid{r'})}\mskip1.5mu\},{}\<[E]%
\\
\>[3]{}\hsindent{2}{}\<[5]%
\>[5]{}\op{var}\;\Varid{x}\;{}\<[18]%
\>[18]{}\Varid{r}\;\Varid{k}{}\<[23]%
\>[23]{}\mathbin{↦}\mathbf{let}\;\Varid{v}\mathrel{=}\Varid{lookup}\;\Varid{x}\;\Varid{r}\;\mathbf{in}\;\Varid{k}\;\Varid{r}\;\{\mskip1.5mu \Varid{v}\mskip1.5mu\},{}\<[E]%
\\
\>[3]{}\hsindent{2}{}\<[5]%
\>[5]{}\hl{\op{abort$^\prime$}\;\Varid{m}\;\Varid{k}\mathbin{↦}\Varid{m},}{}\<[E]%
\\
\>[3]{}\hsindent{2}{}\<[5]%
\>[5]{}\kw{return}\;\Varid{x}\;\Varid{r}\mathbin{↦}\{\mskip1.5mu \Varid{x}\mskip1.5mu\}{}\<[E]%
\\
\>[B]{}\hsindent{3}{}\<[3]%
\>[3]{}\mskip1.5mu\}\;\Varid{r}_{0}\;\Varid{prog}\kw{!}{}\<[E]%
\ColumnHook
\end{hscode}\resethooks

It is also possible to define a handler for a function abstraction effect where application operations have a suspended argument, such that we can do call-by-need evaluation.
The code artifact accompanying this paper contains such a handler.

The examples in this section have demonstrated the expressiveness of \ensuremath{\Varid{λ}^{\Varid{hop}}} on several examples.
In the remainder of the paper, we introduce a core calculus for \ensuremath{\Varid{λ}^{\Varid{hop}}}, and use that to define a separation of concerns property, and define criteria under which the separation of concerns property is guaranteed to hold.

\newtheorem*{case*}{\textbf{Case}}

\section{A Calculus with Handlers for Higher-Order Effects}
\label{sec:calculus}

In this section, we define the syntax and semantics of a calculus with row-typed
algebraic effects and handlers, called \ensuremath{\Varid{λ}^{\Varid{hop}}}, that supports \emph{higher-order
  operations}. The calculus is founded on the call-by-value \ensuremath{\Varid{λ}}-calculus, but
maintains a distinction between \emph{values} and \emph{computations} in the
spirit of Levy's \emph{call-by-push-value}~\cite{DBLP:books/sp/Levy2004}. Its
type-and-effect system (\cref{sec:type-system}) features
Hindley/Milner-style~\cite{DBLP:journals/jcss/Milner78} polymorphic abstraction
for types and effect rows, and draws inspiration from the core calculi defined
for Koka~\cite{DBLP:conf/popl/Leijen17} and
Frank~\cite{DBLP:journals/jfp/ConventLMM20}. We specify the operational
semantics of \ensuremath{\Varid{λ}^{\Varid{hop}}} using a reduction semantics
(\cref{sec:operational-semantics}), and we show that the type-and-effect system
is sound with respect to this semantics (\cref{sec:type-soundness}).

\subsection{Syntax}

\begin{figure}
\hrule%
\vspace{.5em}
\begin{align*}
\ensuremath{\Varid{ℓ}} & \ensuremath{\mathbin{∈}\Conid{String}} & \text{(Effect labels)} \\
  \ensuremath{\Varid{x},\Varid{y},\Varid{p},\Varid{m},\Varid{k},\Varid{α},\Varid{r}} & \ensuremath{\mathbin{∈}\Conid{String}} & \text{(Variables)} \\
   & & \\   
\ensuremath{\Varid{k}} &::= \ensuremath{\mathbin{★}\mathbin{∣}\Conid{R}\mathbin{∣}\Varid{k}\to \Varid{k}}  
& \text{(Kinds)} \\
\ensuremath{\Varid{τ}}  &::= \ensuremath{\Varid{α}\mathbin{∣}\Varid{τ}\;\Varid{τ}\mathbin{∣}\Varid{τ}\to \Varid{τ}\mathbin{∣}}\ann{ε \ensuremath{\mathbin{*}} ε}{\susp{τ}}\ensuremath{} 
& \text{(Types)} \\
\ensuremath{\Varid{ε}}  &::= \ensuremath{\mathbin{⟨⟩}\mathbin{∣}\mathbin{⟨}\Varid{ℓ},\Varid{ε}\mathbin{⟩}\mathbin{∣}\Varid{r}} & \text{(Effect rows)} \\
\ensuremath{\Varid{σ}}  &::= \ensuremath{\mathbin{∀}}\ \overline{\ensuremath{\Varid{α}}}\ \overline{r}\ .\ \ensuremath{\Varid{τ}} & \text{(Type Schemes}) \\ 
  \ensuremath{\Conid{Δ},\Conid{Γ}} &::= \ensuremath{\mathbin{∅}\mathbin{∣}\Conid{Γ},(\Varid{x}\mathrel{=}\Varid{v})} & \text{(Environments)} \\
 & & \\ 
\ensuremath{\Varid{e}} &::= \ensuremath{\Varid{e}\;\Varid{e}\mathbin{∣}\mathbf{let}\;\Varid{x}\mathrel{=}\Varid{e}\;\mathbf{in}\;\Varid{e}\mathbin{∣}\kw{handle}^{\ell}\;\{\mskip1.5mu \overline{C}\mskip1.5mu\}\;\Varid{e}\;\Varid{e}\mathbin{∣}\Varid{v}\mathbin{∣}\Varid{e}\kw{!}} 
  & \text{(Expressions)} \\
\ensuremath{\Varid{v}} &::= \ensuremath{\Varid{x}\mathbin{∣}\Varid{λx}}.\ensuremath{\Varid{e}\mathbin{∣}\Varid{c}\mathbin{∣}}\susp{e}
  & \text{(Values)} \\
\ensuremath{\Conid{C}} &::= \ensuremath{\op{op}^{\ell}\;\Varid{f}\;\overline{x}\;\overline{m}\;\Varid{p}\;\Varid{k}\mathbin{↦}\Varid{e}\mathbin{∣}\kw{return}\;\Varid{x}\;\Varid{p}\mathbin{↦}\Varid{e}} & \text{(Clauses)}   
\end{align*}
\hrule
\caption{Syntax of \ensuremath{\Varid{λ}^{\Varid{hop}}}}
\label{fig:syntax}
\end{figure}

\cref{fig:syntax} shows the syntax of \ensuremath{\Varid{λ}^{\Varid{hop}}}.

\paragraph{Notational conventions} Variables can range over expressions (\ensuremath{\Varid{x}},
\ensuremath{\Varid{y}}, \ensuremath{\Varid{p}}, \ensuremath{\Varid{m}}, \ensuremath{\Varid{k}}), types (\ensuremath{\Varid{α}}), and rows (\ensuremath{\Varid{r}}), potentially using subscripts
or primes to distinguish them. We use an overline (e.g., \ensuremath{\overline{\Varid{α}}}) to denote a
collection of zero or more objects, and use a subscript \ensuremath{\Varid{i}} to iterate over all
objects in such a collection. Operations are treated internally as variables,
but we distinguish them visually: \ensuremath{\op{op}^{\ell}\;\Varid{f}} denotes the operation \ensuremath{\Varid{f}} of the
effect \ensuremath{\Varid{ℓ}}. We write $e[\ensuremath{\Varid{x}}/\ensuremath{\Varid{v}}]$ for the capture-avoiding substitution of \ensuremath{\Varid{x}}
for \ensuremath{\Varid{v}} in \ensuremath{\Varid{e}}. Multiple substitutions are comma-separated, and we use overlines
again (e.g., $\overline{\ensuremath{\Varid{x}}/\ensuremath{\Varid{v}}}$) to indicate substitution of a collection of
variables.

\paragraph{Types and rows} As shown in \cref{fig:syntax}, types in \ensuremath{\Varid{λ}^{\Varid{hop}}} are
either \emph{variable}, \emph{type applications}, \emph{function types}, or
\emph{suspensions}. To keep the type syntax as simple as possible, it lacks an
introduction form for type constructors. Type constructors are really only
necessary for the types of handlers, which may decorate the resulting value.
For example, the exception handler decorates its result with a \ensuremath{\Conid{Maybe}}. We treat
type constructors like \ensuremath{\Conid{Maybe}} as type variables, which we assume are bound with
kind \ensuremath{\mathbin{...}\to \mathbin{★}} the context \ensuremath{\Conid{Δ}} (see \cref{fig:well-formedness}).

A suspension type indicates a thunked computations whose side effects have yet
to take place. We denote suspensions as $\ann{\ensuremath{\Varid{ε}\mathbin{*}\epsilon_{l}}}{\susp{\ensuremath{\Varid{τ}}}}$: a
computation that, when enacted, returns a value of type \ensuremath{\Varid{τ}} with side effects \ensuremath{\Varid{ε}\mathbin{*}\epsilon_{l}}. Both \ensuremath{\Varid{ε}} and \ensuremath{\epsilon_{l}} range over \emph{effect rows}, where \ensuremath{\Varid{ε}} indicates the
computation's \emph{immediate} effects, and \ensuremath{\epsilon_{l}} the computation's \emph{latent
  effects}.\footnote{We use the term ``latent effects'' to mean something
  different from \citet{DBLP:conf/aplas/BergSPW21}. We relate to their work in
  \cref{sec:related}.}  We discuss this distinction in more detail in
\cref{sec:immediate-latent}, but informally speaking the immediate effets are
the unhandled effects, while the latent effects have already been
handled. Effect rows themselves can be the empty row, a row extension with a
known effect label \ensuremath{\Varid{ℓ}}, or a row variable. We will occasionally abuse notation,
and denote rows using a list-like syntax (e.g., \ensuremath{\mathbin{⟨}\ell_{1},\ell_{2}\mathbin{⟩}}).

\paragraph{Type well-formedness} 

\begin{figure}
\hrule%
\vspace{.5em}\raggedleft\framebox{\ensuremath{\Conid{Δ}\vdash^{\star}\Varid{τ}\mathbin{:}\Varid{k}}}%
\begin{mathpar}
  \inferrule [WF-Var]
  {  \ensuremath{\Conid{Δ}\;(\Varid{α})\mathrel{=}\Varid{k}}  }
  {  \ensuremath{\Conid{Δ}\vdash^{\star}\Varid{α}\mathbin{:}\Varid{k}}  }
  \and
  \inferrule [WF-App]
  {  \ensuremath{\Conid{Δ}\vdash^{\star}\tau_{1}\mathbin{:}k_{\mathrm{1}}}
  \\ \ensuremath{\Conid{Δ}\vdash^{\star}\tau_{2}\mathbin{:}k_{\mathrm{1}}\to k_{\mathrm{2}}} }
  {  \ensuremath{\Conid{Δ}\vdash^{\star}\tau_{1}\;\tau_{2}\mathbin{:}k_{\mathrm{2}}}  }
  \and
  \inferrule [WF-Arrow]
  {  \ensuremath{\Conid{Δ}\vdash^{\star}\tau_{1}\mathbin{:}\mathbin{★}}
  \\ \ensuremath{\Conid{Δ}\vdash^{\star}\tau_{2}\mathbin{:}\mathbin{★}}
  \\ \ensuremath{\Conid{Δ}\vdash^{\star}\Varid{ε}\mathbin{:}\Conid{R}}
  \\ \ensuremath{\Conid{Δ}\vdash^{\star}\epsilon_{l}\mathbin{:}\Conid{R}}}
  {  \ensuremath{\Conid{Δ}\vdash^{\star}\tau_{1}\to \tau_{2}\mathbin{:}\mathbin{★}}}
  \and
  \inferrule [WF-Suspend]
  {  \ensuremath{\Conid{Δ}\vdash^{\star}\Varid{τ}\mathbin{:}\mathbin{★}}
  \\ \ensuremath{\Conid{Δ}\vdash^{\star}\Varid{ε}\mathbin{:}\Conid{R}}
  \\ \ensuremath{\Conid{Δ}\vdash^{\star}\epsilon_{l}\mathbin{:}\Conid{R}} }
  {  \ensuremath{\Conid{Δ}\vdash^{\star}}\ann{\ensuremath{\Varid{ε}\mathbin{*}\epsilon_{l}}}{\susp{τ}}\ensuremath{\mathbin{:}\mathbin{★}} }
  \and
  \inferrule [WF-Empty]
  { }
  {\ensuremath{\Conid{Δ}\vdash^{\star}\mathbin{⟨⟩}\mathbin{:}\Conid{R}}}
  \and
  \inferrule [WF-Cons]
  {  \ensuremath{\Varid{ℓ}\mathbin{∈}\Conid{Σ}} 
  \\ \ensuremath{\Conid{Δ}\vdash^{\star}\Varid{ε}\mathbin{:}\Conid{R}} }
  {  \ensuremath{\Conid{Δ}\vdash^{\star}\mathbin{⟨}\Varid{ℓ},\Varid{ε}\mathbin{⟩}\mathbin{:}\Conid{R}}}
\end{mathpar}
\vspace{.5em}%
\hrule
\caption{Well-formedness rules for types}
\label{fig:well-formedness}
\end{figure}

We use \emph{kinds} to ensure that types are well-formed. The judgment \ensuremath{\Conid{Δ}\vdash^{\star}\Varid{τ}\mathbin{:}\Varid{k}} states that the type \ensuremath{\Varid{τ}} is well-formed with kind \ensuremath{\Varid{k}} under context \ensuremath{\Conid{Δ}},
where \ensuremath{\Conid{Δ}} contains mappings for both type and row variables. We present the
corresponding rules in \cref{fig:well-formedness}. The rules for type variables
(\textsc{WF-Var}), type application (\textsc{WF-App}), and function types
(\textsc{WF-Arrow}) are entirely standard. The rule \textsc{WF-Suspend} asserts
well-formedness of a suspension type, and features additional requirements that
its effect annotations \ensuremath{\Varid{ε}} and \ensuremath{\epsilon_{l}} are well-formed.

Well-formedness for effect rows is defined with the rules \textsc{WF-Empty} and
\textsc{WF-Cons}, declaring respectively that the empty row is well formed, and
that a non-empty row is well-formed if its head is a known effect and its tail
is well-formed. To assert that a label \ensuremath{\Varid{ℓ}} is a known effect, we assume a global
context \ensuremath{\Conid{Σ}} containing all declared effects, which maps effect labels to the set
of operations of that effect.

\paragraph{Expressions and values} The term syntax of \ensuremath{\Varid{λ}^{\Varid{hop}}}, which
distinguishes expressions from values, takes its basic constructs from the
\ensuremath{\Varid{λ}}-calculus; i.e., abstraction, application, variables, let-bindings, and
constants. Additionally, \ensuremath{\Varid{λ}^{\Varid{hop}}} includes an introduction/elimination pair for
suspension types: $\susp{\ensuremath{\Varid{e}}}$ suspends the side effects of \ensuremath{\Varid{e}}, while \ensuremath{\Varid{e}\kw{!}}
enacts any suspended side effects.

The \ensuremath{\kw{handle}^{\ell}} construct installs a handler for the effect \ensuremath{\Varid{ℓ}}, and has three
arguments: a list of clauses (\ensuremath{\overline{\Conid{C}}}) and two expressions. The first expression
argument is the \emph{handler parameter}, which is distributed to all handler
clauses to allow us to define handlers that interpret into a function space
(e.g., a state effect which takes the initial state as input). The second
expression argument is the expression in which we handle the operations of
\ensuremath{\Varid{ℓ}}. The list of clauses, \ensuremath{\overline{\Conid{C}}}, defines the semantics of these operations,
and should contain exactly one operation clause for each operation of \ensuremath{\Varid{ℓ}},
together with a single \ensuremath{\kw{return}} clause.

\paragraph{Row equivalence and concatenation}

\begin{figure}
  \hrule
\vspace{.5em}\raggedleft\framebox{\ensuremath{\epsilon_{1}\simeq\epsilon_{2}}}%
\begin{mathpar}
  \inferrule [Eq-Refl]
  { }
  {  \ensuremath{\Varid{ε}\simeq\Varid{ε}}}
  \and
  \inferrule [Eq-Trans]
  {  \ensuremath{\epsilon_{1}\simeq\epsilon_{2}} \\ \ensuremath{\epsilon_{2}\simeq\epsilon_{3}} }
  {  \ensuremath{\epsilon_{1}\simeq\epsilon_{3}} }
  \and 
  \inferrule [Eq-Head]
  {  \ensuremath{\epsilon_{1}\simeq\epsilon_{2}} }
  {  \ensuremath{\mathbin{⟨}\Varid{ℓ},\epsilon_{1}\mathbin{⟩}\simeq\mathbin{⟨}\Varid{ℓ},\epsilon_{2}\mathbin{⟩}} }
  \and
  \inferrule [Eq-Swap]
  {  \ensuremath{\ell_{1}\neq\ell_{2}} }
  { \ensuremath{\mathbin{⟨}\ell_{1},\mathbin{⟨}\ell_{2},\Varid{ε}\mathbin{⟩⟩}\simeq\mathbin{⟨}\ell_{2},\mathbin{⟨}\ell_{1},\Varid{ε}\mathbin{⟩⟩}} }
\end{mathpar}
\raggedleft{\vspace{.5em}\framebox{\ensuremath{\epsilon_{1}\mathbin{⊕}\epsilon_{2}\simeq\epsilon_{3}}}}%
\begin{mathpar}
  \inferrule [Concat-Nil]
  { }
  { \ensuremath{\mathbin{⟨⟩}\mathbin{⊕}\Varid{ε}\simeq\Varid{ε}} }
  \and
  \inferrule [Concat-Cons]
  { \ensuremath{\epsilon_{1}\mathbin{⊕}\epsilon_{2}\simeq\epsilon_{3}}}
  { \ensuremath{\mathbin{⟨}\Varid{ℓ},\epsilon_{1}\mathbin{⟩}\mathbin{⊕}\epsilon_{2}\simeq\mathbin{⟨}\Varid{ℓ},\epsilon_{3}\mathbin{⟩}} }
\end{mathpar}
\hrule
\caption{Effect row equivalence and concatenation}
\label{fig:row-equivalence}
\end{figure}

In the definition of \ensuremath{\Varid{λ}^{\Varid{hop}}}'s type-and-effect system, we will work modulo a
notion of \emph{row equivalence}. That is, whenever two terms are annotated with
equal rows, we only assume that these rows are equivalent up to
reordering. \cref{fig:row-equivalence} defines this equivalence relation, \ensuremath{\simeq},
for effect rows. Furthermore, we will occasionally need to refer to the
\emph{concatenation} of two rows. Formally, we define (in
\cref{fig:row-equivalence}) row concatenation as a ternary relation, where a
proof of the form \ensuremath{\epsilon_{1}\mathbin{⊕}\epsilon_{2}\simeq\epsilon_{3}} witnesses that the concatenation of \ensuremath{\epsilon_{1}} and
\ensuremath{\epsilon_{2}} is equal (up to reordering) to \ensuremath{\epsilon_{3}}. In the typing rules
(\cref{fig:typing-rules}), however, we leave row reorderings implicit, and write
\ensuremath{\epsilon_{1}\mathbin{⊕}\epsilon_{2}} instead.

\subsection{Type-and-Effect System}
\label{sec:type-system}

The type-and-effect system is defined using a single judgment, for which we
present the rules in \cref{fig:typing-rules}. A judgment of the form \ensuremath{\Conid{Γ}\mathbin{⊢}\Varid{e}\mathbin{:}\Varid{τ}\mathbin{∣}\Varid{ε}\mathbin{*}\epsilon_{l}} asserts that the expression \ensuremath{\Varid{e}} has type \ensuremath{\Varid{τ}} under context \ensuremath{\Conid{Γ}}, and
side effects \ensuremath{\Varid{ε}} (immediate) and \ensuremath{\epsilon_{l}} (latent). As we show in
\cref{sec:type-soundness}, the effect annotations are a sound over-approximation
of an expression's side effects: the effects in the annotation may or may not
happen, but any side-effect that does happen is included in the annotation.

Variables do not have any side effects, as reflected in the rule \textsc{T-Var},
which asserts that variable expressions are typeable under any effect
annotation. Similarly, \ensuremath{\Varid{λ}}-abstractions also do not produce any side effects,
and hence the rule \textsc{T-Abs} permits any effect annotation. The body of a
\ensuremath{\Varid{λ}}-abstraction, however, must by typeable with the empty set of effects, or in
other words: it must be \emph{pure}. All functions in \ensuremath{\Varid{λ}^{\Varid{hop}}} are pure by
definition, and effectful functions can only be encoded by returning an
effect-annotated suspension. This also clarifies the \textsc{T-App} rule, where
only the effect-annotations of the function and argument expression inform the
effects of the resulting expression, but not the function type itself. We treat
the components of a let-expressions as pure expressions in a similar
fashion. The \textsc{T-Gen} and \textsc{T-Inst} rules are the familiar rules for
generalization and instantiation of the polymorphic \ensuremath{\Varid{λ}}-calculus.

\paragraph{Effect resurfacing}
While annotations distinguish between an expression's immediate and latent
effects, latent effects may \emph{resurface} and become immediate again with the
\textsc{T-Resurface} rule. The intuition behind resurfacing is that it allows us
to re-annotate an expression with a less precise approximation of its side
effects, in the sense that we lose the information that the effect \ensuremath{\Varid{ℓ}} occurs
exclusively under at least one layer of suspension. Resurfacing is necessary to
type handlers that recurse on sub-computations, where the handled effect should
be discharged to the surrounding context instead of becoming latent. We discuss
the need for resurfacing in more detail in \cref{sec:immediate-latent}.

\paragraph{Suspension/Enactment}
\textsc{T-Suspend} and \textsc{T-Enact} define respectively introduction and
elimination rules for suspension types. A suspended computation has no effects,
so we may annotate it however is most convenient: the effects of a suspension are
completely unrelated to its context. The opposite is true for enactment, where
it is required that the effects of the enacted computation exactly match the
effects provided by the surrounding context.

\begin{figure} 
\hrule
\vspace{.5em}\raggedleft\framebox{\ensuremath{\Conid{Γ}\mathbin{⊢}\Varid{e}\mathbin{:}\Varid{σ}\mathbin{∣}\Varid{ε}\mathbin{*}\epsilon_{l}}}%
\begin{mathpar}
  \inferrule [T-Var]
  {  \ensuremath{\Conid{Γ}\;(\Varid{x})\mathrel{=}\Varid{σ}}  }
  {  \ensuremath{\Conid{Γ}\mathbin{⊢}\Varid{x}\mathbin{:}\Varid{σ}\mathbin{∣}\Varid{ε}\mathbin{*}\epsilon_{l}}  }
  \and
  \inferrule [T-Abs]
  {  \ensuremath{\Conid{Γ},\Varid{x}\mathbin{:}\tau_{1}\mathbin{⊢}\Varid{e}\mathbin{:}\tau_{2}\mathbin{∣}\mathbin{⟨⟩}\mathbin{*}\mathbin{⟨⟩}}}
  {  \ensuremath{\Conid{Γ}\mathbin{⊢}\Varid{λ}\;\Varid{x}}.\ensuremath{\Varid{e}\mathbin{:}\tau_{1}\to \tau_{2}\mathbin{∣}\Varid{ε}\mathbin{*}\epsilon_{l}} }
  \and
  \inferrule [T-App]
  {  \ensuremath{\Conid{Γ}\mathbin{⊢}e_{1}\mathbin{:}\tau_{1}\to \tau_{2}\mathbin{∣}\Varid{ε}\mathbin{*}\epsilon_{l}}
  \\ \ensuremath{\Conid{Γ}\mathbin{⊢}e_{2}\mathbin{:}\tau_{2}\mathbin{∣}\Varid{ε}\mathbin{*}\epsilon_{l}}  }
  {  \ensuremath{\Conid{Γ}\mathbin{⊢}e_{1}\;e_{2}\mathbin{:}\tau_{2}\mathbin{∣}\Varid{ε}\mathbin{*}\epsilon_{l}} }
  \and
  \inferrule [T-Let]
  {  \ensuremath{\Conid{Γ}\mathbin{⊢}e_{1}\mathbin{:}\tau_{1}\mathbin{∣}\Varid{ε}\mathbin{*}\epsilon_{l}}
  \\ \ensuremath{\Conid{Γ},\Varid{x}\mathrel{=}\tau_{1}\mathbin{⊢}e_{2}\mathbin{:}\tau_{2}\mathbin{∣}\Varid{ε}\mathbin{*}\epsilon_{l}} }
  {  \ensuremath{\Conid{Γ}\mathbin{⊢}\mathbf{let}\;\Varid{x}\mathrel{=}e_{1}\;\mathbf{in}\;e_{2}\mathbin{:}\tau_{2}\mathbin{∣}\Varid{ε}\mathbin{*}\epsilon_{l}} }
  \\
  \inferrule [T-Gen]
  {  \ensuremath{\Conid{Γ}\mathbin{⊢}\Varid{e}\mathbin{:}\Varid{τ}\mathbin{∣}\Varid{ε}\mathbin{*}\epsilon_{l}}
  \\ \ensuremath{\overline{\Varid{α}},\overline{\Varid{r}}}\ \textsf{not free in \ensuremath{\Conid{Γ}}}}  
  {  \ensuremath{\Conid{Γ}\mathbin{⊢}\Varid{e}\mathbin{:}\mathbin{∀}\overline{\Varid{α}}\;\overline{\Varid{r}}}.\ensuremath{\Varid{τ}\mathbin{∣}\Varid{ε}\mathbin{*}\epsilon_{l}} }
  \and
  \inferrule [T-Inst]
  {  \ensuremath{\Conid{Γ}\mathbin{⊢}\Varid{e}\mathbin{:}\mathbin{∀}\overline{\Varid{α}}\;\overline{\Varid{r}}}.\ensuremath{\Varid{τ}\mathbin{∣}\Varid{ε}^{\prime}\mathbin{*}\epsilon_{l}^{\prime}} }
  {  \ensuremath{\Conid{Γ}\mathbin{⊢}\Varid{e}\mathbin{:}\Varid{τ}}[\overline{\ensuremath{\Varid{α}}/\ensuremath{\Varid{τ}^{\prime}}},\overline{\ensuremath{\Varid{r}}/\ensuremath{\Varid{ε}}}]\ensuremath{\mathbin{∣}\Varid{ε}^{\prime}\mathbin{*}\epsilon_{l}^{\prime}} }
  \and
  \inferrule [T-Resurface]
  {  \ensuremath{\Conid{Γ}\mathbin{⊢}\Varid{e}\mathbin{:}\Varid{τ}\mathbin{∣}\Varid{ε}\mathbin{*}\mathbin{⟨}\Varid{ℓ},\epsilon_{l}\mathbin{⟩}}}
  {  \ensuremath{\Conid{Γ}\mathbin{⊢}\Varid{e}\mathbin{:}\Varid{τ}\mathbin{∣}\mathbin{⟨}\Varid{ℓ},\Varid{ε}\mathbin{⟩}\mathbin{*}\epsilon_{l}}}
  \and
  \inferrule [T-Suspend]
  {  \ensuremath{\Conid{Γ}\mathbin{⊢}\Varid{e}\mathbin{:}\Varid{τ}\mathbin{∣}\Varid{ε}\mathbin{*}\epsilon_{l}} }
  {  \ensuremath{\Conid{Γ}\mathbin{⊢}}\susp{e}\ensuremath{\mathbin{:}}\ann{\ensuremath{\Varid{ε}\mathbin{*}\epsilon_{l}}}{\susp{τ}}\ensuremath{\mathbin{∣}\Varid{ε}^{\prime}\mathbin{*}\epsilon_{l}^{\prime}} }
  \and
  \inferrule [T-Enact]
  {  \ensuremath{\Conid{Γ}\mathbin{⊢}\Varid{e}\mathbin{:}}\ann{\ensuremath{\Varid{ε}\mathbin{*}\epsilon_{l}}}{\susp{τ}}\ensuremath{\mathbin{∣}\Varid{ε}\mathbin{*}\epsilon_{l}}}
  {  \ensuremath{\Conid{Γ}\mathbin{⊢}\Varid{e}\kw{!}\mathbin{:}\Varid{τ}\mathbin{∣}\Varid{ε}\mathbin{*}\epsilon_{l}} }
  \and 
  \inferrule [T-Handle]
  {  \ensuremath{\Conid{Γ}^{\prime}\triangleq\Conid{Γ},}\overline{\ensuremath{x_{i}\mathbin{:}\tau_{i}}[\ensuremath{\Varid{r}}/\ensuremath{\Varid{ε}},\ensuremath{r_{l}}/\ensuremath{\epsilon_{l}}]}\ensuremath{,(\Varid{p}\mathbin{:}\tau_{p}),(\Varid{k}\mathbin{:}\tau_{p}\to }\ann{\ensuremath{\mathbin{⟨}\Varid{ℓ},\Varid{ε}\mathbin{⊕}\epsilon_{l}\mathbin{⟩}\mathbin{*}\mathbin{⟨⟩}}}{\susp{\ensuremath{\tau_{1}}}}\ \ensuremath{\to }\ann{\ensuremath{\Varid{ε}\mathbin{*}\mathbin{⟨}\Varid{ℓ},\epsilon_{l}\mathbin{⟩}}}{\susp{\ensuremath{\tau_{2}}}}\ensuremath{)}
  \\ \ensuremath{\Conid{Γ}\mathbin{⊢}\op{op}^{ℓ}_{i}\mathbin{:}\mathbin{∀}\Varid{r}\;r_{l}}\ .\ \ensuremath{\overline{\tau_{i}}\to }\ann{\ensuremath{\mathbin{⟨}\Varid{ℓ},\Varid{r}\mathbin{⊕}r_{l}\mathbin{⟩}\mathbin{*}\mathbin{⟨⟩}}}{\susp{\ensuremath{\tau_{i}}}}\ \ensuremath{\mathbin{∣}\Varid{ε}^{\prime}\mathbin{*}\epsilon_{l}^{\prime}} 
  \\ \ensuremath{\Conid{Γ},\Varid{x}\mathbin{:}\tau_{1},\Varid{p}\mathbin{:}\tau_{p}\mathbin{⊢}e_{r}\mathbin{:}}\ann{\ensuremath{\Varid{ε}\mathbin{*}\mathbin{⟨}\Varid{ℓ},\epsilon_{l}\mathbin{⟩}}}{\susp{\ensuremath{\tau_{2}}}}\ \ensuremath{\mathbin{∣}\Varid{ε}^{\prime}\mathbin{*}\epsilon_{l}^{\prime}}
  \\ \ensuremath{\Conid{Γ}^{\prime}\mathbin{⊢}e_{i}\mathbin{:}}\ann{\ensuremath{\Varid{ε}\mathbin{*}\mathbin{⟨}\Varid{ℓ},\epsilon_{l}\mathbin{⟩}}}{\susp{\ensuremath{\tau_{2}}}}\ \ensuremath{\mathbin{∣}\Varid{ε}^{\prime}\mathbin{*}\epsilon_{l}^{\prime}}
  \\ \ensuremath{\Conid{Γ}\mathbin{⊢}e_{p}\mathbin{:}\tau_{p}\mathbin{∣}\Varid{ε}^{\prime}\mathbin{*}\epsilon_{l}^{\prime}}
  \\ \ensuremath{\Conid{Γ}\mathbin{⊢}\Varid{e}\mathbin{:}\tau_{1}\mathbin{∣}\mathbin{⟨}\Varid{ℓ},\Varid{ε}\mathbin{⟩}\mathbin{*}\epsilon_{l}}
  \\ \overline{\ensuremath{\op{op}^{ℓ}_{i}\;\Varid{f}}}\ \ensuremath{\mathrel{=}\Conid{Σ}}(ℓ) }
  {  \ensuremath{\Conid{Γ}\mathbin{⊢}\kw{handle}^{\ell}\;\{\mskip1.5mu }\overline{\ensuremath{\op{op}^{\ell}\;\Varid{f}\;\overline{x}\;\overline{m}\;\Varid{p}\;\Varid{k}\mathbin{↦}\Varid{e}}}\ \ensuremath{;\kw{return}\;\Varid{x}\;\Varid{p}\mathbin{↦}e_{r}\mskip1.5mu\}\;e_{p}\;\Varid{e}\mathbin{:}}\ann{\ensuremath{\Varid{ε}\mathbin{*}\mathbin{⟨}\Varid{ℓ},\epsilon_{l}\mathbin{⟩}}}{\susp{\ensuremath{\tau_{2}}}}\ensuremath{\mathbin{∣}\Varid{ε}^{\prime}\mathbin{*}\epsilon_{l}^{\prime}} }
\end{mathpar}
\hrule
\caption{Typing rules}
\label{fig:typing-rules}
\end{figure}

\paragraph{Operations}
Before considering how to type effect handlers, we must say a few words about
the type of operations. While operations are typed using \textsc{T-Var}, we do
make some assumptions about their type. In general, we say that the operations
of an effect \ensuremath{\Varid{ℓ}} return a computation with the following type:
$\ann{\ensuremath{\mathbin{⟨}\Varid{ℓ},\Varid{r}\mathbin{⊕}r_{l}\mathbin{⟩}\mathbin{*}\mathbin{⟨⟩}}}{\susp{\ensuremath{\Varid{τ}}}}$, where \ensuremath{\Varid{τ}} is the return type of the
operation, and \ensuremath{\Varid{r}} and \ensuremath{r_{l}} are universally quantified row variables that are
instantiated at the call-site. Effectively, this means that the operations of
the effect \ensuremath{\Varid{ℓ}} can be called in any context where the immediate effects contain
the label \ensuremath{\Varid{ℓ}}. For instance, consider the type of the \ensuremath{\op{throw}} operation:
\begin{equation*}
\ensuremath{\op{throw}\mathbin{:}\mathbin{∀}\Varid{α}\;\Varid{r}\;r_{l}}\ .\ \ann{\ensuremath{\mathbin{⟨}\Varid{exc},\Varid{r}\mathbin{⊕}r_{l}\mathbin{⟩}\mathbin{*}\mathbin{⟨⟩}}}{\susp{\ensuremath{\Varid{α}}}}
\end{equation*}

\noindent
We assume that the type of every defined operation universally quantifies over
\ensuremath{\Varid{r}} and \ensuremath{r_{l}}, and that these denote respectively the immediate and latent
effects of the context of its handler. This implies that between an operation
and its handler, there must be handlers for all latent effects. This restricts
the usage of operations somewhat, as we discuss in \cref{sec:immediate-latent}.

Operations with one or more arguments are assigned a function type, where the
type of the arguments is determined by the operation's definition. Formally, we
do not distinguish sub-computations from ``regular'' parameters, but we reserve
the terminology for parameters that are suspended computations. To illustrate,
consider the \ensuremath{\op{put}} operation, which has one regular argument, and the \ensuremath{\op{catch}}
operation, which takes two sub-computations:
\begin{gather*}
\ensuremath{\op{put}\mathbin{:}\mathbin{∀}\Varid{r}\;r_{l}}\ .\ \ensuremath{\Varid{s}\to }\ann{\ensuremath{\mathbin{⟨}\Varid{state},\Varid{r}\mathbin{⊕}r_{l}\mathbin{⟩}\mathbin{*}\mathbin{⟨⟩}}}{\susp{\ensuremath{()}}} \\
\ensuremath{\op{catch}\mathbin{:}\mathbin{∀}\Varid{α}\;\Varid{r}\;r_{l}}\ .\ \ann{\ensuremath{\mathbin{⟨}\Varid{exc},\Varid{r}\mathbin{⊕}r_{l}\mathbin{⟩}\mathbin{*}\mathbin{⟨⟩}}}{\susp{\ensuremath{\Varid{α}}}}\ \ensuremath{\to }\ann{\ensuremath{\mathbin{⟨}\Varid{exc},\Varid{r}\mathbin{⊕}r_{l}\mathbin{⟩}\mathbin{*}\mathbin{⟨⟩}}}{\susp{\ensuremath{\Varid{α}}}}\ensuremath{\to } \ann{\ensuremath{\mathbin{⟨}\Varid{exc},\Varid{r}\mathbin{⊕}r_{l}\mathbin{⟩}\mathbin{*}\mathbin{⟨⟩}}}{\susp{\ensuremath{\Varid{α}}}} 
\end{gather*}

\paragraph{Typing handlers}
Effect handlers are typed with \textsc{T-Handle}. This rule is quite a handful, so
let us unpack its definition. The conclusion states that handling the effect \ensuremath{\Varid{ℓ}}
produces a suspension, annotated with effects \ensuremath{\Varid{ε}\mathbin{*}\mathbin{⟨}\Varid{ℓ},\epsilon_{l}\mathbin{⟩}}, where \ensuremath{\Varid{ε}} and
\ensuremath{\epsilon_{l}} depend on the context in which the handler is invoked. Similar to
\textsc{T-Suspend}, we say that the side effects of producing this suspension
are completely independent from the side effects of the returned
computation. Furthermore, in the conclusion we syntactically impose the
restriction that the clause list ends with a single return clause, and we assume
that the collection of operation clauses, \ensuremath{\overline{(\op{op}^{\ell}\;\Varid{f})}}, matches the operations
defined for \ensuremath{\Varid{ℓ}} in \ensuremath{\Conid{Σ}}.

In the premise of \textsc{T-Handle}, we first define the context \ensuremath{\Conid{Γ}^{\prime}} under
which operation clauses will be typed. \ensuremath{\Conid{Γ}^{\prime}} extends \ensuremath{\Conid{Γ}} with bindings for the
operation's arguments (\ensuremath{\overline{x_{i}}}), the handler parameter (\ensuremath{\Varid{p}}), and the
continuation (\ensuremath{\Varid{k}}). The continuation returns a suspension that matches the type
of the entire expression --- i.e.  $\ann{\ensuremath{\Varid{ε}\mathbin{*}\mathbin{⟨}\Varid{ℓ},\epsilon_{l}\mathbin{⟩}}}{\susp{\ensuremath{\tau_{2}}}}$ --- and
takes two arguments: the handler parameter and a computation argument of type
$\ann{\ensuremath{\mathbin{⟨}\Varid{ℓ},\Varid{ε}\mathbin{⊕}\epsilon_{l}\mathbin{⟩}\mathbin{*}\mathbin{⟨⟩}}}{\susp{\ensuremath{\tau_{2}}}}$. At run-time, this computation
argument will be evaluated under the same evaluation context as the handled
operation, for which the type-and-effect system statically guarantees that it
contains handlers for all effects in \ensuremath{\epsilon_{l}}. We substitute \ensuremath{\Varid{r}} and \ensuremath{r_{l}} in the
types of \ensuremath{\overline{x_{i}}} precisely such that sub-computations will match the type of
this computation argument, providing an ``escape hatch'' for the effects in
\ensuremath{\epsilon_{l}}, which are immediate in sub-computations, but latent in the handler's
context. These effects are discharged to the handlers in the evaluation context
of the handled operation. Indeed, the operation clause for \ensuremath{\op{catch}} in
\cref{sec:ho-effects} enacts its sub-computations as part of \ensuremath{\Varid{k}}'s computational
argument. We discuss the operational semantics of \ensuremath{\Varid{λ}^{\Varid{hop}}} in more detail in
\cref{sec:operational-semantics}.

Furthermore, \textsc{T-Handle} asserts that the implementation of the \ensuremath{\kw{return}}
clause, \ensuremath{e_{r}}, can be typed under context \ensuremath{\Conid{Γ},\Varid{x}\mathbin{:}\tau_{1},\Varid{p}\mathbin{:}\tau_{p}}, and that the
implementations of all operation clauses, \ensuremath{e_{i}}, can be typed under the
aforementioned context \ensuremath{\Conid{Γ}^{\prime}}. Both \ensuremath{e_{r}} and \ensuremath{e_{i}} should have the same type as the
handle expression as a whole. Finally, we also assert that the parameter
expression, \ensuremath{e_{p}}, and the handled expression, \ensuremath{\Varid{e}}, are well-typed, where the
effect annotation of \ensuremath{e_{p}} matches the entire expression, and \ensuremath{\Varid{e}} that of the
returned suspension, with the label \ensuremath{\Varid{ℓ}} shifted to the latent side.

\subsection{On Immediate and Latent Effects}
\label{sec:immediate-latent}

As described above, a suspension type \ensuremath{\ann{\Varid{ε}\mathbin{*}\epsilon_{l}}{\susp{\Varid{τ}}}} distinguishes a row of \emph{immediate effects} (\ensuremath{\Varid{ε}}) and a row of \emph{latent effects} (\ensuremath{\epsilon_{l}}).
The immediate effects describe the set of operations for which the suspension has no handler installed, whereas the latent effects describe the set of operations for which there is a handler installed inside the suspension.
We briefly discuss the need for tracking both immediate and latent effects, and how this affects what programs can and cannot be typed using our type-and-effect system. 

\paragraph{The need to track immediate and latent effects}
It is necessary to track latent effects because handlers do not eagerly propagate into suspensions.
For example, consider the following expression, where \ensuremath{\op{op}^{\ell}} is some higher-order operation:
\begin{equation*}
\ensuremath{\kw{handle}}^{\ensuremath{\Conid{State}}}\ \ensuremath{\{\mskip1.5mu }\ldots\ensuremath{\mskip1.5mu\}\;(\op{op}^{\ell}\;\{\mskip1.5mu (\op{put}\;\mathrm{0})\kw{!}\mskip1.5mu\})} 
\end{equation*}

\noindent
This expression has a suspension type: \ensuremath{\ann{\Varid{ε}\mathbin{*}\mathbin{⟨}\Conid{State},\epsilon_{l}\mathbin{⟩}}{\susp{\Varid{τ}}}} for some τ and some ε, where ε contains at least the effect label of the \ensuremath{\op{op}^{\ell}} operation.
In conventional effect handler type systems, applying a handler such as $\ensuremath{\kw{handle}}^{State}$ would remove the \ensuremath{\Conid{State}} effect from the result type.
The reason \ensuremath{\Conid{State}} features as a latent effect in the type fo the expression above is that the \ensuremath{\op{put}} operation occurs in a sub-computation of \ensuremath{\op{op}^{\ell}}.
In \ensuremath{\Varid{λ}^{\Varid{hop}}}, the state handler is \emph{not} automatically applied to sub-computations, so if a handler of the \ensuremath{\op{op}^{\ell}} operation enacts its sub-computation, the \ensuremath{\Conid{State}} effect becomes immediate again.
That is why our type system tracks latent effects: simply removing an effect from the annotation of suspensions when applying a handler loses information about the immediate effects of sub-computations.
To avoid this information loss, handlers ``shift'' a label from the immediate to the latent annotation (see \textsc{T-Handle} in \cref{fig:typing-rules}).

\paragraph{Effect resurfacing}
The relocation of annotations from the immediate to the latent annotation by
handlers has subtle implications for \emph{recursive handlers}; i.e., a handler
that invokes itself. The handler \ensuremath{\Varid{hCatch}} defined \cref{sec:ho-effects} does
exactly this: the handler clause for \ensuremath{\op{catch}} is implemented by recursively
invoking \ensuremath{\Varid{hCatch}}, and inspecting the result:
\begin{hscode}\SaveRestoreHook
\column{B}{@{}>{\hspre}l<{\hspost}@{}}%
\column{3}{@{}>{\hspre}l<{\hspost}@{}}%
\column{24}{@{}>{\hspre}l<{\hspost}@{}}%
\column{E}{@{}>{\hspre}l<{\hspost}@{}}%
\>[3]{}\op{catch}\;\Varid{m}_{1}\;\Varid{m}_{2}\;\Varid{k}\mathbin{↦}\Varid{k}\;\{\mskip1.5mu {}\<[24]%
\>[24]{}\kw{match}\;(\Varid{hCatch}\;\Varid{m}_{1})\kw{!}\mid \Conid{Nothing}\to \Varid{m}_{2}\kw{!}\mid \Conid{Just}\;\Varid{x}\to \Varid{x}\mskip1.5mu\}{}\<[E]%
\ColumnHook
\end{hscode}\resethooks

\noindent
Following \textsc{T-Handle}, this clause is part of a \ensuremath{\kw{handle}} expression annotated with \ensuremath{\Varid{ε}\mathbin{*}\mathbin{⟨}\Conid{Catch},\epsilon_{l}\mathbin{⟩}}. 
Consequently, the continuation argument is annotated with \ensuremath{\mathbin{⟨}\Conid{Catch},\Varid{ε}\mathbin{⊕}\epsilon_{l}\mathbin{⟩}\mathbin{*}\mathbin{⟨⟩}}.
The recursive invocation \ensuremath{(\Varid{hCatch}\;\Varid{m}_{1})\kw{!}}, however, has effects \ensuremath{\Varid{ε}\mathbin{⊕}\epsilon_{l}\mathbin{*}\mathbin{⟨}\Conid{Catch}\mathbin{⟩}}, which is incompatible with \ensuremath{\mathbin{⟨}\Conid{Catch},\Varid{ε}\mathbin{⊕}\epsilon_{l}\mathbin{⟩}\mathbin{*}\mathbin{⟨⟩}}.
The \textsc{T-Resurface} rule mediates \ensuremath{\Conid{Catch}} from a latent effect into an immediate effect to make the recursive call typeable.
This mediation is essentially a form of up-casting. 

\paragraph{Operation types restrict use of inline handlers}
As discussed in \cref{sec:type-system}, our type-and-effect system types operations of effect \ensuremath{\Varid{ℓ}} as having a type $\ann{\ensuremath{\mathbin{⟨}\Varid{ℓ},\Varid{r}\mathbin{⊕}r_{l}\mathbin{⟩}\mathbin{*}\mathbin{⟨⟩}}}{\susp{\ensuremath{\Varid{τ}}}}$ where \ensuremath{\Varid{r}} and \ensuremath{r_{l}} are polymorphically-quantified row variables.
As a result, operations are only callable from a context with an empty latent effect row.
Without this restriction, we would be able to type a suspension to make it appear as if we have pre-applied a handler to a sub-computation.
For example, we could type the expression \ensuremath{\op{catch}\;\{\mskip1.5mu (\op{put}\;\mathrm{0})\mathbin{!}\mskip1.5mu\}\;\{\mskip1.5mu ()\mskip1.5mu\}} as \ensuremath{\ann{\mathbin{⟨}\Conid{Catch}\mathbin{⟩}\mathbin{*}\mathbin{⟨}\Conid{State}\mathbin{⟩}}{\susp{()}}}.
However, this type would suggest that the following program is safe, which it is not:
\begin{hscode}\SaveRestoreHook
\column{B}{@{}>{\hspre}l<{\hspost}@{}}%
\column{E}{@{}>{\hspre}l<{\hspost}@{}}%
\>[B]{}(\Varid{hCatch}\;(\op{catch}\;\{\mskip1.5mu (\op{put}\;\mathrm{0})\kw{!}\mskip1.5mu\}\;\{\mskip1.5mu ()\mskip1.5mu\}))\kw{!}{}\<[E]%
\ColumnHook
\end{hscode}\resethooks
This program gets stuck because the \ensuremath{\op{put}\;\mathrm{0}} operation in the first branch of \ensuremath{\op{catch}} is unhandled.

This restrictive typing has implications beyond ruling out programs that go wrong.
It also rules out programs some that go right, and which readers familiar with, e.g., Koka~\cite{DBLP:conf/popl/Leijen17} or Frank~\cite{DBLP:journals/jfp/ConventLMM20}, might expect to be well-typed.
For example, consider the following program:
\begin{hscode}\SaveRestoreHook
\column{B}{@{}>{\hspre}l<{\hspost}@{}}%
\column{E}{@{}>{\hspre}l<{\hspost}@{}}%
\>[B]{}\Varid{example}\mathrel{=}\Varid{hState}\;\mathrm{0}\;\{\mskip1.5mu \mathbf{let}\;()\mathrel{=}(\op{put}\;\mathrm{1})\mathbin{!}\mathbf{in}\;(\Varid{hLocal}\;\mathrm{2}\;\{\mskip1.5mu \op{ask}\mathbin{!}\mskip1.5mu\})\mathbin{!}\mskip1.5mu\}{}\<[E]%
\ColumnHook
\end{hscode}\resethooks
Evaluation of this program will not get stuck: the state handler first updates the state to \ensuremath{\mathrm{1}}; subsequently, the local handler asks for the local binding \ensuremath{\mathrm{2}}.
Still, our type-and-effect system rejects it.
The reason is that \textsc{T-Let} requires both branches of the \ensuremath{\mathbf{let}} expression to have the same latent effects.
This implies that the type of the \ensuremath{\op{put}} operation is \ensuremath{\ann{\mathbin{⟨}\Conid{State}\mathbin{⟩}\mathbin{*}\mathbin{⟨}\Conid{Local}\mathbin{⟩}}{\susp{()}}}, which violates the operation typing restriction discussed above.

By insisting that operations have no latent effects we thus restrict programs that contain inline handlers, such as the \ensuremath{\Varid{example}} program above.
This may come as a surprise, as inline handlers are widely used in, e.g., Koka~\cite{DBLP:conf/popl/Leijen17} and Frank~\cite{DBLP:journals/jfp/ConventLMM20}.
However, in a system like \ensuremath{\Varid{λ}^{\Varid{hop}}}, we can generally use a higher-order operation anywhere you would want to use an inline handler.
In fact, we argue that we \emph{should} replace occurrences of inline handlers by higher-order operations, since inline handlers break down the distinction between the syntax and semantics of programs.
That is, the syntax of programs should be given by the operations that occur in them, while handlers should assign a semantics to those operations.
Programs implemented using inline handlers are not syntactic in this sense. 

\subsection{Operational Semantics}
\label{sec:operational-semantics}

We specify the run-time behaviour of \ensuremath{\Varid{λ}^{\Varid{hop}}} by giving a reduction semantics \`a
la Felleisen's \emph{prompt-control}~\cite{DBLP:conf/popl/Felleisen88}. The
rules are presented in \cref{fig:operational-semantics}. The operational
semantics is specified using three key components: 

\begin{enumerate}

  \item \emph{Evaluation contexts} (\ensuremath{\Conid{E}\mathbin{/}\Conid{E}^{\Varid{ℓ}}}), that capture the decomposition of a term into
    its \emph{leftmost-outermost} redex, and the surrounding expression. 

  \item
    A \emph{contraction relation} (\ensuremath{\mathbin{-->}}), that describes a single reduction
    step for a redex. 

  \item
    A \emph{reduction relation} (\ensuremath{ \longmapsto }), that describes the reduction of an
    expression at its leftmost-outermost redex. 
  
\end{enumerate}

\noindent
Together, these components specify evaluation for \ensuremath{\Varid{λ}^{\Varid{hop}}}: we repeatedly reduce
the leftmost-outermost redex of a term until it cannot be further reduced (i.e.,
it's a value). This evaluation strategy is formally captured by the reflexive
transitive closure of \ensuremath{ \longmapsto }, which we denote by \ensuremath{ \longtwoheadmapsto }. In
\cref{sec:type-soundness}, we show that, per Milner's motto,
evaluating well-typed expressions using \ensuremath{ \longtwoheadmapsto } \emph{cannot go wrong}.

\begin{figure}
\hrule%
\vspace{.5em}
\begin{align*}
  \ensuremath{\Conid{E}}   &::= \ensuremath{\mathbin{∙}\mathbin{∣}\Conid{E}\;\Varid{e}\mathbin{∣}\Varid{v}\;\Conid{E}\mathbin{∣}\mathbf{let}\;\Varid{x}\mathrel{=}\Conid{E}\;\mathbf{in}\;\Varid{e}\mathbin{∣}\Conid{E}\kw{!}\mathbin{∣}\kw{handle}^{\ell}\;\{\mskip1.5mu \overline{\Conid{C}}\mskip1.5mu\}\;\Conid{E}\;\Varid{e}} & \text{(Evaluation contexts)} \\ 
        & \ \ \ensuremath{\mathbin{∣}}\ \ \ensuremath{\kw{handle}^{\ell}\;\{\mskip1.5mu \overline{\Conid{C}}\mskip1.5mu\}\;\Varid{v}\;\Conid{E}}  \\
  \ensuremath{\Conid{E}^{\Varid{ℓ}}}  &::= \ensuremath{\mathbin{∙}\mathbin{∣}\Conid{E}^{\Varid{ℓ}}\;\Varid{e}\mathbin{∣}\Varid{v}\;\Conid{E}^{\Varid{ℓ}}\mathbin{∣}\mathbf{let}\;\Varid{x}\mathrel{=}\Conid{E}^{\Varid{ℓ}}\;\mathbf{in}\;\Varid{e}\mathbin{∣}\Conid{E}^{\Varid{ℓ}}\kw{!}\mathbin{∣}\kw{handle}^{\ell^\prime}\;\{\mskip1.5mu \overline{\Conid{C}}\mskip1.5mu\}\;\Conid{E}^{\Varid{ℓ}}\;\Varid{e}} \\
        & \ \ \ensuremath{\mathbin{∣}}\ \ \ensuremath{\kw{handle}^{\ell^\prime}\;\{\mskip1.5mu \overline{\Conid{C}}\mskip1.5mu\}\;\Varid{v}\;\Conid{E}^{\Varid{ℓ}}\;(}\text{if}\ \ensuremath{\Varid{ℓ}\neq\Varid{ℓ}^{\prime})} 
\end{align*}
\vspace{.5em}\raggedleft\framebox{\ensuremath{\Varid{e}\mathbin{-->}\Varid{e}}}%
\begin{mathpar}
  \inferrule[\ensuremath{\delta}]
  {\ensuremath{\delta} \ensuremath{(\Varid{c},\Varid{v})}\ \textsf{is defined} }
  {\ensuremath{\Varid{c}\;\Varid{v}\mathbin{-->}\delta} \ensuremath{(\Varid{c},\Varid{v})}}
  \and 
  \inferrule[\ensuremath{\Varid{β}}]{ }
  {\ensuremath{(\Varid{λ}\;\Varid{x}}.\ensuremath{\Varid{e})\;\Varid{v}\mathbin{-->}\Varid{e}}[\ensuremath{\Varid{x}}/\ensuremath{\Varid{v}}]}
  \and
  \inferrule[Let]{ }
  {\ensuremath{\mathbf{let}\;\Varid{x}\mathrel{=}\Varid{v}\;\mathbf{in}\;\Varid{e}\mathbin{-->}\Varid{e}}[x/v]}
  \and
  \inferrule[Enact]{ } 
  {\susp{e}\ensuremath{\kw{!}\mathbin{-->}\Varid{e}}}
  \and 
  \inferrule[Return]{ \ensuremath{\kw{return}\;\Varid{x}\;\Varid{p}\mathbin{↦}\Varid{e}\mathbin{∈}\overline{\Conid{C}}} }
  {\ensuremath{\kw{handle}^{\ell}\;\{\mskip1.5mu \overline{\Conid{C}}\mskip1.5mu\}\;\Varid{v}_{\Varid{p}}\;\Varid{v}\mathbin{-->}\Varid{e}}[p/\ensuremath{\Varid{v}_{\Varid{p}}},x/\ensuremath{\Varid{v}}]}
  \and
  \inferrule[Handle]
  {\ensuremath{\op{op}^{\ell}\;\overline{\Varid{x}}\;\Varid{p}\;\Varid{k}\mathbin{↦}\Varid{e}\mathbin{∈}\overline{\Conid{C}}}}
  {\ensuremath{\kw{handle}^{\ell}\;\{\mskip1.5mu \overline{\Conid{C}}\mskip1.5mu\}}\ \inE{\ensuremath{(\op{op}^{\ell}\;\Varid{f}\;\overline{\Varid{v}})\kw{!}}}\ \ensuremath{\Varid{v}_{\Varid{p}}\mathbin{-->}\Varid{e}}[\bsubst{x}{\ensuremath{\Varid{v}}},p/\ensuremath{\Varid{v}_{\Varid{p}}},k/\ensuremath{(\Varid{λ}\;\Varid{y}\;\Varid{q}}.\ensuremath{\kw{handle}^{\ell}\;\{\mskip1.5mu \overline{\Conid{C}}\mskip1.5mu\}\;\Varid{q}}\ \inE{\ensuremath{\Varid{y}\kw{!}}}\ \ensuremath{)}]}
\end{mathpar}
\raggedleft{\vspace{.5em}\framebox{\ensuremath{\Varid{e} \longmapsto \Varid{e}}}}%
\begin{mathpar}
  \inferrule
  { \ensuremath{e_{1}\mathbin{-->}e_{2}}  }
  { \ensuremath{E[e_{1}] \longmapsto }\ E[\ensuremath{e_{2}}] }
  \and
\end{mathpar}
\vspace{.5em}%
\hrule
\caption{Operational semantics}
\label{fig:operational-semantics}
\end{figure}

Evaluation contexts describe expressions with a single \emph{hole} (\ensuremath{\mathbin{∙}}) in
them. Their definition is carefully chosen to ensure that the hole always points
to the leftmost-outermost redex whenever we replace the hole with a contractible
expression. We define a variation, \ensuremath{\Conid{E}^{\Varid{ℓ}}}, that additionally ensures that the
context contains no handlers for the effect \ensuremath{\Varid{ℓ}}.

The \ensuremath{\delta}, \ensuremath{\Varid{β}}, and \textsc{Let} rules are the usual reduction rules for
call-by-value evaluation, where \ensuremath{\delta} is a partial function describing how to
evaluate fully-reduced applications of constants. \textsc{Enact} replaces
consecutive suspension and enactment with the original expression, and the
\textsc{Return} rule invokes the return clause of a handler, if both the handler
parameter and the handler argument have been evaluated to a value.

The \textsc{Handle} rule replaces operations with their semantics, assuming that
the handler parameter has been fully reduced. The handler's argument is a
decomposed expression, $\ensuremath{\Conid{E}^{\Varid{ℓ}}}[\ensuremath{(\op{op}^{\ell}\;\Varid{f}\;\overline{\Varid{v}})\kw{!}}]$, where the context \ensuremath{\Conid{E}^{\Varid{ℓ}}}
ensures that there are no closer handlers for \ensuremath{\Varid{ℓ}}. Reducing a handler means
replacing it with the corresponding clause for \ensuremath{\op{op}^{\ell}\;\Varid{f}}, appropriately
substituting the operation arguments, handler parameter, and continuation. In
the continuation, we install the same handler on the composite expression
$\ensuremath{\Conid{E}^{\Varid{ℓ}}}[\ensuremath{\Varid{y}\kw{!}}]$, giving the clause implementation an opportunity to discharge
any latent effects in sub-computations of \ensuremath{\op{op}^{\ell}\;\Varid{f}}, or the effect \ensuremath{\Varid{ℓ}} itself, to
the context of \ensuremath{\op{op}^{\ell}\;\Varid{f}}. Importantly, the type-and-effect system statically
guarantees that the context \ensuremath{\Conid{E}^{\Varid{ℓ}}} contains handlers for all latent effects, by
restricting the latent effect annotation of operations to the empty row. 

\subsection{Type Soundness}
\label{sec:type-soundness}

We prove soundness of our type system (\cref{fig:typing-rules}) with respect to
the operational semantics (\cref{fig:operational-semantics}) via the syntactic
approach of \citet{DBLP:journals/iandc/WrightF94}. The proof for the final
soundness theorem (\cref{thm:soundness}) follows from \emph{subject reduction}
(reducing expressions preserves their type, \cref{lem:subject-reduction}) and
\emph{progress} (every expression is either a value or can be reduced,
\cref{lem:progress}).

We start by proving subject reduction. For this, we need three additional
lemmas. 
\begin{lemma}[Substitution]
  \label{lem:substitution}
  Given an expression \ensuremath{\Varid{e}} such that
  \ensuremath{\Conid{Γ},\Varid{x}\mathbin{:∀}\overline{\Varid{α}}\;\overline{\Varid{r}}}.\ensuremath{\tau_{1}\mathbin{⊢}\Varid{e}\mathbin{:}\tau_{2}\mathbin{∣}\Varid{ε}\mathbin{*}\epsilon_{l}}, such that $x \notin \text{dom}(Γ)$, then for every value \ensuremath{\Varid{v}} where \ensuremath{\Conid{Γ}\mathbin{⊢}\Varid{v}\mathbin{:}\tau_{1}\mathbin{∣}\Varid{ε}\mathbin{*}\epsilon_{l}}, it follows that \ensuremath{\Conid{Γ}\mathbin{⊢}\Varid{e}}[x/v]\ensuremath{\mathbin{:}\tau_{2}\mathbin{∣}\Varid{ε}\mathbin{*}\epsilon_{l}} provided that \ensuremath{\overline{\Varid{α}},\overline{\Varid{r}}}$\ \notin \text{ftv}(Γ)$.
\end{lemma}
 
\begin{lemma}[Replacement]
  \label{lem:replacement}
  Given a deduction of the form \ensuremath{\Conid{Γ}\mathbin{⊢}E[\Varid{e}]\mathbin{:}\Varid{τ}\mathbin{∣}\Varid{ε}\mathbin{*}\epsilon_{l}} and its sub-deduction
  \ensuremath{\Conid{Γ}\mathbin{⊢}\Varid{e}\mathbin{:}\Varid{τ}^{\prime}\mathbin{∣}\Varid{ε}^{\prime}\mathbin{*}\epsilon_{l}^{\prime}}, we can replace this sub-deduction by another
  sub-deduction \ensuremath{\Conid{Γ}\mathbin{⊢}\Varid{e}^{\prime}\mathbin{:}\Varid{τ}^{\prime}\mathbin{∣}\Varid{ε}^{\prime}\mathbin{*}\epsilon_{l}^{\prime}} such that \ensuremath{\Conid{Γ}\mathbin{⊢}E[\Varid{e}]^{\prime}\mathbin{:}\Varid{τ}\mathbin{∣}\Varid{ε}\mathbin{*}\Varid{ε}^{\prime}}. 
\end{lemma}

\begin{lemma}[Re-annotation]
  \label{lem:reannotation}
  Values can be annotated with any effect. That is, for every value \ensuremath{\Varid{v}}, if \ensuremath{\Conid{Γ}\mathbin{⊢}\Varid{v}\mathbin{:}\Varid{τ}\mathbin{∣}\Varid{ε}\mathbin{*}\epsilon_{l}}, then also \ensuremath{\Conid{Γ}\mathbin{⊢}\Varid{v}\mathbin{:}\Varid{τ}\mathbin{∣}\Varid{ε}^{\prime}\mathbin{*}\epsilon_{l}^{\prime}} for every \ensuremath{\Varid{ε}^{\prime}}, \ensuremath{\epsilon_{l}^{\prime}}.
\end{lemma}

Proofs for \cref{lem:substitution} and \cref{lem:replacement} are
straightforward adaptations of the corresponding proofs by
\citet{DBLP:journals/iandc/WrightF94}. A proof for \cref{lem:reannotation}
follows from case analysis on \ensuremath{\Varid{v}}. We will use these lemmas to prove subject
reduction. First, we prove subject reduction for the contraction relation \ensuremath{\mathbin{-->}}.

\begin{lemma}[Subject Reduction/Contraction]
  \label{lem:subject-reduction-contraction}
  For all expressions \ensuremath{e_{1}} and \ensuremath{e_{2}}, if \ensuremath{\Conid{Γ}\mathbin{⊢}e_{1}\mathbin{:}\Varid{τ}\mathbin{∣}\Varid{ε}\mathbin{*}\epsilon_{l}} and \ensuremath{e_{1}\mathbin{-->}e_{2}},
  it follows that \ensuremath{\Conid{Γ}\mathbin{⊢}e_{2}\mathbin{:}\Varid{τ}\mathbin{∣}\Varid{ε}\mathbin{*}\epsilon_{l}}. 
\end{lemma}

The proof proceeds by rule induction over the derivation \ensuremath{\Conid{Γ}\mathbin{⊢}e_{1}\mathbin{:}\Varid{τ}\mathbin{∣}\Varid{ε}\mathbin{*}\epsilon_{l}}. See \cref{sec:proof-appendix} for a proof sketch.

Next, we extend subject reduction to the reduction relation \ensuremath{ \longmapsto }: 

\begin{lemma}[Subject Reduction]
\label{lem:subject-reduction} 
  For all expressions \ensuremath{e_{1}} and \ensuremath{e_{2}}, if \ensuremath{\Conid{Γ}\mathbin{⊢}e_{1}\mathbin{:}\Varid{τ}\mathbin{∣}\Varid{ε}\mathbin{*}\epsilon_{l}} and \ensuremath{e_{1} \longmapsto e_{2}},
  then \ensuremath{\Conid{Γ}\mathbin{⊢}e_{2}\mathbin{:}\Varid{τ}\mathbin{∣}\Varid{ε}\mathbin{*}\epsilon_{l}}
\end{lemma}

\begin{proof}
  The proof is an immediate consequence from replacement
  (\cref{lem:replacement}), subject reduction for the contraction relation
  (\cref{lem:subject-reduction-contraction}), and the definition of \ensuremath{ \longmapsto }
  (\cref{fig:operational-semantics}).
\end{proof} 

The next part of proving type soundness is to show that well-typed programs
will not get stuck, which we prove directly by rule induction over the typing
derivation.

\begin{lemma}[Progress]
  \label{lem:progress}
  For every expression \ensuremath{\Varid{e}}, if \ensuremath{\Conid{Γ}_{\mathrm{0}}\mathbin{⊢}\Varid{e}\mathbin{:}\Varid{τ}\mathbin{∣}\Varid{ε}\mathbin{*}\epsilon_{l}}, then either there exists
  some \ensuremath{\Varid{e}^{\prime}} such that \ensuremath{\Varid{e} \longmapsto \Varid{e}^{\prime}}, or \ensuremath{\Varid{e}} is a value, where \ensuremath{\Conid{Γ}_{\mathrm{0}}} is an initial
  context that contains bindings for the operations of all known effects in \ensuremath{\Conid{Σ}}.
\end{lemma}

The proof is routine and proceeds by rule induction over the derivation for \ensuremath{\Conid{Γ}_{\mathrm{0}}\mathbin{⊢}\Varid{e}\mathbin{:}\Varid{τ}\mathbin{∣}\Varid{ε}\mathbin{*}\epsilon_{l}}. See \cref{sec:proof-appendix} for a more detailed proof
sketch.

\paragraph{Type Soundness}

We now have all the necessary ingredients to state and show type soundness of
our type system. Let \ensuremath{\Varid{e}\Uparrow} denote a \emph{diverging computation}; i.e., the
expression \ensuremath{\Varid{e}} has a never-ending reduction. 

\begin{theorem}[Type Soundness]
  \label{thm:soundness}
  For every \ensuremath{\Varid{e}}, if \ensuremath{\Conid{Γ}_{\mathrm{0}}\mathbin{⊢}\Varid{e}\mathbin{:}\Varid{τ}\mathbin{∣}\Varid{ε}\mathbin{*}\epsilon_{l}}, then either \ensuremath{\Varid{e}\Uparrow}, or there exists a
  value \ensuremath{\Varid{v}} and rows \ensuremath{\Varid{ε}^{\prime}}, \ensuremath{\epsilon_{l}^{\prime}}, such that \ensuremath{\Varid{e} \longtwoheadmapsto \Varid{v}}, \ensuremath{\Conid{Γ}_{\mathrm{0}}\mathbin{⊢}\Varid{v}\mathbin{:}\Varid{τ}\mathbin{∣}\Varid{ε}\mathbin{*}\epsilon_{l}},
  and \ensuremath{\Varid{ε}\mathbin{⊕}\epsilon_{l}\simeq\Varid{ε}^{\prime}\mathbin{⊕}\epsilon_{l}^{\prime}}.
\end{theorem}

\begin{proof}
  From \cref{lem:progress} it follows that if \ensuremath{\Conid{Γ}_{\mathrm{0}}\mathbin{⊢}\Varid{e}\mathbin{:}\Varid{τ}\mathbin{∣}\Varid{ε}\mathbin{*}\epsilon_{l}}, either there
  exists some \ensuremath{\Varid{v}} such that \ensuremath{\Varid{e} \longtwoheadmapsto \Varid{v}}, or \ensuremath{\Varid{e}\Uparrow}. If \ensuremath{\Varid{e} \longtwoheadmapsto \Varid{v}}, then from
  \cref{lem:subject-reduction} and the definition of \ensuremath{ \longtwoheadmapsto } we conclude that \ensuremath{\Conid{Γ}_{\mathrm{0}}\mathbin{⊢}\Varid{v}\mathbin{:}\Varid{τ}\mathbin{∣}\Varid{ε}\mathbin{*}\epsilon_{l}}.
\end{proof}

\section{A Formal Characterization of Handlers with Separate Concerns}
\label{sec:meta-theory}

As discussed in \cref{sec:approach}, many effect handlers can be defined and handled in a way that running handlers in different orders only changes how return results are packaged, but not the return results themselves.
We argue that handlers that have this property are well-behaved, because it is possible to seamlessly compose these handlers in any order, without surprising effect interactions.
But what, precisely, does it mean that running handlers in different orders only changes how return results are packaged, but not the return results themselves?  And which handlers have this property?
In this section we characterize what it means for handlers to have \emph{separate concerns}.
In the next section we prove that common algebraic effect handlers have this property, and discuss counter-examples of handlers that do not have the property.

Notationally, we will use both the symbol \ensuremath{ \Longrightarrow } and horizontal inference lines for logical implication in this section.
The functions we define in this section are generally logical meta-level functions as opposed to \ensuremath{\Varid{λ}^{\Varid{hop}}} functions.

\subsection{Separation of Concerns}

We say that handlers have the \emph{separation of concerns} (s.o.c.) property when the handlers can be applied in any order, to yield equivalent results.
As hinted at above, we will compare results up-to how they are ``packaged''.
By packaging we mean the type \ensuremath{\Conid{F}\mathbin{:}\mathbin{★}\mathbin{→}\mathbin{★}} that handlers decorate their return results by (we will use both the term \emph{functor} and \emph{answer type modifier} to describe \ensuremath{\Conid{F}}).
For example, the answer type modifier of \ensuremath{\Varid{hCatch}} is \ensuremath{\Conid{Maybe}}, since the type of the \ensuremath{\Varid{hCatch}} handler is \ensuremath{\mathbin{∀}\Varid{α}\;\Varid{r}\;\Varid{r}^{\prime}\!.\!\;\ann{\mathbin{⟨}\Conid{Ca},\Varid{r}\mathbin{⟩}\mathbin{*}\Varid{r}^{\prime}}{\susp{\Varid{α}}}\mathbin{→}\ann{\Varid{r}\mathbin{*}\mathbin{⟨}\Conid{Ca},\Varid{r}^{\prime}\mathbin{⟩}}{\susp{\Conid{Maybe}\;\Varid{α}}}}.
For handlers that have no answer type modifier, we assume that there is an implicit identity answer type modifier \ensuremath{\Conid{Id}\mathbin{:}\mathbin{★}\mathbin{→}\mathbin{★}}.
We introduce an equivalence for comparing results by ``extracting'' bags of answers from answer type modifier wrapped results.

\subsubsection{Baggable Functor Equivalence}

We say that a functor \ensuremath{\Conid{F}} is \emph{baggable} when it has a function that extracts the bag (also known as a \emph{multiset}) of elements of type \ensuremath{\Conid{A}} contained in a functor structure \ensuremath{\Conid{F}\;\Conid{A}}; i.e.:
\begin{equation*}
  \ensuremath{\Varid{toBag}_{\Conid{F}}\mathbin{:}\Conid{F}\;\Conid{A}\mathbin{→}\Conid{Bag}\;\Conid{A}}
\end{equation*}
The notion of equivalence we will use to define our s.o.c. property is the following.
For any two baggable functors \ensuremath{\Conid{F}} and \ensuremath{\Conid{G}}, and for any type \ensuremath{\Conid{A}}, we can compare results of type \ensuremath{\Conid{F}\;\Conid{A}} and \ensuremath{\Conid{G}\;\Conid{A}} by comparing their corresponding bags; i.e.:
\begin{align*}
  \inferrule*
  {    \ensuremath{\Varid{x}\mathbin{:}\Conid{F}\;\Conid{A}}
  \and \ensuremath{\Varid{y}\mathbin{:}\Conid{G}\;\Conid{A}}
  \and \ensuremath{\Varid{toBag}_{\Conid{F}}\;\Varid{x}\mathrel{=}\Varid{toBag}_{\Conid{G}}\;\Varid{y}}
  }
  { \ensuremath{\Varid{x}\sim\Varid{y}} }
\end{align*}

To define the s.o.c. property, we will also make use of the fact that, for any two baggable functors \ensuremath{\Conid{F}} and \ensuremath{\Conid{G}}, we get the following \ensuremath{\Varid{toBag}_{\Conid{F}\mathbin{∘}\Conid{G}}} function for free:
\begin{align*}
  &\ensuremath{\Varid{toBag}_{\Conid{F}\mathbin{∘}\Conid{G}}} : \ensuremath{\Conid{F}\;(\Conid{G}\;\Conid{A})\mathbin{→}\Conid{Bag}\;\Conid{A}}
  \\
  &\ensuremath{\Varid{toBag}_{\Conid{F}\mathbin{∘}\Conid{G}}} = \ensuremath{\Varid{join}\mathbin{∘}\Varid{fmap}_{\Conid{Bag}}\;\Varid{toBag}_{\Conid{F}}\mathbin{∘}\Varid{fmap}_{\Conid{G}}}
\end{align*}
Here, \ensuremath{\Varid{join}\mathbin{:}\Conid{Bag}\;(\Conid{Bag}\;\Conid{A})\mathbin{→}\Conid{Bag}\;\Conid{A}} is a function that joins a bag of bags using bag union, and \ensuremath{\Varid{fmap}_{\Conid{Bag}}\mathbin{:}(\Conid{A}\mathbin{→}\Conid{B})\mathbin{→}\Conid{Bag}\;\Conid{A}\mathbin{→}\Conid{Bag}\;\Conid{B}} maps a function over a bag to modify its elements.

\subsubsection{The Separation of Concerns Property}
We will define separation of concerns for handlers by assuming that a handler of effect \ensuremath{\Varid{ℓ}} is given by a function \ensuremath{\Varid{h}\mathbin{:}\mathbin{∀}\Varid{α}\;\Varid{r}\;\Varid{r'}\;\!.\!\;\ann{\mathbin{⟨}\Varid{ℓ},\Varid{r}\mathbin{⟩}\mathbin{*}\Varid{r'}}{\susp{\Varid{α}}}\mathbin{→}\Conid{P}\mathbin{→}\ann{\Varid{r}\mathbin{*}\mathbin{⟨}\Varid{ℓ},\Varid{r'}\mathbin{⟩}}{\susp{\Conid{F}\;\Varid{α}}}}, where \ensuremath{\Conid{P}} is the parameter type of the handler.
We will write \ensuremath{\Varid{h}^{\Varid{p}}} to represent an invocation of a handler \ensuremath{\Varid{h}} with parameter \ensuremath{\Varid{p}}.
Now, for any two handlers
\begin{gather*}
  \ensuremath{\Varid{h}_1\mathbin{:}\mathbin{∀}\Varid{α}\;\Varid{r}\;\Varid{r'}} .\ \ensuremath{\Conid{P}_1\mathbin{→}\ann{\mathbin{⟨}\ell_{1},\Varid{r}\mathbin{⟩}\mathbin{*}\Varid{r'}}{\susp{\Varid{α}}}\mathbin{→}\ann{\Varid{r}\mathbin{*}\mathbin{⟨}\ell_{1},\Varid{r'}\mathbin{⟩}}{\susp{\Conid{F}_1\;\Varid{α}}}}
  \\
  \ensuremath{\Varid{h}_2\mathbin{:}\mathbin{∀}\Varid{α}\;\Varid{r}\;\Varid{r'}} .\ \ensuremath{\Conid{P}_2\mathbin{→}\ann{\mathbin{⟨}\ell_{2},\Varid{r}\mathbin{⟩}\mathbin{*}\Varid{r'}}{\susp{\Varid{α}}}\mathbin{→}\ann{\Varid{r}\mathbin{*}\mathbin{⟨}\ell_{2},\Varid{r'}\mathbin{⟩}}{\susp{\Conid{F}_2\;\Varid{α}}}}
\end{gather*}
where \ensuremath{\ell_{1}\neq\ell_{2}} and where \ensuremath{\Conid{F}_1} and \ensuremath{\Conid{F}_2} are baggable functors, we say that \ensuremath{\Varid{h}_1} and \ensuremath{\Varid{h}_2} have \emph{separate concerns} when, for any \ensuremath{\Varid{e}\mathbin{:}\ann{\mathbin{⟨}\ell_{1},\ell_{2}\mathbin{⟩}\mathbin{*}\Varid{ε'}}{\susp{\Conid{X}}}},
\begin{equation*}
  \inferrule*
  {    
       \ensuremath{(\Varid{h}_1^{\Varid{p}_{1}}\;(\Varid{h}_2^{\Varid{p}_{2}}\;\Varid{e}))\kw{!} \longtwoheadmapsto \Varid{v}_{\mathrm{21}}}
  \and \ensuremath{(\Varid{h}_2^{\Varid{p}_{2}}\;(\Varid{h}_1^{\Varid{p}_{1}}\;\Varid{e}))\kw{!} \longtwoheadmapsto \Varid{v}_{\mathrm{12}}} }
  { \ensuremath{\Varid{v}_{\mathrm{21}}\sim\Varid{v}_{\mathrm{12}}} }.
\end{equation*}
This says that \ensuremath{\Varid{h}_1} and \ensuremath{\Varid{h}_2} can be run in any order, when \ensuremath{\Varid{e}} \emph{only} has residual effects that are handled by \ensuremath{\Varid{h}_1} and \ensuremath{\Varid{h}_2}; i.e., \ensuremath{\ell_{1}} and \ensuremath{\ell_{2}}.
More generally, we say that \ensuremath{\Varid{h}_1} and \ensuremath{\Varid{h}_2} have separate concerns under a handler context \ensuremath{\Varid{h}\mathbin{:}\mathbin{∀}\Varid{α}}.\ \ensuremath{\Conid{P}\mathbin{→}\ann{\Varid{ε}\mathbin{*}\mathbin{⟨}\ell_{1},\mathbin{⟨}\ell_{2},\Varid{ε'}\mathbin{⟩}\mathbin{⟩}}{\susp{\Varid{α}}}\mathbin{→}\ann{\mathbin{⟨⟩}\mathbin{*}(\Varid{ε}\mathbin{⊕}\mathbin{⟨}\ell_{1},\mathbin{⟨}\ell_{2},\Varid{ε'}\mathbin{⟩}\mathbin{⟩})}{\susp{\Conid{G}\;\Varid{α}}}}, where \ensuremath{\Conid{G}} is a baggable functor, when for any \ensuremath{\Varid{e}\mathbin{:}\ann{\mathbin{⟨}\ell_{1},\mathbin{⟨}\ell_{2},\Varid{ε}\mathbin{⟩⟩}\mathbin{*}\Varid{ε'}}{\susp{\Conid{X}}}},
\begin{equation*}
  \inferrule
  {  \ensuremath{(\Varid{h}^{\Varid{p}}\;(\Varid{h}_1^{\Varid{p}_{1}}\;(\Varid{h}_2^{\Varid{p}_{2}}\;\Varid{e})))\kw{!} \longtwoheadmapsto \Varid{v}_{\mathrm{21}}}
  \\ \ensuremath{(\Varid{h}^{\Varid{h}}\;(\Varid{h}_2^{\Varid{p}_{2}}\;(\Varid{h}_1^{\Varid{p}_{1}}\;\Varid{e})))\kw{!} \longtwoheadmapsto \Varid{v}_{\mathrm{12}}}
  }
  {  \ensuremath{\Varid{v}_{\mathrm{21}}\sim\Varid{v}_{\mathrm{12}}} }
\end{equation*}

In the proposition above, the handler context \ensuremath{\Varid{h}} may be a nested stack of handlers.
Similarly, \ensuremath{\Varid{e}} may contain nested handlers.
Since we can compose and re-associate handlers using function composition, the proposition tells us that we can reorder any pair of handlers in a handler stack.
We thus say that two handlers \ensuremath{\Varid{h}_1} and \ensuremath{\Varid{h}_2} have separate concerns iff reordering them in a nested stack of handlers yields an equivalent bag of values.
And we say that a handler \ensuremath{\Varid{h}_1} has separate concerns iff we can reorder it with any different handler \ensuremath{\Varid{h}_2} in a handler context to yield an equivalent bag of values.
In the next section we provide a set of criteria for proving that a given handler has separate concerns.

\subsubsection{Criteria for Separation of Concerns}

A handler has separate concerns when it satisfies the \emph{handler return lemma} (\cref{lem:handler-return}) and the \emph{handler op lemma} (\cref{lem:handler-op}).

\begin{lemma}[Handler Return Lemma]
  The return case of handler \ensuremath{\Varid{h}\mathbin{:}\mathbin{∀}\Varid{α}\;\Varid{r}\;\Varid{r'}} .\ \ensuremath{\Conid{P}\mathbin{→}\ann{\mathbin{⟨}\ell_{1},\Varid{r}\mathbin{⟩}\mathbin{*}\Varid{r'}}{\susp{\Varid{α}}}\mathbin{→}\ann{\Varid{r}\mathbin{*}\mathbin{⟨}\ell_{1},\Varid{r'}\mathbin{⟩}}{\susp{\Conid{F}\;\Varid{α}}}}, where \ensuremath{\Conid{F}} is a baggable functor, yields a result that corresponds to a singleton bag.
  That is, for any value \ensuremath{\Varid{v}} and parameter \ensuremath{\Varid{p}},
  \begin{equation*}
  \ensuremath{\Varid{toBag}_{\Conid{F}}\;(\Varid{h}^{\Varid{p}}\;\susp{\Varid{v}})\kw{!}\mathrel{=}\{\mskip1.5mu \Varid{v}\mskip1.5mu\}}
  \end{equation*}
  \label{lem:handler-return}
\end{lemma}

\begin{lemma}[Handler Op Lemma]
  A handler \ensuremath{\Varid{h}_1\mathbin{:}\mathbin{∀}\Varid{α}\;\Varid{r}\;\Varid{r'}} .\ \ensuremath{\Conid{P}_1\mathbin{→}\ann{\mathbin{⟨}\Varid{ℓ},\Varid{r}\mathbin{⟩}\mathbin{*}\Varid{r'}}{\susp{\Varid{α}}}\mathbin{→}\ann{\Varid{r}\mathbin{*}\mathbin{⟨}\Varid{ℓ},\Varid{r'}\mathbin{⟩}}{\susp{\Conid{F}_1\;\Varid{α}}}} has separate concerns when all of its handler clauses are handled in a way that provides separation of concerns.
  That is, for any \ensuremath{\op{op}^{\ell}\;\Varid{f}\;\overline{\Varid{v}}},
  any \ensuremath{\Varid{h}\mathbin{:}\mathbin{∀}\Varid{α}\;\!.\!\;\Conid{P}\mathbin{→}\ann{\Varid{ε}\mathbin{*}\mathbin{⟨}\ell_{1},\mathbin{⟨}\ell_{2},\Varid{ε}^{\prime}\mathbin{⟩}\mathbin{⟩}}{\susp{\Varid{α}}}\mathbin{→}\ann{\mathbin{⟨⟩}\mathbin{*}(\Varid{ε}\mathbin{⊕}\mathbin{⟨}\ell_{1},\mathbin{⟨}\ell_{2},\Varid{ε}^{\prime}\mathbin{⟩}\mathbin{⟩})}{\susp{\Conid{G}\;\Varid{α}}}},
  any \ensuremath{\Varid{h}_2\mathbin{:}\mathbin{∀}\Varid{α}\;\Varid{r}\;\Varid{r'}\;\!.\!\;\Conid{P}_2\mathbin{→}\ann{\mathbin{⟨}\Varid{ℓ'},\Varid{r}\mathbin{⟩}\mathbin{*}\Varid{r'}}{\susp{\Varid{α}}}\mathbin{→}\ann{\Varid{r}\mathbin{*}\mathbin{⟨}\Varid{ℓ'},\Varid{r'}\mathbin{⟩}}{\susp{\Conid{F}_2\;\Varid{α}}}},
  and any \ensuremath{\Varid{e}\mathbin{:}\ann{\mathbin{⟨}\ell_{1},\mathbin{⟨}\ell_{2},\Varid{ε}\mathbin{⟩⟩}\mathbin{*}\Varid{ε'}}{\susp{\Conid{X}}}},
  where \ensuremath{\ell_{1}\neq\ell_{2}}, and where \ensuremath{\Varid{h}_1} and \ensuremath{\Varid{h}_2} satisfy \cref{lem:handler-return}, the following holds:
  \begin{equation*}
    \inferrule
    {    \ensuremath{(\Varid{h}^{\Varid{p}}\;(\Varid{h}_1^{\Varid{p}_{1}}\;(\Varid{h}_2^{\Varid{p}_{2}}\;\Conid{E}[\op{op}^{\ell}\;\Varid{f}\;\overline{\Varid{v}}])))\kw{!} \longmapsto \Varid{e'}}
    \and \ensuremath{\Varid{e}^{\prime} \longtwoheadmapsto \Varid{v}_{\mathrm{21}}}
    \\   \ensuremath{(\Varid{h}^{\Varid{p}}\;(\Varid{h}_2^{\Varid{p}_{2}}\;(\Varid{h}_1^{\Varid{p}_{1}}\;\Conid{E}[\op{op}^{\ell}\;\Varid{f}\;\overline{\Varid{v}}])))\kw{!} \longtwoheadmapsto \Varid{v}_{\mathrm{12}}}
    \\   \left(\ensuremath{\mathbin{∀}e_{0}\;\!.\!\;\Varid{e}^{\prime} \longtwoheadmapsto _{\beta}^{*}(\Varid{h}^{\Varid{p}}\;(\Varid{h}_1^{\Varid{p}_{1}}\;(\Varid{h}_2^{\Varid{p}_{2}}\;e_{0})))\kw{!} \Longrightarrow \Varid{h}^{\Varid{p}}\;(\Varid{h}_2^{\Varid{p}_{2}}\;(\Varid{h}_1^{\Varid{p}_{1}}\;e_{0})) \longtwoheadmapsto \Varid{v}_{\mathrm{12}} \Longrightarrow \Varid{v}_{\mathrm{21}}\sim\Varid{v}_{\mathrm{12}}}\right)
    }
    { \ensuremath{\Varid{v}_{\mathrm{21}}\sim\Varid{v}_{\mathrm{12}}}
    }
  \end{equation*}
  \label{lem:handler-op}
\end{lemma}

The \ensuremath{ \longtwoheadmapsto _{\beta}^{*}} relation in \cref{lem:handler-op} represents a sequence of zero or more \ensuremath{\Varid{β}} reduction steps.

\Cref{thm:soc} below shows that any handler that satisfies the handler return and op lemmas has separate concerns.

\begin{theorem}[Separation of Concerns]
  Any handler \ensuremath{\Varid{h}_1\mathbin{:}\mathbin{∀}\Varid{α}\;\Varid{r}\;\Varid{r'}} .\ \ensuremath{\Conid{P}_1\mathbin{→}\ann{\mathbin{⟨}\Varid{ℓ},\Varid{r}\mathbin{⟩}\mathbin{*}\Varid{r'}}{\susp{\Varid{α}}}\mathbin{→}\ann{\Varid{r}\mathbin{*}\mathbin{⟨}\Varid{ℓ},\Varid{r'}\mathbin{⟩}}{\susp{\Conid{F}_1\;\Varid{α}}}} that satisfies the handler return lemma (\cref{lem:handler-return}) and handler op lemma (\cref{lem:handler-op}) has \emph{separate concerns}.
  That is, for
  any \ensuremath{\Varid{h}\mathbin{:}\mathbin{∀}\Varid{α}\;\!.\!\;\Conid{P}\mathbin{→}\ann{\Varid{ε}\mathbin{*}\mathbin{⟨}\ell_{1},\mathbin{⟨}\ell_{2},\Varid{ε}^{\prime}\mathbin{⟩}\mathbin{⟩}}{\susp{\Varid{α}}}\mathbin{→}\ann{\mathbin{⟨⟩}\mathbin{*}(\Varid{ε}\mathbin{⊕}\mathbin{⟨}\ell_{1},\mathbin{⟨}\ell_{2},\Varid{ε}^{\prime}\mathbin{⟩}\mathbin{⟩})}{\susp{\Conid{G}\;\Varid{α}}}}
  where each handler in the handler context \ensuremath{\Varid{h}} satisfies the handler return lemma,
  any \ensuremath{\Varid{h}_2\mathbin{:}\mathbin{∀}\Varid{α}\;\Varid{r}\;\Varid{r'}\;\!.\!\;\Conid{P}_2\mathbin{→}\ann{\mathbin{⟨}\Varid{ℓ'},\Varid{r}\mathbin{⟩}\mathbin{*}\Varid{r'}}{\susp{\Varid{α}}}\mathbin{→}\ann{\Varid{r}\mathbin{*}\mathbin{⟨}\Varid{ℓ'},\Varid{r'}\mathbin{⟩}}{\susp{\Conid{F}_2\;\Varid{α}}}}
  that also satisfies the handler return and the handler op lemma,
  and any \ensuremath{\Varid{e}\mathbin{:}\ann{\mathbin{⟨}\ell_{1},\mathbin{⟨}\ell_{2},\Varid{ε}\mathbin{⟩⟩}\mathbin{*}\Varid{ε'}}{\susp{\Conid{X}}}},
  \begin{equation*}
    \inferrule
    {  \ensuremath{(\Varid{h}^{\Varid{p}}\;(\Varid{h}_1^{\Varid{p}_{1}}\;(\Varid{h}_2^{\Varid{p}_{2}}\;\Varid{e})))\kw{!} \longtwoheadmapsto \Varid{v}_{\mathrm{21}}}
    \\ \ensuremath{(\Varid{h}^{\Varid{p}}\;(\Varid{h}_2^{\Varid{p}_{2}}\;(\Varid{h}_1^{\Varid{p}_{1}}\;\Varid{e})))\kw{!} \longtwoheadmapsto \Varid{v}_{\mathrm{12}}}
    }
    {  \ensuremath{\Varid{v}_{\mathrm{21}}\sim\Varid{v}_{\mathrm{12}}} }
  \end{equation*}
  \begin{proof}[Proof (sketch -- see \cref{sec:soc-appendix} for full proof)]
    The proof is by rule induction on \ensuremath{ \longtwoheadmapsto } in the premise of the goal.
    The base case, where \ensuremath{ \longtwoheadmapsto } represents zero steps, holds vacuously, because \ensuremath{(\Varid{h}^{\Varid{p}}\;(\Varid{h}_1^{\Varid{p}_{1}}\;(\Varid{h}_2^{\Varid{p}_{2}}\;\Varid{e})))\kw{!}} is not a value.
    The inductive case proceeds by inversion on the \ensuremath{ \longmapsto } and \ensuremath{\mathbin{-->}} relations, and the resulting goals follow from the handler return and handler op lemmas, since the induction hypothesis matches matches exactly the premise of the op handler lemma.
  \end{proof}
  \label{thm:soc}
\end{theorem}

\subsection{Discussion}

We have given a formal characterization of handlers that have separate concerns.
While our formal characterization is defined in terms of the notion of handlers of higher-order effects that this paper introduces, the characterization is expected to apply to plain algebraic effect handlers, and (perhaps with some adjustments) to scoped effect handlers.

Our characterization is currently limited to talking about a class of handlers that handles a \emph{single} effect and is polymorphic in both the remaining set of effects and in the return type of the computation being handled.
While this class of handlers is representative of many effects and effect handlers found in the literature, it does not contain effect handlers that handle multiple effects at once, nor effect handlers that assume that the input computation has a particular return type or a particular set of effects.
Our characterization of separate concerns readily generalizes to effect handlers that handle multiple effects at once, and to effect handlers that are non-polymorphic in the type of their input computation; e.g., \ensuremath{\Varid{h}\mathbin{:}\mathbin{∀}\Varid{r}\;\Varid{r}^{\prime}\!.\!\;\Conid{P}\mathbin{→}\ann{\mathbin{⟨}\Varid{ℓ},\Varid{r}\mathbin{⟩}\mathbin{*}\Varid{r'}}{\susp{\Conid{X}}}\mathbin{→}\ann{\Varid{r}\mathbin{*}\mathbin{⟨}\Varid{ℓ},\Varid{r'}\mathbin{⟩}}{\susp{\Conid{F}\;\Conid{X}}}} for some specific type \ensuremath{\Conid{X}}.
Handlers that are not polymorphic in \ensuremath{\Varid{r}} and \ensuremath{\Varid{r}^{\prime}} cannot be reordered as freely, as a specific row of handlers implies the presence of a particular set of effects---which may change when reordering handlers.
But it is still possible to prove a form of separation of concerns for such handlers, as we demonstrate below in \cref{sec:soc-ho-eff} where we prove a notion of separation of concerns for a handler for \ensuremath{\Varid{λ}} abstraction which is non-polymorphic in its row of effects.


\section{Separation of Concerns for Handlers of Common Effects}
\label{sec:soc-handlers}

We show (in \Cref{sec:soc-alg-eff,sec:soc-ho-eff}) that the characterization of separate concerns from \cref{sec:meta-theory} applies to handlers of common effects found in programming languages.
We also show (in \Cref{sec:soc-non-separable}) how handlers of other common effects, such as non-determinism, do \emph{not} have separable concerns.
Finally, in \cref{sec:soc-conjecture} we discuss a general criterion for handlers, which we conjecture implies separation of concerns.
The examples in this section are defined for a calculus that extends the core calculus from \cref{sec:calculus} with recursive bindings, \ensuremath{\kw{match}} expressions, and with data values (i.e., terms that represent constructors of data types).

\subsection{Separation of Concerns for Three Common Algebraic Effects}
\label{sec:soc-alg-eff}

\Cref{fig:alg-eff-handlers} shows the handlers of three common algebraic effects: a \emph{state effect}, \emph{abort effect}, and \emph{read-only effect}.
We show that these handlers satisfy \cref{lem:handler-return} and \cref{lem:handler-op}, and thus have separate concerns.

\begin{figure}
{\renewcommand{\hscodestyle}{\Small}
\begin{minipage}[t]{0.32\textwidth}
  \begin{hscode}\SaveRestoreHook
\column{B}{@{}>{\hspre}l<{\hspost}@{}}%
\column{3}{@{}>{\hspre}l<{\hspost}@{}}%
\column{5}{@{}>{\hspre}l<{\hspost}@{}}%
\column{6}{@{}>{\hspre}c<{\hspost}@{}}%
\column{6E}{@{}l@{}}%
\column{9}{@{}>{\hspre}l<{\hspost}@{}}%
\column{13}{@{}>{\hspre}l<{\hspost}@{}}%
\column{16}{@{}>{\hspre}l<{\hspost}@{}}%
\column{19}{@{}>{\hspre}l<{\hspost}@{}}%
\column{E}{@{}>{\hspre}l<{\hspost}@{}}%
\>[B]{}\Varid{hSt}{}\<[6]%
\>[6]{}\mathbin{:}{}\<[6E]%
\>[9]{}\mathbin{∀}\Varid{α}\;\Varid{r}\;\Varid{r}^{\prime}\!.\!\;\Conid{S}\mathbin{→}\ann{\mathbin{⟨}\Conid{St},\Varid{r}\mathbin{⟩}\mathbin{*}\Varid{r}^{\prime}}{\susp{\Varid{α}}}{}\<[E]%
\\
\>[6]{}\mathbin{→}{}\<[6E]%
\>[9]{}\ann{\Varid{r}\mathbin{*}\mathbin{⟨}\Conid{St},\Varid{r}^{\prime}\mathbin{⟩}}{\susp{(\Varid{α},\Conid{S})}}{}\<[E]%
\\
\>[B]{}\Varid{hSt}\;\Varid{s}_{0}\;\Varid{prog}\mathrel{=}{}\<[E]%
\\
\>[B]{}\hsindent{3}{}\<[3]%
\>[3]{}\kw{handle}\;\{\mskip1.5mu {}\<[E]%
\\
\>[3]{}\hsindent{2}{}\<[5]%
\>[5]{}\op{get}\;{}\<[13]%
\>[13]{}\Varid{s}\;{}\<[16]%
\>[16]{}\Varid{k}{}\<[19]%
\>[19]{}\mathbin{↦}\Varid{k}\;\Varid{s}\;\susp{\Varid{s}},{}\<[E]%
\\
\>[3]{}\hsindent{2}{}\<[5]%
\>[5]{}\op{put}\;\Varid{s'}\;{}\<[13]%
\>[13]{}\Varid{s}\;{}\<[16]%
\>[16]{}\Varid{k}{}\<[19]%
\>[19]{}\mathbin{↦}\Varid{k}\;\Varid{s'}\;\susp{()},{}\<[E]%
\\
\>[3]{}\hsindent{2}{}\<[5]%
\>[5]{}\kw{return}\;\Varid{x}\;\Varid{s}{}\<[19]%
\>[19]{}\mathbin{↦}\susp{(\Varid{x},\Varid{s})}{}\<[E]%
\\
\>[B]{}\hsindent{3}{}\<[3]%
\>[3]{}\mskip1.5mu\}\;\Varid{s}_{0}\;\Varid{prog}\kw{!}{}\<[E]%
\ColumnHook
\end{hscode}\resethooks
\end{minipage}
\begin{minipage}[t]{0.32\textwidth}
  \begin{hscode}\SaveRestoreHook
\column{B}{@{}>{\hspre}l<{\hspost}@{}}%
\column{3}{@{}>{\hspre}l<{\hspost}@{}}%
\column{5}{@{}>{\hspre}l<{\hspost}@{}}%
\column{9}{@{}>{\hspre}c<{\hspost}@{}}%
\column{9E}{@{}l@{}}%
\column{12}{@{}>{\hspre}l<{\hspost}@{}}%
\column{15}{@{}>{\hspre}l<{\hspost}@{}}%
\column{E}{@{}>{\hspre}l<{\hspost}@{}}%
\>[B]{}\Varid{hAbort}{}\<[9]%
\>[9]{}\mathbin{:}{}\<[9E]%
\>[12]{}\mathbin{∀}\Varid{α}\;\Varid{r}\;\Varid{r}^{\prime}\!.\!\;\ann{\mathbin{⟨}\Conid{Ab},\Varid{r}\mathbin{⟩}\mathbin{*}\Varid{r}^{\prime}}{\susp{\Varid{α}}}{}\<[E]%
\\
\>[9]{}\mathbin{→}{}\<[9E]%
\>[12]{}\ann{\Varid{r}\mathbin{*}\mathbin{⟨}\Conid{Ab},\Varid{r}^{\prime}\mathbin{⟩}}{\susp{\Conid{Maybe}\;\Varid{α}}}{}\<[E]%
\\
\>[B]{}\Varid{hAbort}\;\Varid{prog}\mathrel{=}{}\<[E]%
\\
\>[B]{}\hsindent{3}{}\<[3]%
\>[3]{}\kw{handle}\;\{\mskip1.5mu {}\<[E]%
\\
\>[3]{}\hsindent{2}{}\<[5]%
\>[5]{}\op{abort}\;\anonymous {}\<[15]%
\>[15]{}\mathbin{↦}\susp{\Conid{Nothing}},{}\<[E]%
\\
\>[3]{}\hsindent{2}{}\<[5]%
\>[5]{}\kw{return}\;\Varid{x}{}\<[15]%
\>[15]{}\mathbin{↦}\susp{\Conid{Just}\;\Varid{x}}{}\<[E]%
\\
\>[B]{}\hsindent{3}{}\<[3]%
\>[3]{}\mskip1.5mu\}\;\Varid{prog}\kw{!}{}\<[E]%
\ColumnHook
\end{hscode}\resethooks
\end{minipage}
\begin{minipage}[t]{0.32\textwidth}
  \begin{hscode}\SaveRestoreHook
\column{B}{@{}>{\hspre}l<{\hspost}@{}}%
\column{3}{@{}>{\hspre}l<{\hspost}@{}}%
\column{5}{@{}>{\hspre}l<{\hspost}@{}}%
\column{8}{@{}>{\hspre}c<{\hspost}@{}}%
\column{8E}{@{}l@{}}%
\column{11}{@{}>{\hspre}l<{\hspost}@{}}%
\column{17}{@{}>{\hspre}l<{\hspost}@{}}%
\column{E}{@{}>{\hspre}l<{\hspost}@{}}%
\>[B]{}\Varid{hRead}{}\<[8]%
\>[8]{}\mathbin{:}{}\<[8E]%
\>[11]{}\mathbin{∀}\Varid{α}\;\Varid{r}\;\Varid{r}^{\prime}\!.\!\;\Conid{R}\mathbin{→}\ann{\mathbin{⟨}\Conid{Rd},\Varid{r}\mathbin{⟩}\mathbin{*}\Varid{r}^{\prime}}{\susp{\Varid{α}}}{}\<[E]%
\\
\>[8]{}\mathbin{→}{}\<[8E]%
\>[11]{}\ann{\Varid{r}\mathbin{*}\mathbin{⟨}\Conid{Rd},\Varid{r}^{\prime}\mathbin{⟩}}{\susp{\Varid{α}}}{}\<[E]%
\\
\>[B]{}\Varid{hRead}\;\Varid{r}\;\Varid{prog}\mathrel{=}{}\<[E]%
\\
\>[B]{}\hsindent{3}{}\<[3]%
\>[3]{}\kw{handle}\;\{\mskip1.5mu {}\<[E]%
\\
\>[3]{}\hsindent{2}{}\<[5]%
\>[5]{}\op{read}\;\Varid{r}\;\Varid{k}{}\<[17]%
\>[17]{}\mathbin{↦}\Varid{k}\;\Varid{r}\;\susp{\Varid{r}},{}\<[E]%
\\
\>[3]{}\hsindent{2}{}\<[5]%
\>[5]{}\kw{return}\;\Varid{x}\;\anonymous {}\<[17]%
\>[17]{}\mathbin{↦}\susp{\Varid{x}}{}\<[E]%
\\
\>[B]{}\hsindent{3}{}\<[3]%
\>[3]{}\mskip1.5mu\}\;\Varid{r}\;\Varid{prog}\kw{!}{}\<[E]%
\ColumnHook
\end{hscode}\resethooks
\end{minipage}
\renewcommand{\hscodestyle}{}
}
  \caption{Handlers for common algebraic effects defined using handlers of higher-order effects}
  \label{fig:alg-eff-handlers}
\end{figure}

\begin{lemma}[Handler Return Lemmas for \Cref{fig:alg-eff-handlers}]
  Each of the handlers in \cref{fig:alg-eff-handlers} satisfy the handler return lemma (\cref{lem:handler-return}), using the following baggable functor definitions:\\
  \begin{minipage}[t]{0.32\textwidth}
    \begin{align*}
      \ensuremath{\Varid{toBag}_{(,)\;\Conid{S}}\ (\Varid{x},\Varid{s})} &= \ensuremath{\{\mskip1.5mu \Varid{x}\mskip1.5mu\}}
    \end{align*}
  \end{minipage}
  \begin{minipage}[t]{0.32\textwidth}
    \begin{align*}
      \ensuremath{\Varid{toBag}_{\Conid{Maybe}}\ (\Conid{Just}\;\Varid{x})} &= \ensuremath{\{\mskip1.5mu \Varid{x}\mskip1.5mu\}}
      \\
      \ensuremath{\Varid{toBag}_{\Conid{Maybe}}\ \Conid{Nothing}}    &= \ensuremath{\{\mskip1.5mu \mskip1.5mu\}}
    \end{align*}
  \end{minipage}
  \begin{minipage}[t]{0.32\textwidth}
    \begin{align*}
      \ensuremath{\Varid{toBag}_{\Conid{Id}}\ \Varid{x}} &= \ensuremath{\{\mskip1.5mu \Varid{x}\mskip1.5mu\}}
    \end{align*}
  \end{minipage}
  \begin{proof}
    Follows from the definition of the \ensuremath{\kw{return}} clauses in \cref{fig:alg-eff-handlers} and the $\textit{toBag}$ definitions above.
  \end{proof}
  \label{lem:handler-ret-alg-eff}
\end{lemma}

\begin{lemma}[Handler Op Lemmas for \Cref{fig:alg-eff-handlers}]
  Each of the handlers in \cref{fig:alg-eff-handlers} satisfy the handler op lemma (\cref{lem:handler-op}).
  \begin{proof}[Proof (sketch -- see \cref{sec:soc-state-appendix} for a full proof of \ensuremath{\Varid{hState}})]
    The proofs are by inversion on \ensuremath{ \longtwoheadmapsto }, \ensuremath{ \longmapsto }, and \ensuremath{\mathbin{-->}}.
    The goals follow from the ``induction step premise'' of the op handler lemma (see \cref{lem:handler-op}), and/or from the fact that \ensuremath{\Varid{h}} and \ensuremath{\Varid{h}_2} satisfy the handler return lemmas (\cref{lem:handler-return}).
  \end{proof}
\end{lemma}


\subsection{Separation of Concerns for Three Common Higher-Order Effects}
\label{sec:soc-ho-eff}

\Cref{fig:ho-eff-handlers} shows three handlers of higher-order effects: a \emph{local} effect (akin to Haskell's \ensuremath{\Conid{Reader}} monad), an \emph{exception catching} effect, and an effect with operations for the lambda abstraction effect discussed in \cref{sec:lambda-effect}.

\begin{figure}
{
\renewcommand{\hscodestyle}{\Small}
\begin{minipage}[t]{0.30\textwidth}
  \begin{hscode}\SaveRestoreHook
\column{B}{@{}>{\hspre}l<{\hspost}@{}}%
\column{3}{@{}>{\hspre}l<{\hspost}@{}}%
\column{5}{@{}>{\hspre}l<{\hspost}@{}}%
\column{7}{@{}>{\hspre}l<{\hspost}@{}}%
\column{9}{@{}>{\hspre}c<{\hspost}@{}}%
\column{9E}{@{}l@{}}%
\column{12}{@{}>{\hspre}l<{\hspost}@{}}%
\column{16}{@{}>{\hspre}l<{\hspost}@{}}%
\column{21}{@{}>{\hspre}l<{\hspost}@{}}%
\column{E}{@{}>{\hspre}l<{\hspost}@{}}%
\>[B]{}\Varid{hLocal}{}\<[9]%
\>[9]{}\mathbin{:}{}\<[9E]%
\>[12]{}\mathbin{∀}\Varid{α}\;\Varid{r}\;\Varid{r}^{\prime}\!.\!\;\Conid{R}\mathbin{→}\ann{\mathbin{⟨}\Conid{Lo},\Varid{r}\mathbin{⟩}\mathbin{*}\Varid{r}^{\prime}}{\susp{\Varid{α}}}{}\<[E]%
\\
\>[9]{}\mathbin{→}{}\<[9E]%
\>[12]{}\ann{\Varid{r}\mathbin{*}\mathbin{⟨}\Conid{Lo},\Varid{r}^{\prime}\mathbin{⟩}}{\susp{\Varid{α}}}{}\<[E]%
\\
\>[B]{}\Varid{hLocal}\;\Varid{r}\;\Varid{prog}\mathrel{=}{}\<[E]%
\\
\>[B]{}\hsindent{3}{}\<[3]%
\>[3]{}\kw{handle}\;\{\mskip1.5mu {}\<[E]%
\\
\>[3]{}\hsindent{2}{}\<[5]%
\>[5]{}\op{ask}\;{}\<[16]%
\>[16]{}\Varid{r}\;\Varid{k}{}\<[21]%
\>[21]{}\mathbin{↦}\Varid{k}\;\Varid{r}\;\susp{\Varid{r}}{}\<[E]%
\\
\>[3]{}\hsindent{2}{}\<[5]%
\>[5]{}\op{local}\;\Varid{f}\;\Varid{m}\;{}\<[16]%
\>[16]{}\Varid{r}\;\Varid{k}{}\<[21]%
\>[21]{}\mathbin{↦}{}\<[E]%
\\
\>[5]{}\hsindent{2}{}\<[7]%
\>[7]{}\Varid{k}\;\Varid{r}\;(\Varid{hLocal}\;(\Varid{f}\;\Varid{r})\;\Varid{m}\kw{!}){}\<[E]%
\\
\>[3]{}\hsindent{2}{}\<[5]%
\>[5]{}\kw{return}\;\Varid{x}\;\anonymous {}\<[21]%
\>[21]{}\mathbin{↦}\susp{\Varid{x}}{}\<[E]%
\\
\>[B]{}\hsindent{3}{}\<[3]%
\>[3]{}\mskip1.5mu\}\;\Varid{r}\;\Varid{prog}\kw{!}{}\<[E]%
\ColumnHook
\end{hscode}\resethooks
\end{minipage}
\begin{minipage}[t]{0.31\textwidth}
  \begin{hscode}\SaveRestoreHook
\column{B}{@{}>{\hspre}l<{\hspost}@{}}%
\column{3}{@{}>{\hspre}l<{\hspost}@{}}%
\column{5}{@{}>{\hspre}l<{\hspost}@{}}%
\column{7}{@{}>{\hspre}l<{\hspost}@{}}%
\column{9}{@{}>{\hspre}c<{\hspost}@{}}%
\column{9E}{@{}l@{}}%
\column{12}{@{}>{\hspre}l<{\hspost}@{}}%
\column{15}{@{}>{\hspre}l<{\hspost}@{}}%
\column{18}{@{}>{\hspre}l<{\hspost}@{}}%
\column{30}{@{}>{\hspre}l<{\hspost}@{}}%
\column{E}{@{}>{\hspre}l<{\hspost}@{}}%
\>[B]{}\Varid{hCatch}{}\<[9]%
\>[9]{}\mathbin{:}{}\<[9E]%
\>[12]{}\mathbin{∀}\Varid{α}\;\Varid{r}\;\Varid{r}^{\prime}\!.\!\;\ann{\mathbin{⟨}\Conid{Ca},\Varid{r}\mathbin{⟩}\mathbin{*}\Varid{r}^{\prime}}{\susp{\Varid{α}}}{}\<[E]%
\\
\>[9]{}\mathbin{→}{}\<[9E]%
\>[12]{}\ann{\Varid{r}\mathbin{*}\mathbin{⟨}\Conid{Ca},\Varid{r}^{\prime}\mathbin{⟩}}{\susp{\Conid{Maybe}\;\Varid{α}}}{}\<[E]%
\\
\>[B]{}\Varid{hCatch}\;\Varid{prog}\mathrel{=}{}\<[E]%
\\
\>[B]{}\hsindent{3}{}\<[3]%
\>[3]{}\kw{handle}\;\{\mskip1.5mu {}\<[E]%
\\
\>[3]{}\hsindent{2}{}\<[5]%
\>[5]{}\op{throw}\;{}\<[18]%
\>[18]{}\Varid{k}\mathbin{↦}\susp{\Conid{Nothing}},{}\<[E]%
\\
\>[3]{}\hsindent{2}{}\<[5]%
\>[5]{}\op{catch}\;\Varid{m}_{1}\;\Varid{m}_{2}\;{}\<[18]%
\>[18]{}\Varid{k}\mathbin{↦}{}\<[E]%
\\
\>[5]{}\hsindent{2}{}\<[7]%
\>[7]{}\Varid{k}\;()\;\{\mskip1.5mu {}\<[15]%
\>[15]{}\kw{match}\;(\Varid{hCatch}\;\Varid{m}_{1})\kw{!}{}\<[E]%
\\
\>[15]{}\mid \Conid{Nothing}\to \Varid{m}_{2}\kw{!}{}\<[E]%
\\
\>[15]{}\mid \Conid{Just}\;\Varid{x}\to \Varid{x}{}\<[30]%
\>[30]{}\mskip1.5mu\},{}\<[E]%
\\
\>[3]{}\hsindent{2}{}\<[5]%
\>[5]{}\kw{return}\;\Varid{x}\mathbin{↦}\susp{\Conid{Just}\;\Varid{x}}{}\<[E]%
\\
\>[B]{}\hsindent{3}{}\<[3]%
\>[3]{}\mskip1.5mu\}\;\Varid{prog}\kw{!}{}\<[E]%
\ColumnHook
\end{hscode}\resethooks
\end{minipage}
\begin{minipage}[t]{0.35\textwidth}
  \begin{hscode}\SaveRestoreHook
\column{B}{@{}>{\hspre}l<{\hspost}@{}}%
\column{3}{@{}>{\hspre}l<{\hspost}@{}}%
\column{5}{@{}>{\hspre}l<{\hspost}@{}}%
\column{6}{@{}>{\hspre}c<{\hspost}@{}}%
\column{6E}{@{}l@{}}%
\column{7}{@{}>{\hspre}l<{\hspost}@{}}%
\column{8}{@{}>{\hspre}c<{\hspost}@{}}%
\column{8E}{@{}l@{}}%
\column{9}{@{}>{\hspre}l<{\hspost}@{}}%
\column{11}{@{}>{\hspre}l<{\hspost}@{}}%
\column{12}{@{}>{\hspre}l<{\hspost}@{}}%
\column{18}{@{}>{\hspre}l<{\hspost}@{}}%
\column{23}{@{}>{\hspre}l<{\hspost}@{}}%
\column{24}{@{}>{\hspre}c<{\hspost}@{}}%
\column{24E}{@{}l@{}}%
\column{26}{@{}>{\hspre}l<{\hspost}@{}}%
\column{E}{@{}>{\hspre}l<{\hspost}@{}}%
\>[B]{}\Varid{hLam}{}\<[8]%
\>[8]{}\mathbin{:}{}\<[8E]%
\>[11]{}\mathbin{∀}\Varid{α}\;\Varid{r}\;\Varid{r}^{\prime}\!.\!\;\Conid{Env}\mathbin{→}\ann{\mathbin{⟨}\Conid{La},\mathbin{⟨}\Conid{Ab},\Varid{r}\mathbin{⟩}\mathbin{⟩}\mathbin{*}\Varid{r}^{\prime}}{\susp{\Varid{α}}}{}\<[E]%
\\
\>[8]{}\mathbin{→}{}\<[8E]%
\>[11]{}\ann{\mathbin{⟨}\Conid{Ab},\Varid{r}\mathbin{⟩}\mathbin{*}\mathbin{⟨}\Conid{La},\Varid{r}^{\prime}\mathbin{⟩}}{\susp{\Varid{α}}}{}\<[E]%
\\
\>[B]{}\Varid{hLam}\;\Varid{r}\;\Varid{prog}\mathrel{=}{}\<[E]%
\\
\>[B]{}\hsindent{3}{}\<[3]%
\>[3]{}\kw{handle}\;\{\mskip1.5mu {}\<[E]%
\\
\>[3]{}\hsindent{2}{}\<[5]%
\>[5]{}\op{abs}\;\Varid{x}\;\Varid{m}\;{}\<[18]%
\>[18]{}\Varid{r}\;\Varid{k}{}\<[23]%
\>[23]{}\mathbin{↦}\Varid{k}\;\Varid{r}\;\susp{\Conid{Clo}\;\Varid{x}\;\Varid{m}\;\Varid{r}},{}\<[E]%
\\
\>[3]{}\hsindent{2}{}\<[5]%
\>[5]{}\op{app}\;\Varid{v}_{1}\;\Varid{v}_{2}\;{}\<[18]%
\>[18]{}\Varid{r}\;\Varid{k}{}\<[23]%
\>[23]{}\mathbin{↦}\kw{match}\;\Varid{v}_{1}{}\<[E]%
\\
\>[5]{}\hsindent{2}{}\<[7]%
\>[7]{}\mid (\Conid{Clo}\;\Varid{x}\;\Varid{m}\;\Varid{r}^{\prime}){}\<[24]%
\>[24]{}\to {}\<[24E]%
\\
\>[7]{}\hsindent{2}{}\<[9]%
\>[9]{}\Varid{k}\;(\Varid{hLam}\;((\Varid{x},\Varid{v}_{2})\mathbin{::}\Varid{r}^{\prime})\;\Varid{m}){}\<[E]%
\\
\>[5]{}\hsindent{2}{}\<[7]%
\>[7]{}\mid \anonymous {}\<[26]%
\>[26]{}\to \op{abort}\kw{!},{}\<[E]%
\\
\>[3]{}\hsindent{2}{}\<[5]%
\>[5]{}\op{var}\;\Varid{x}\;{}\<[18]%
\>[18]{}\Varid{r}\;\Varid{k}{}\<[23]%
\>[23]{}\mathbin{↦}\Varid{k}\;\Varid{r}\;\susp{\Varid{lookup}\;\Varid{x}\;\Varid{r}},{}\<[E]%
\\
\>[3]{}\hsindent{2}{}\<[5]%
\>[5]{}\kw{return}\;\Varid{x}\;\Varid{r}{}\<[23]%
\>[23]{}\mathbin{↦}\susp{\Varid{x}}{}\<[E]%
\\
\>[B]{}\hsindent{3}{}\<[3]%
\>[3]{}\mskip1.5mu\}\;\Varid{r}\;\Varid{prog}\kw{!}{}\<[E]%
\\[\blanklineskip]%
\>[B]{}\mathbf{where}{}\<[E]%
\\[\blanklineskip]%
\>[B]{}\Conid{Clo}{}\<[6]%
\>[6]{}\mathbin{:}{}\<[6E]%
\>[12]{}\Conid{String}\mathbin{→}\ann{\mathbin{⟨}\Conid{La},\mathbin{⟨}\Conid{Ab},\Varid{r}\mathbin{⟩}\mathbin{⟩}\mathbin{*}\Varid{r}^{\prime}}{\susp{\Conid{Val}\;\Varid{r}\;\Varid{r}^{\prime}}}{}\<[E]%
\\
\>[6]{}\mathbin{→}{}\<[6E]%
\>[12]{}\Conid{Env}\mathbin{→}\Conid{Val}\;\Varid{r}\;\Varid{r}^{\prime}{}\<[E]%
\ColumnHook
\end{hscode}\resethooks
\end{minipage}
\renewcommand{\hscodestyle}{}
}
\caption{Handlers for common higher-order effects}
\label{fig:ho-eff-handlers}
\end{figure}

\begin{lemma}[Handler Return Lemmas for \cref{fig:ho-eff-handlers}]
  Each of the handlers in \cref{fig:ho-eff-handlers} satisfy the handler return lemma (\cref{lem:handler-return}), using the $\mathit{toBag}_{\mathit{Id}}$ and $\mathit{toBag}_{\mathit{Maybe}}$ definitions from \cref{lem:handler-ret-alg-eff}.
  \begin{proof}
    Follows from the definition of the \ensuremath{\kw{return}} clauses in \cref{fig:ho-eff-handlers} and the $\mathit{toBag}$ definitions from \cref{lem:handler-ret-alg-eff}.
  \end{proof}
\end{lemma}

\begin{lemma}[Handler Op Lemmas for \ensuremath{\Varid{hLocal}} and \ensuremath{\Varid{hCatch}}]
  Each of the handlers in \cref{fig:ho-eff-handlers} satisfy the handler op lemma (\cref{lem:handler-op}).
  \begin{proof}[Proof (sketch -- see \cref{app:handler-op-catch} for full proof of \ensuremath{\Varid{hCatch}})]
    The proof is by inversion on \ensuremath{ \longtwoheadmapsto }, \ensuremath{ \longmapsto }, and \ensuremath{\mathbin{-->}}.
    The goals follow from the ``induction step premise'' of the op handler lemma (see \cref{lem:handler-op}), and/or from the fact that \ensuremath{\Varid{h}} and \ensuremath{\Varid{h}_2} satisfy the handler return lemmas (\cref{lem:handler-return}).
  \end{proof}
\end{lemma}

\subsection{Handlers with Non-Separable Concerns: Non-Determinism, Revisited}
\label{sec:soc-non-separable}

As argued in \cref{sec:nondet-nonsep}, the non-determinism handler \ensuremath{\Varid{hND}} does not treat the \ensuremath{\op{flip}} operation in a way that it is a separate concern.
Here, we make that argument formal.
To do so, we need to provide a baggable functor instance for \ensuremath{\Conid{List}}s.
A sensible choice is to translate a list into a bag, as follows:
\begin{align*}
  \ensuremath{\Varid{toBag}_{\Conid{List}}\ [\mskip1.5mu \mskip1.5mu]} &= \ensuremath{\{\mskip1.5mu \mskip1.5mu\}} \\
  \ensuremath{\Varid{toBag}_{\Conid{List}}\ (\Varid{x}\mathbin{::}\Varid{xs})} &= \ensuremath{\{\mskip1.5mu \Varid{x}\mskip1.5mu\}\cup\Varid{toBag}_{\Conid{List}}\ \Varid{xs}}
\end{align*}
It is also possible to define $\mathit{toBag}_{\mathit{List}}$ in a way chooses the first or an arbitrary element of a list, and returns a singleton bag with that element.
It is, however, not possible to define $\mathit{toBag}_{\mathit{List}}$ in a way that returns the empty bag, since the handler return lemma would be violated for the \ensuremath{\Varid{hND}} handler.
From the $\mathit{toBag}_{\mathit{List}}$ definition above, the examples we considered in \cref{sec:nondet-nonsep} are counter-examples of the separate concerns property, since:
\begin{equation*}
  \ensuremath{\Varid{toBag}_{\Conid{List}\mathbin{∘}\Conid{Maybe}}\ [\mskip1.5mu \Conid{Nothing},\Conid{Just}\;\mathrm{1}\mskip1.5mu]\mathrel{=}\{\mskip1.5mu \mathrm{1}\mskip1.5mu\}\neq\{\mskip1.5mu \mskip1.5mu\}\mathrel{=}\Varid{toBag}_{\Conid{Maybe}\mathbin{∘}\Conid{List}}\ \Conid{Nothing}}
\end{equation*}


\subsection{Conjectured Criterion for Separation of Concerns}
\label{sec:soc-conjecture}

In this section we observed that many common handlers of higher-order effects have separate concerns, and also observed that others do not.
From the observed examples, we can see that the handlers that call their continuation argument \ensuremath{\Varid{k}} \emph{at most} once---such as \ensuremath{\Varid{hSt}}, \ensuremath{\Varid{hAbort}}, \ensuremath{\Varid{hRead}}, \ensuremath{\Varid{hLocal}}, \ensuremath{\Varid{hCatch}}, and \ensuremath{\Varid{hLam}}---have separate concerns.
We can also see that handlers that may call their continuation argument more than once---such as \ensuremath{\Varid{hND}}, or the ``scoped'' encodings \ensuremath{\Varid{hCatch}_{2}} and \ensuremath{\Varid{hLam}_{1}}---do not have separate concerns.

Based on these observations, we conjecture that handlers that make \emph{affine use} of their continuation argument have separate concerns.
By affine use we mean that the continuation is either called once (like in the \ensuremath{\op{ask}} clause in \cref{fig:ho-eff-handlers}) or not at all (like in the \ensuremath{\op{throw}} clause in \cref{fig:ho-eff-handlers}).

%
%
%

\section{Related Work}
\label{sec:related}

\paragraph{Models of effectful computations} \emph{Monads}, originally by
\citet{DBLP:journals/iandc/Moggi91} and later popularized in the functional
programming community by \citet{DBLP:conf/popl/Wadler92}, have long been the
dominant approach to modelling programs with side effects. They are, however,
famously hard to compose, thus \emph{monad
  transformers}~\cite{DBLP:conf/popl/LiangHJ95} were proposed as a means to
build monads from individual definitions of effects. \emph{Algebraic
  effects}~\cite{DBLP:journals/acs/PlotkinP03} provide a more structured
approach, by specifying effects in terms of \emph{operations} whose behaviour is
characterized by a set of equations that specify when it is well-behaved. Later,
\citet{DBLP:conf/esop/PlotkinP09} extended the approach with \emph{handlers},
which define interpretations of effectful operations by defining a homomorphism
from a \emph{free model} that trivially inhabits the equational theory (i.e.,
syntax) to a programmer-defined domain, making the approach attractive for
implementing effects as well. Perhaps the most well-known implementation of
algebraic effects and handlers is the \emph{free
  monad}~\cite{DBLP:conf/icfp/KammarLO13}, and this implementation is often
taken as the semantic foundation of languages with support for effect
handlers. \citet{DBLP:conf/haskell/SchrijversPWJ19} showed that algebraic
effects implemented using the free monad correspond to a sub-class of
monad-transformers.

It is future work to give a denotational model of \ensuremath{\Varid{λ}^{\Varid{hop}}} in terms of the free monad.
The obvious approach is to interpret suspension types as the free monad of a
signature functor corresponding to its annotation. However, since continuations
receive a computation argument as well, the usual inductive definition would
necessarily have a contravariant recursive occurence, meaning it lacks a
functorial semantics.  While it is possible to define such a type in, e.g.,
Haskell, it would be wrong to market it as a denotational model for \ensuremath{\Varid{λ}^{\Varid{hop}}}.


\paragraph{Models of effects with higher-order operations}

A crucial difference between our approach and plain algebraic effects is the
built-in support for \emph{higher-order operations}. Though it is possible to
encode such operations with algebraic effects, this usually introduces at least
some mixing of operations and handlers. For example,
\citet{DBLP:conf/esop/PlotkinP09} implement a \ensuremath{\op{catch}} operation in terms of
the exception handler. \emph{Scoped Effects}~\cite{DBLP:conf/haskell/WuSH14}
were proposed as an alternative flavor of algebraic effects that syntactically
supports higher-operations, recovering the separation between syntax
semantics. In subsequent work, \citet{DBLP:conf/lics/PirogSWJ18} adapted the
categorical formulation of algebraic effects to Scoped Effects, giving them the
same formal underpinning.

Higher-order operations in scoped effects are not unlike operations with
suspension parameters in \ensuremath{\Varid{λ}^{\Varid{hop}}}, although sub-computations must have the exact
same effects as their parent operation. In \ensuremath{\Varid{λ}^{\Varid{hop}}}, suspension parameters may be
annotated with any effect. In scoped effects one can recover multiple different
\emph{interactions} for effects by changing the order in which they are
handled. In \ensuremath{\Varid{λ}^{\Varid{hop}}}, on the other hand, strives for \emph{separation of concern}:
effects should interact as little as possible to make their semantics more
predictable. Consequently, interaction semantics have to be explicitly
defined. Another key difference between scoped effects and our approach is the
treatment of effects in sub-computations. In \ensuremath{\Varid{λ}^{\Varid{hop}}}, the syntax of
suspended computations is left intact until they are enacted. In scoped effects,
handlers are eagerly applied to sub-computations, replacing operations with
their semantics. As a result, handlers for higher-order operations necessarily
have to mediate the answer type modifications of previously-applied handlers
when enacting a sub-computation. This works well for polymorphic operations like
\ensuremath{\op{catch}}, where the semantic traces of earlier handlers can be hidden in the
existentially quantified type of its branches, but operations with
sub-computations that return a known type---such as \ensuremath{\op{abs}}---are awkward to
capture.

\emph{Latent effects} were developed by \citet{DBLP:conf/aplas/BergSPW21} as a
refinement of scoped effects that solves this issue by explicitly tracking such
semantic residues using a \emph{latent effect functor}, allowing for a more
fine-grained specification of the types of sub-computations. With their
approach, it is possible to capture function abstraction as a higher-order
operation. The latent effect functor serves an equivalent purpose as our latent
effect annotations, though latent effects remembers a semantic imprint of the
already-handled effects while we remember their syntax.
\paragraph{Implementations of Algebraic Effects and Handlers}

There is an abundance of languages with support for algebraic effects and
handlers. Perhaps closest to our work is Koka~\cite{DBLP:conf/popl/Leijen17},
whose core calculus features a similar Hindley/Milner-style row-polymorphic type
system as \ensuremath{\Varid{λ}^{\Varid{hop}}}. While we maintain a
CBPV-inspired~\cite{DBLP:books/sp/Levy2004} distinction between computations and
values, Koka is purely call-by-value, equating computations to
functions. Frank~\cite{DBLP:journals/jfp/ConventLMM20}, on the other hand, does
maintain this distinction. Its type system, however, is quite different from
ours, relying on an \emph{ambient ability} and implicit row polymorphism to type
effectful operations. Handlers are not first-class constructs in Frank. Instead,
functions may adjust the ambient ability of their arguments by specifying the
behaviour of operations. This provides some additional flexibility over
conventional handers, permitting for example \emph{multihandlers} that handle
multiple effects at once.  Both Koka and Frank lack native support for higher
order effects, thus encoding higher-order operations requires a similar kind of
inlining of handlers as discussed in \cref{sec:immediate-latent}.

Eff~\cite{DBLP:journals/jlp/BauerP15} is a functional language with support for
algebraic effects and handlers, with the possibility to dynamically generate new
operations and effects. In later work, \citet{DBLP:journals/corr/BauerP13}
developed a type-and-effect system for Eff, together with an inference
algorithm~\cite{DBLP:journals/corr/Pretnar13}. The language
Links~\cite{DBLP:conf/tldi/LindleyC12} employs row-typed algebraic effects in
the context of database programming. Their system is based on System-F extended
with effect rows and row polymorphism, and limits effectful computations to
functions similar to Koka. Importantly, their system tracks effects using
R\'emy-style rows~\cite{DBLP:conf/popl/Remy89}, maintaining so-called
\emph{presence types} that can additionally express an effect's absence from a
computation. \citet{DBLP:journals/pacmpl/BrachthauserSO20} presented
\emph{Effekt} as a more practical implementation of effects and handlers, using
\emph{capability based} type system where effect types express a set of
capabilities that a computation requires from its context.

\paragraph{Efficiency}

We have proposed \ensuremath{\Varid{λ}^{\Varid{hop}}} as an approach to handling higher-order effects.
The approach builds on what Lindley\footnote{\url{https://effect-handlers.org/}} and \citet{hillerstrom2021unix} calls \emph{the effect-handler oriented programming paradigm}.
An important aspect of this paradigm (and any programming paradigm) is being able to run programs efficiently.
The question of how to execute effect handlers efficiently is an active area of research.
We summarize some of the main recent results in this area below.
In future work, we intend to explore how to efficiently implement \ensuremath{\Varid{λ}^{\Varid{hop}}} in practice, and the literature we summarize below will be valuable to that end.

Effect handlers evidently~\cite{DBLP:journals/pacmpl/XieBHSL20} shows obtain efficient run times of algebraic effect handlers, by translating effect handlers to \emph{multi-prompt delimited continuations}~\cite{DBLP:conf/popl/Felleisen88,DBLP:conf/fpca/GunterRR95}.
The translation of \citet{DBLP:journals/pacmpl/XieBHSL20} restricts effect handlers to use \emph{scoped resumptions}; i.e., continuation values should not escape their handlers.
They argue that the restriction is mild, and that it mainly rules out ill-behaved programs.

\citet{DBLP:journals/pacmpl/KarachaliasKPS21} show how to compile Eff code into OCaml code in a way that aggressively reduce handler applications.
Their approach is different, but orthogonal, to the capability-passing compilation strategy due to \citet{DBLP:journals/pacmpl/SchusterBO20}.
The compilation approaches of \citet{DBLP:journals/pacmpl/KarachaliasKPS21} and \citet{DBLP:journals/pacmpl/SchusterBO20} both demonstrate speedups that bring the run time efficiency effect handler oriented code very close to the run time of plain, direct style code.
OCaml has also recently been retrofitted~\cite{DBLP:conf/pldi/Sivaramakrishnan21} with built-in support effect handlers, with impressive performance numbers.

\citet{DBLP:conf/mpc/WuS15} demonstrate how to significantly improve the run time efficiency of free monad based algebraic effect handlers in Haskell by \emph{fusing} handlers.
This efficient encoding of effect handler has been adopted in the \texttt{fused-effects} library with support for scope effects in Haskell.\footnote{\url{https://hackage.haskell.org/package/fused-effects}}

\section{Conclusion}

We have presented \ensuremath{\Varid{λ}^{\Varid{hop}}}, a calculus with effect handlers for higher-order effects.
We have shown that it can be used to implement all algebraic effects, common higher-order effects such as exception catching, and even effects with advanced control-flow such as lambda abstraction.
Additionally, \ensuremath{\Varid{λ}^{\Varid{hop}}} enables a definition of handlers for effects with separation of concerns, meaning the order of applying effect handlers does not change the essential results of the computation.
We have formalized this property and proved it for a selection of handlers for common algebraic and higher-order effects.

\bibliography{references}


\begin{thebibliography}{29}


\ifx \showCODEN    \undefined \def \showCODEN     #1{\unskip}     \fi
\ifx \showDOI      \undefined \def \showDOI       #1{#1}\fi
\ifx \showISBNx    \undefined \def \showISBNx     #1{\unskip}     \fi
\ifx \showISBNxiii \undefined \def \showISBNxiii  #1{\unskip}     \fi
\ifx \showISSN     \undefined \def \showISSN      #1{\unskip}     \fi
\ifx \showLCCN     \undefined \def \showLCCN      #1{\unskip}     \fi
\ifx \shownote     \undefined \def \shownote      #1{#1}          \fi
\ifx \showarticletitle \undefined \def \showarticletitle #1{#1}   \fi
\ifx \showURL      \undefined \def \showURL       {\relax}        \fi
\providecommand\bibfield[2]{#2}
\providecommand\bibinfo[2]{#2}
\providecommand\natexlab[1]{#1}
\providecommand\showeprint[2][]{arXiv:#2}

\bibitem[\protect\citeauthoryear{Bauer and Pretnar}{Bauer and Pretnar}{2014}]%
        {DBLP:journals/corr/BauerP13}
\bibfield{author}{\bibinfo{person}{Andrej Bauer} {and} \bibinfo{person}{Matija
  Pretnar}.} \bibinfo{year}{2014}\natexlab{}.
\newblock \showarticletitle{An Effect System for Algebraic Effects and
  Handlers}.
\newblock \bibinfo{journal}{\emph{Log. Methods Comput. Sci.}}
  \bibinfo{volume}{10}, \bibinfo{number}{4} (\bibinfo{year}{2014}).
\newblock
\urldef\tempurl%
\url{https://doi.org/10.2168/LMCS-10(4:9)2014}
\showDOI{\tempurl}


\bibitem[\protect\citeauthoryear{Bauer and Pretnar}{Bauer and Pretnar}{2015}]%
        {DBLP:journals/jlp/BauerP15}
\bibfield{author}{\bibinfo{person}{Andrej Bauer} {and} \bibinfo{person}{Matija
  Pretnar}.} \bibinfo{year}{2015}\natexlab{}.
\newblock \showarticletitle{Programming with algebraic effects and handlers}.
\newblock \bibinfo{journal}{\emph{J. Log. Algebraic Methods Program.}}
  \bibinfo{volume}{84}, \bibinfo{number}{1} (\bibinfo{year}{2015}),
  \bibinfo{pages}{108--123}.
\newblock
\urldef\tempurl%
\url{https://doi.org/10.1016/j.jlamp.2014.02.001}
\showDOI{\tempurl}


\bibitem[\protect\citeauthoryear{Brachth{\"{a}}user, Schuster, and
  Ostermann}{Brachth{\"{a}}user et~al\mbox{.}}{2020}]%
        {DBLP:journals/pacmpl/BrachthauserSO20}
\bibfield{author}{\bibinfo{person}{Jonathan~Immanuel Brachth{\"{a}}user},
  \bibinfo{person}{Philipp Schuster}, {and} \bibinfo{person}{Klaus Ostermann}.}
  \bibinfo{year}{2020}\natexlab{}.
\newblock \showarticletitle{Effects as capabilities: effect handlers and
  lightweight effect polymorphism}.
\newblock \bibinfo{journal}{\emph{Proc. {ACM} Program. Lang.}}
  \bibinfo{volume}{4}, \bibinfo{number}{{OOPSLA}} (\bibinfo{year}{2020}),
  \bibinfo{pages}{126:1--126:30}.
\newblock
\urldef\tempurl%
\url{https://doi.org/10.1145/3428194}
\showDOI{\tempurl}


\bibitem[\protect\citeauthoryear{Convent, Lindley, McBride, and
  McLaughlin}{Convent et~al\mbox{.}}{2020}]%
        {DBLP:journals/jfp/ConventLMM20}
\bibfield{author}{\bibinfo{person}{Lukas Convent}, \bibinfo{person}{Sam
  Lindley}, \bibinfo{person}{Conor McBride}, {and} \bibinfo{person}{Craig
  McLaughlin}.} \bibinfo{year}{2020}\natexlab{}.
\newblock \showarticletitle{Doo bee doo bee doo}.
\newblock \bibinfo{journal}{\emph{J. Funct. Program.}}  \bibinfo{volume}{30}
  (\bibinfo{year}{2020}), \bibinfo{pages}{e9}.
\newblock
\urldef\tempurl%
\url{https://doi.org/10.1017/S0956796820000039}
\showDOI{\tempurl}


\bibitem[\protect\citeauthoryear{Felleisen}{Felleisen}{1988}]%
        {DBLP:conf/popl/Felleisen88}
\bibfield{author}{\bibinfo{person}{Matthias Felleisen}.}
  \bibinfo{year}{1988}\natexlab{}.
\newblock \showarticletitle{The Theory and Practice of First-Class Prompts}. In
  \bibinfo{booktitle}{\emph{Conference Record of the Fifteenth Annual {ACM}
  Symposium on Principles of Programming Languages, San Diego, California, USA,
  January 10-13, 1988}}, \bibfield{editor}{\bibinfo{person}{Jeanne Ferrante}
  {and} \bibinfo{person}{Peter Mager}} (Eds.). \bibinfo{publisher}{{ACM}
  Press}, \bibinfo{pages}{180--190}.
\newblock
\urldef\tempurl%
\url{https://doi.org/10.1145/73560.73576}
\showDOI{\tempurl}


\bibitem[\protect\citeauthoryear{Gunter, R{\'{e}}my, and Riecke}{Gunter
  et~al\mbox{.}}{1995}]%
        {DBLP:conf/fpca/GunterRR95}
\bibfield{author}{\bibinfo{person}{Carl~A. Gunter}, \bibinfo{person}{Didier
  R{\'{e}}my}, {and} \bibinfo{person}{Jon~G. Riecke}.}
  \bibinfo{year}{1995}\natexlab{}.
\newblock \showarticletitle{A Generalization of Exceptions and Control in
  ML-like Languages}. In \bibinfo{booktitle}{\emph{Proceedings of the seventh
  international conference on Functional programming languages and computer
  architecture, {FPCA} 1995, La Jolla, California, USA, June 25-28, 1995}},
  \bibfield{editor}{\bibinfo{person}{John Williams}} (Ed.).
  \bibinfo{publisher}{{ACM}}, \bibinfo{pages}{12--23}.
\newblock
\showISBNx{0-89791-719-7}
\urldef\tempurl%
\url{https://doi.org/10.1145/224164.224173}
\showDOI{\tempurl}


\bibitem[\protect\citeauthoryear{Hillerstr\"{o}m}{Hillerstr\"{o}m}{2021}]%
        {hillerstrom2021unix}
\bibfield{author}{\bibinfo{person}{Daniel Hillerstr\"{o}m}.}
  \bibinfo{year}{2021}\natexlab{}.
\newblock \bibinfo{title}{Composing UNIX with Effect Handlers: A Case Study in
  Effect Handler Oriented Programming}.
\newblock
\newblock
\urldef\tempurl%
\url{https://dhil.net/research/papers/unix-ml2021.pdf}
\showURL{%
\tempurl}
\newblock
\shownote{Extended abstract, presented at ML Family Workshop 2021.}


\bibitem[\protect\citeauthoryear{Kammar, Lindley, and Oury}{Kammar
  et~al\mbox{.}}{2013}]%
        {DBLP:conf/icfp/KammarLO13}
\bibfield{author}{\bibinfo{person}{Ohad Kammar}, \bibinfo{person}{Sam Lindley},
  {and} \bibinfo{person}{Nicolas Oury}.} \bibinfo{year}{2013}\natexlab{}.
\newblock \showarticletitle{Handlers in action}. In
  \bibinfo{booktitle}{\emph{{ACM} {SIGPLAN} International Conference on
  Functional Programming, ICFP'13, Boston, MA, {USA} - September 25 - 27,
  2013}}, \bibfield{editor}{\bibinfo{person}{Greg Morrisett} {and}
  \bibinfo{person}{Tarmo Uustalu}} (Eds.). \bibinfo{publisher}{{ACM}},
  \bibinfo{pages}{145--158}.
\newblock
\urldef\tempurl%
\url{https://doi.org/10.1145/2500365.2500590}
\showDOI{\tempurl}


\bibitem[\protect\citeauthoryear{Karachalias, Koprivec, Pretnar, and
  Schrijvers}{Karachalias et~al\mbox{.}}{2021}]%
        {DBLP:journals/pacmpl/KarachaliasKPS21}
\bibfield{author}{\bibinfo{person}{Georgios Karachalias},
  \bibinfo{person}{Filip Koprivec}, \bibinfo{person}{Matija Pretnar}, {and}
  \bibinfo{person}{Tom Schrijvers}.} \bibinfo{year}{2021}\natexlab{}.
\newblock \showarticletitle{Efficient compilation of algebraic effect
  handlers}.
\newblock \bibinfo{journal}{\emph{Proc. {ACM} Program. Lang.}}
  \bibinfo{volume}{5}, \bibinfo{number}{{OOPSLA}} (\bibinfo{year}{2021}),
  \bibinfo{pages}{1--28}.
\newblock
\urldef\tempurl%
\url{https://doi.org/10.1145/3485479}
\showDOI{\tempurl}


\bibitem[\protect\citeauthoryear{Leijen}{Leijen}{2017}]%
        {DBLP:conf/popl/Leijen17}
\bibfield{author}{\bibinfo{person}{Daan Leijen}.}
  \bibinfo{year}{2017}\natexlab{}.
\newblock \showarticletitle{Type directed compilation of row-typed algebraic
  effects}. In \bibinfo{booktitle}{\emph{Proceedings of the 44th {ACM}
  {SIGPLAN} Symposium on Principles of Programming Languages, {POPL} 2017,
  Paris, France, January 18-20, 2017}},
  \bibfield{editor}{\bibinfo{person}{Giuseppe Castagna} {and}
  \bibinfo{person}{Andrew~D. Gordon}} (Eds.). \bibinfo{publisher}{{ACM}},
  \bibinfo{pages}{486--499}.
\newblock
\urldef\tempurl%
\url{https://doi.org/10.1145/3009837.3009872}
\showDOI{\tempurl}


\bibitem[\protect\citeauthoryear{Levy}{Levy}{2004}]%
        {DBLP:books/sp/Levy2004}
\bibfield{author}{\bibinfo{person}{Paul~Blain Levy}.}
  \bibinfo{year}{2004}\natexlab{}.
\newblock \bibinfo{booktitle}{\emph{Call-By-Push-Value: {A}
  Functional/Imperative Synthesis}}. \bibinfo{series}{Semantics Structures in
  Computation}, Vol.~\bibinfo{volume}{2}.
\newblock \bibinfo{publisher}{Springer}.
\newblock
\showISBNx{1-4020-1730-8}


\bibitem[\protect\citeauthoryear{Liang, Hudak, and Jones}{Liang
  et~al\mbox{.}}{1995}]%
        {DBLP:conf/popl/LiangHJ95}
\bibfield{author}{\bibinfo{person}{Sheng Liang}, \bibinfo{person}{Paul Hudak},
  {and} \bibinfo{person}{Mark~P. Jones}.} \bibinfo{year}{1995}\natexlab{}.
\newblock \showarticletitle{Monad Transformers and Modular Interpreters}. In
  \bibinfo{booktitle}{\emph{Conference Record of POPL'95: 22nd {ACM}
  {SIGPLAN-SIGACT} Symposium on Principles of Programming Languages, San
  Francisco, California, USA, January 23-25, 1995}},
  \bibfield{editor}{\bibinfo{person}{Ron~K. Cytron} {and}
  \bibinfo{person}{Peter Lee}} (Eds.). \bibinfo{publisher}{{ACM} Press},
  \bibinfo{pages}{333--343}.
\newblock
\urldef\tempurl%
\url{https://doi.org/10.1145/199448.199528}
\showDOI{\tempurl}


\bibitem[\protect\citeauthoryear{Lindley and Cheney}{Lindley and
  Cheney}{2012}]%
        {DBLP:conf/tldi/LindleyC12}
\bibfield{author}{\bibinfo{person}{Sam Lindley} {and} \bibinfo{person}{James
  Cheney}.} \bibinfo{year}{2012}\natexlab{}.
\newblock \showarticletitle{Row-based effect types for database integration}.
  In \bibinfo{booktitle}{\emph{Proceedings of {TLDI} 2012: The Seventh {ACM}
  {SIGPLAN} Workshop on Types in Languages Design and Implementation,
  Philadelphia, PA, USA, Saturday, January 28, 2012}},
  \bibfield{editor}{\bibinfo{person}{Benjamin~C. Pierce}} (Ed.).
  \bibinfo{publisher}{{ACM}}, \bibinfo{pages}{91--102}.
\newblock
\urldef\tempurl%
\url{https://doi.org/10.1145/2103786.2103798}
\showDOI{\tempurl}


\bibitem[\protect\citeauthoryear{Milner}{Milner}{1978}]%
        {DBLP:journals/jcss/Milner78}
\bibfield{author}{\bibinfo{person}{Robin Milner}.}
  \bibinfo{year}{1978}\natexlab{}.
\newblock \showarticletitle{A Theory of Type Polymorphism in Programming}.
\newblock \bibinfo{journal}{\emph{J. Comput. Syst. Sci.}} \bibinfo{volume}{17},
  \bibinfo{number}{3} (\bibinfo{year}{1978}), \bibinfo{pages}{348--375}.
\newblock
\urldef\tempurl%
\url{https://doi.org/10.1016/0022-0000(78)90014-4}
\showDOI{\tempurl}


\bibitem[\protect\citeauthoryear{Moggi}{Moggi}{1991}]%
        {DBLP:journals/iandc/Moggi91}
\bibfield{author}{\bibinfo{person}{Eugenio Moggi}.}
  \bibinfo{year}{1991}\natexlab{}.
\newblock \showarticletitle{Notions of Computation and Monads}.
\newblock \bibinfo{journal}{\emph{Inf. Comput.}} \bibinfo{volume}{93},
  \bibinfo{number}{1} (\bibinfo{year}{1991}), \bibinfo{pages}{55--92}.
\newblock
\urldef\tempurl%
\url{https://doi.org/10.1016/0890-5401(91)90052-4}
\showDOI{\tempurl}


\bibitem[\protect\citeauthoryear{Pir{\'{o}}g, Schrijvers, Wu, and
  Jaskelioff}{Pir{\'{o}}g et~al\mbox{.}}{2018}]%
        {DBLP:conf/lics/PirogSWJ18}
\bibfield{author}{\bibinfo{person}{Maciej Pir{\'{o}}g}, \bibinfo{person}{Tom
  Schrijvers}, \bibinfo{person}{Nicolas Wu}, {and} \bibinfo{person}{Mauro
  Jaskelioff}.} \bibinfo{year}{2018}\natexlab{}.
\newblock \showarticletitle{Syntax and Semantics for Operations with Scopes}.
  In \bibinfo{booktitle}{\emph{Proceedings of the 33rd Annual {ACM/IEEE}
  Symposium on Logic in Computer Science, {LICS} 2018, Oxford, UK, July 09-12,
  2018}}, \bibfield{editor}{\bibinfo{person}{Anuj Dawar} {and}
  \bibinfo{person}{Erich Gr{\"{a}}del}} (Eds.). \bibinfo{publisher}{{ACM}},
  \bibinfo{pages}{809--818}.
\newblock
\urldef\tempurl%
\url{https://doi.org/10.1145/3209108.3209166}
\showDOI{\tempurl}


\bibitem[\protect\citeauthoryear{Plotkin and Power}{Plotkin and Power}{2003}]%
        {DBLP:journals/acs/PlotkinP03}
\bibfield{author}{\bibinfo{person}{Gordon~D. Plotkin} {and}
  \bibinfo{person}{John Power}.} \bibinfo{year}{2003}\natexlab{}.
\newblock \showarticletitle{Algebraic Operations and Generic Effects}.
\newblock \bibinfo{journal}{\emph{Appl. Categorical Struct.}}
  \bibinfo{volume}{11}, \bibinfo{number}{1} (\bibinfo{year}{2003}),
  \bibinfo{pages}{69--94}.
\newblock
\urldef\tempurl%
\url{https://doi.org/10.1023/A:1023064908962}
\showDOI{\tempurl}


\bibitem[\protect\citeauthoryear{Plotkin and Pretnar}{Plotkin and
  Pretnar}{2009}]%
        {DBLP:conf/esop/PlotkinP09}
\bibfield{author}{\bibinfo{person}{Gordon~D. Plotkin} {and}
  \bibinfo{person}{Matija Pretnar}.} \bibinfo{year}{2009}\natexlab{}.
\newblock \showarticletitle{Handlers of Algebraic Effects}. In
  \bibinfo{booktitle}{\emph{Programming Languages and Systems, 18th European
  Symposium on Programming, {ESOP} 2009, Held as Part of the Joint European
  Conferences on Theory and Practice of Software, {ETAPS} 2009, York, UK, March
  22-29, 2009. Proceedings}} \emph{(\bibinfo{series}{Lecture Notes in Computer
  Science})}, \bibfield{editor}{\bibinfo{person}{Giuseppe Castagna}} (Ed.),
  Vol.~\bibinfo{volume}{5502}. \bibinfo{publisher}{Springer},
  \bibinfo{pages}{80--94}.
\newblock
\urldef\tempurl%
\url{https://doi.org/10.1007/978-3-642-00590-9\_7}
\showDOI{\tempurl}


\bibitem[\protect\citeauthoryear{Pretnar}{Pretnar}{2014}]%
        {DBLP:journals/corr/Pretnar13}
\bibfield{author}{\bibinfo{person}{Matija Pretnar}.}
  \bibinfo{year}{2014}\natexlab{}.
\newblock \showarticletitle{Inferring Algebraic Effects}.
\newblock \bibinfo{journal}{\emph{Log. Methods Comput. Sci.}}
  \bibinfo{volume}{10}, \bibinfo{number}{3} (\bibinfo{year}{2014}).
\newblock
\urldef\tempurl%
\url{https://doi.org/10.2168/LMCS-10(3:21)2014}
\showDOI{\tempurl}


\bibitem[\protect\citeauthoryear{R{\'{e}}my}{R{\'{e}}my}{1989}]%
        {DBLP:conf/popl/Remy89}
\bibfield{author}{\bibinfo{person}{Didier R{\'{e}}my}.}
  \bibinfo{year}{1989}\natexlab{}.
\newblock \showarticletitle{Typechecking Records and Variants in a Natural
  Extension of {ML}}. In \bibinfo{booktitle}{\emph{Conference Record of the
  Sixteenth Annual {ACM} Symposium on Principles of Programming Languages,
  Austin, Texas, USA, January 11-13, 1989}}. \bibinfo{publisher}{{ACM} Press},
  \bibinfo{pages}{77--88}.
\newblock
\urldef\tempurl%
\url{https://doi.org/10.1145/75277.75284}
\showDOI{\tempurl}


\bibitem[\protect\citeauthoryear{Schrijvers, Pir{\'{o}}g, Wu, and
  Jaskelioff}{Schrijvers et~al\mbox{.}}{2019}]%
        {DBLP:conf/haskell/SchrijversPWJ19}
\bibfield{author}{\bibinfo{person}{Tom Schrijvers}, \bibinfo{person}{Maciej
  Pir{\'{o}}g}, \bibinfo{person}{Nicolas Wu}, {and} \bibinfo{person}{Mauro
  Jaskelioff}.} \bibinfo{year}{2019}\natexlab{}.
\newblock \showarticletitle{Monad transformers and modular algebraic effects:
  what binds them together}. In \bibinfo{booktitle}{\emph{Proceedings of the
  12th {ACM} {SIGPLAN} International Symposium on Haskell, Haskell@ICFP 2019,
  Berlin, Germany, August 18-23, 2019}},
  \bibfield{editor}{\bibinfo{person}{Richard~A. Eisenberg}} (Ed.).
  \bibinfo{publisher}{{ACM}}, \bibinfo{pages}{98--113}.
\newblock
\urldef\tempurl%
\url{https://doi.org/10.1145/3331545.3342595}
\showDOI{\tempurl}


\bibitem[\protect\citeauthoryear{Schuster, Brachth{\"{a}}user, and
  Ostermann}{Schuster et~al\mbox{.}}{2020}]%
        {DBLP:journals/pacmpl/SchusterBO20}
\bibfield{author}{\bibinfo{person}{Philipp Schuster},
  \bibinfo{person}{Jonathan~Immanuel Brachth{\"{a}}user}, {and}
  \bibinfo{person}{Klaus Ostermann}.} \bibinfo{year}{2020}\natexlab{}.
\newblock \showarticletitle{Compiling effect handlers in capability-passing
  style}.
\newblock \bibinfo{journal}{\emph{Proc. {ACM} Program. Lang.}}
  \bibinfo{volume}{4}, \bibinfo{number}{{ICFP}} (\bibinfo{year}{2020}),
  \bibinfo{pages}{93:1--93:28}.
\newblock
\urldef\tempurl%
\url{https://doi.org/10.1145/3408975}
\showDOI{\tempurl}


\bibitem[\protect\citeauthoryear{Sivaramakrishnan, Dolan, White, Kelly, Jaffer,
  and Madhavapeddy}{Sivaramakrishnan et~al\mbox{.}}{2021}]%
        {DBLP:conf/pldi/Sivaramakrishnan21}
\bibfield{author}{\bibinfo{person}{K.~C. Sivaramakrishnan},
  \bibinfo{person}{Stephen Dolan}, \bibinfo{person}{Leo White},
  \bibinfo{person}{Tom Kelly}, \bibinfo{person}{Sadiq Jaffer}, {and}
  \bibinfo{person}{Anil Madhavapeddy}.} \bibinfo{year}{2021}\natexlab{}.
\newblock \showarticletitle{Retrofitting effect handlers onto OCaml}. In
  \bibinfo{booktitle}{\emph{{PLDI} '21: 42nd {ACM} {SIGPLAN} International
  Conference on Programming Language Design and Implementation, Virtual Event,
  Canada, June 20-25, 2021}}, \bibfield{editor}{\bibinfo{person}{Stephen~N.
  Freund} {and} \bibinfo{person}{Eran Yahav}} (Eds.).
  \bibinfo{publisher}{{ACM}}, \bibinfo{pages}{206--221}.
\newblock
\showISBNx{978-1-4503-8391-2}
\urldef\tempurl%
\url{https://doi.org/10.1145/3453483.3454039}
\showDOI{\tempurl}


\bibitem[\protect\citeauthoryear{van~den Berg, Schrijvers, Poulsen, and
  Wu}{van~den Berg et~al\mbox{.}}{2021}]%
        {DBLP:conf/aplas/BergSPW21}
\bibfield{author}{\bibinfo{person}{Birthe van~den Berg}, \bibinfo{person}{Tom
  Schrijvers}, \bibinfo{person}{Casper~Bach Poulsen}, {and}
  \bibinfo{person}{Nicolas Wu}.} \bibinfo{year}{2021}\natexlab{}.
\newblock \showarticletitle{Latent Effects for Reusable Language Components}.
  In \bibinfo{booktitle}{\emph{Programming Languages and Systems - 19th Asian
  Symposium, {APLAS} 2021, Chicago, IL, USA, October 17-18, 2021, Proceedings}}
  \emph{(\bibinfo{series}{Lecture Notes in Computer Science})},
  \bibfield{editor}{\bibinfo{person}{Hakjoo Oh}} (Ed.),
  Vol.~\bibinfo{volume}{13008}. \bibinfo{publisher}{Springer},
  \bibinfo{pages}{182--201}.
\newblock
\urldef\tempurl%
\url{https://doi.org/10.1007/978-3-030-89051-3\_11}
\showDOI{\tempurl}


\bibitem[\protect\citeauthoryear{Wadler}{Wadler}{1992}]%
        {DBLP:conf/popl/Wadler92}
\bibfield{author}{\bibinfo{person}{Philip Wadler}.}
  \bibinfo{year}{1992}\natexlab{}.
\newblock \showarticletitle{The Essence of Functional Programming}. In
  \bibinfo{booktitle}{\emph{Conference Record of the Nineteenth Annual {ACM}
  {SIGPLAN-SIGACT} Symposium on Principles of Programming Languages,
  Albuquerque, New Mexico, USA, January 19-22, 1992}},
  \bibfield{editor}{\bibinfo{person}{Ravi Sethi}} (Ed.).
  \bibinfo{publisher}{{ACM} Press}, \bibinfo{pages}{1--14}.
\newblock
\urldef\tempurl%
\url{https://doi.org/10.1145/143165.143169}
\showDOI{\tempurl}


\bibitem[\protect\citeauthoryear{Wright and Felleisen}{Wright and
  Felleisen}{1994}]%
        {DBLP:journals/iandc/WrightF94}
\bibfield{author}{\bibinfo{person}{Andrew~K. Wright} {and}
  \bibinfo{person}{Matthias Felleisen}.} \bibinfo{year}{1994}\natexlab{}.
\newblock \showarticletitle{A Syntactic Approach to Type Soundness}.
\newblock \bibinfo{journal}{\emph{Inf. Comput.}} \bibinfo{volume}{115},
  \bibinfo{number}{1} (\bibinfo{year}{1994}), \bibinfo{pages}{38--94}.
\newblock
\urldef\tempurl%
\url{https://doi.org/10.1006/inco.1994.1093}
\showDOI{\tempurl}


\bibitem[\protect\citeauthoryear{Wu and Schrijvers}{Wu and Schrijvers}{2015}]%
        {DBLP:conf/mpc/WuS15}
\bibfield{author}{\bibinfo{person}{Nicolas Wu} {and} \bibinfo{person}{Tom
  Schrijvers}.} \bibinfo{year}{2015}\natexlab{}.
\newblock \showarticletitle{Fusion for Free - Efficient Algebraic Effect
  Handlers}. In \bibinfo{booktitle}{\emph{Mathematics of Program Construction -
  12th International Conference, {MPC} 2015, K{\"{o}}nigswinter, Germany, June
  29 - July 1, 2015. Proceedings}} \emph{(\bibinfo{series}{Lecture Notes in
  Computer Science})}, \bibfield{editor}{\bibinfo{person}{Ralf Hinze} {and}
  \bibinfo{person}{Janis Voigtl{\"{a}}nder}} (Eds.),
  Vol.~\bibinfo{volume}{9129}. \bibinfo{publisher}{Springer},
  \bibinfo{pages}{302--322}.
\newblock
\urldef\tempurl%
\url{https://doi.org/10.1007/978-3-319-19797-5\_15}
\showDOI{\tempurl}


\bibitem[\protect\citeauthoryear{Wu, Schrijvers, and Hinze}{Wu
  et~al\mbox{.}}{2014}]%
        {DBLP:conf/haskell/WuSH14}
\bibfield{author}{\bibinfo{person}{Nicolas Wu}, \bibinfo{person}{Tom
  Schrijvers}, {and} \bibinfo{person}{Ralf Hinze}.}
  \bibinfo{year}{2014}\natexlab{}.
\newblock \showarticletitle{Effect handlers in scope}. In
  \bibinfo{booktitle}{\emph{Proceedings of the 2014 {ACM} {SIGPLAN} symposium
  on Haskell, Gothenburg, Sweden, September 4-5, 2014}},
  \bibfield{editor}{\bibinfo{person}{Wouter Swierstra}} (Ed.).
  \bibinfo{publisher}{{ACM}}, \bibinfo{pages}{1--12}.
\newblock
\urldef\tempurl%
\url{https://doi.org/10.1145/2633357.2633358}
\showDOI{\tempurl}


\bibitem[\protect\citeauthoryear{Xie, Brachth{\"{a}}user, Hillerstr{\"{o}}m,
  Schuster, and Leijen}{Xie et~al\mbox{.}}{2020}]%
        {DBLP:journals/pacmpl/XieBHSL20}
\bibfield{author}{\bibinfo{person}{Ningning Xie},
  \bibinfo{person}{Jonathan~Immanuel Brachth{\"{a}}user},
  \bibinfo{person}{Daniel Hillerstr{\"{o}}m}, \bibinfo{person}{Philipp
  Schuster}, {and} \bibinfo{person}{Daan Leijen}.}
  \bibinfo{year}{2020}\natexlab{}.
\newblock \showarticletitle{Effect handlers, evidently}.
\newblock \bibinfo{journal}{\emph{Proc. {ACM} Program. Lang.}}
  \bibinfo{volume}{4}, \bibinfo{number}{{ICFP}} (\bibinfo{year}{2020}),
  \bibinfo{pages}{99:1--99:29}.
\newblock
\urldef\tempurl%
\url{https://doi.org/10.1145/3408981}
\showDOI{\tempurl}


\end{thebibliography}

\appendix

\section{Proof Sketch of Subject Reduction and Progress Lemmas}
\label{sec:proof-appendix}

\subsection{Subject Reduction}

\begin{lemma}[Subject Reduction/Contraction]
  For all expressions \ensuremath{e_{1}} and \ensuremath{e_{2}}, if \ensuremath{\Conid{Γ}\mathbin{⊢}e_{1}\mathbin{:}\Varid{τ}\mathbin{∣}\Varid{ε}\mathbin{*}\epsilon_{l}} and \ensuremath{e_{1}\mathbin{-->}e_{2}},
  it follows that \ensuremath{\Conid{Γ}\mathbin{⊢}e_{2}\mathbin{:}\Varid{τ}\mathbin{∣}\Varid{ε}\mathbin{*}\epsilon_{l}}. 
\end{lemma}

\begin{proof}
  The proof proceeds by rule induction over the derivation \ensuremath{\Conid{Γ}\mathbin{⊢}e_{1}\mathbin{:}\Varid{τ}\mathbin{∣}\Varid{ε}\mathbin{*}\epsilon_{l}}.
  \begin{case*}[T-Var,T-Abs,T-Suspend]
    Impossible, \ensuremath{e_{1}\mathbin{-->}e_{2}} is empty. 
  \end{case*}
  \begin{case*}[T-App,T-Let] 
    $\ensuremath{e_{1}\mathbin{-->}e_{2}\mathrel{=}(\Varid{λ}\;\Varid{x}}.\ensuremath{\Varid{e})\;\Varid{v}}$. From the sub-derivations of \textsc{T-App},
    from which ends in \textsc{T-Abs}, we derive $\ensuremath{\Conid{Γ},\Varid{x}\mathbin{:}\Varid{τ}^{\prime}\mathbin{⊢}\Varid{e}\mathbin{:}\Varid{τ}\mathbin{∣}\Varid{ε}\mathbin{*}\epsilon_{l}}$
    and $\ensuremath{\Conid{Γ}\mathbin{⊢}\Varid{v}\mathbin{:}\Varid{τ}^{\prime}\mathbin{∣}\Varid{ε}\mathbin{*}\epsilon_{l}}$. Consequently, from \cref{lem:substitution} it
    follows that \ensuremath{\Conid{Γ}\mathbin{⊢}\Varid{e}}$[x/\ensuremath{\Varid{v}}]\ $\ensuremath{\mathbin{:}\Varid{τ}\mathbin{∣}\Varid{ε}\mathbin{*}\epsilon_{l}}. The case for \textsc{Let}
    uses similar reasoning. 
  \end{case*}
  \begin{case*}[T-Inst,T-Gen,T-Resurface] 
    Follow immediately from the induction hypothesis. 
  \end{case*}
  \begin{case*}[T-Enact]
    $\ensuremath{e_{1}\mathbin{-->}e_{2}\mathrel{=}}$\ \susp{e}$\ensuremath{\kw{!}\mathbin{-->}\Varid{e}}$, hence the sub-derivation of
    \textsc{T-Enact} ends with \textsc{T-Suspend}, giving us
    $\ensuremath{\Conid{Γ}\mathbin{⊢}\Varid{e}\mathbin{:}\Varid{τ}\mathbin{∣}\Varid{ε}\mathbin{*}\epsilon_{l}}$. 
  \end{case*}
  \begin{case*}[T-Handle]
    Case analysis on \ensuremath{e_{1}\mathbin{-->}e_{2}} gives us the following two sub-cases:
    \begin{enumerate}

    \item \ensuremath{e_{1}\mathbin{-->}e_{2}\mathrel{=}\kw{handle}^{\ell}\;\{\mskip1.5mu \overline{\Conid{C}}\mskip1.5mu\}\;\Varid{v}\;\Varid{v}_{\Varid{p}}\mathbin{-->}\Varid{e}}$[x/\ensuremath{\Varid{v}},p/\ensuremath{\Varid{v}_{\Varid{p}}}]$
      (\textsc{Return}). From the sub-derivations of \textsc{T-Handle}, we learn
      that $\ensuremath{\Conid{Γ},\Varid{x}\mathbin{:}\Varid{τ}^{\prime},\Varid{p}\mathbin{:}\tau_{p}\mathbin{⊢}\Varid{e}\mathbin{:}}\ann{\ensuremath{\Varid{ε}\mathbin{*}\mathbin{⟨}\Varid{ℓ},\epsilon_{l}\mathbin{⟩}}}{\susp{\ensuremath{\tau_{2}}}}\ensuremath{\mathbin{∣}\Varid{ε}^{\prime}\mathbin{*}\epsilon_{l}^{\prime}}$
      \textbf{(1)}, $\ensuremath{\Conid{Γ}\mathbin{⊢}\Varid{v}\mathbin{:}\tau_{1}\mathbin{∣}\mathbin{⟨}\Varid{ℓ},\Varid{ε}\mathbin{⟩}\mathbin{*}\epsilon_{l}}$ \textbf{(2)}, and
      $\ensuremath{\Conid{Γ}\mathbin{⊢}\Varid{v}_{\Varid{p}}\mathbin{:}\tau_{p}\mathbin{∣}\Varid{ε}^{\prime}\mathbin{*}\epsilon_{l}^{\prime}}$ \textbf{(3)} hold. From \cref{lem:substitution},
      \textbf{(1)}, and \textbf{(3)} we then have
      \ensuremath{\Conid{Γ},\Varid{x}\mathbin{:}\Varid{τ}^{\prime}\mathbin{⊢}\Varid{e}}$[p/\ensuremath{\Varid{v}_{\Varid{p}}}]$\ \ensuremath{\mathbin{:}}\ann{\ensuremath{\Varid{ε}\mathbin{*}\mathbin{⟨}\Varid{ℓ},\epsilon_{l}\mathbin{⟩}}}{\susp{\ensuremath{\tau_{2}}}}\ensuremath{\mathbin{∣}\Varid{ε}^{\prime}\mathbin{*}\epsilon_{l}^{\prime}} \textbf{(4)}. Then, from \cref{lem:substitution},
      \cref{lem:reannotation}, \textbf{(2)}, and \textbf{(4)} we derive
      $\ensuremath{\Conid{Γ}\mathbin{⊢}\Varid{e}}[x/\ensuremath{\Varid{v}},p/\ensuremath{\Varid{v}_{\Varid{p}}}]\ \ensuremath{\mathbin{:}}\ann{\ensuremath{\Varid{ε}\mathbin{*}\mathbin{⟨}\Varid{ℓ},\epsilon_{l}\mathbin{⟩}}}{\susp{\ensuremath{\tau_{2}}}}\ensuremath{\mathbin{∣}\Varid{ε}^{\prime}\mathbin{*}\epsilon_{l}^{\prime}}$.

    \item
      $\ensuremath{e_{1}\mathbin{-->}e_{2}\mathrel{=}\kw{handle}^{\ell}\;\{\mskip1.5mu \overline{\Conid{C}}\mskip1.5mu\}\;\Conid{E}^{\Varid{ℓ}}}[\ensuremath{(\op{op}^{\ell}\;\Varid{f}\;\overline{\Varid{v}})\kw{!}}]\ \ensuremath{\Varid{v}_{\Varid{p}}\mathbin{-->}\Varid{e}}[\bsubst{x}{\ensuremath{\Varid{v}}},p/\ensuremath{\Varid{v}_{\Varid{p}}},k/\ensuremath{(\Varid{λ}\;\Varid{y}\;\Varid{q}}.\ensuremath{\kw{handle}^{\ell}\;\{\mskip1.5mu \overline{\Conid{C}}\mskip1.5mu\}}\ \\ \inE{\ensuremath{\Varid{y}\kw{!}}}\ \ensuremath{\Varid{q})}]$ (\textsc{Handle}). From the sub-derivations of
      \textsc{T-Handle}, we learn that
      $\ensuremath{\Conid{Γ},\overline{x_{i}},\Varid{p},\Varid{k}\mathbin{⊢}\Varid{e}\mathbin{:}}\ann{\ensuremath{\Varid{ε}\mathbin{*}\mathbin{⟨}\Varid{ℓ},\epsilon_{l}\mathbin{⟩}}}{\susp{\ensuremath{\tau_{2}}}}\ensuremath{\mathbin{∣}\Varid{ε}^{\prime}\mathbin{*}\epsilon_{l}^{\prime}}$ \textbf{(1)}, \ensuremath{\Conid{Γ}\mathbin{⊢}\Conid{E}^{\Varid{ℓ}}}$[$\ensuremath{(\op{op}^{\ell}\;\Varid{f}\;\overline{\Varid{v}})\kw{!}}$]\ $\ensuremath{\tau_{1}\mathbin{∣}\mathbin{⟨}\Varid{ℓ},\Varid{ε}\mathbin{⟩}\mathbin{*}\epsilon_{l}^{\prime}} \textbf{(2)},
      $\ensuremath{\Conid{Γ}\mathbin{⊢}\Varid{v}_{\Varid{p}}\mathbin{:}\tau_{p}\mathbin{∣}\Varid{ε}^{\prime}\mathbin{*}\epsilon_{l}^{\prime}}$ \textbf{(3)}, and
      $\ensuremath{\Conid{Γ}\mathbin{⊢}\op{op}^{ℓ}_{i}\mathbin{:}\mathbin{∀}\Varid{r}\;r_{l}\;\!.\!\;\overline{\tau_{i}}}[\ensuremath{\Varid{r}}/\ensuremath{\Varid{ε}},\ensuremath{r_{l}},\ensuremath{\epsilon_{l}}]\ensuremath{\to }\ann{\ensuremath{\Varid{ε}\mathbin{⊕}\epsilon_{l}\mathbin{⊕}\mathbin{⟨}\Varid{ℓ}\mathbin{⟩}\mathbin{*}\mathbin{⟨⟩}}}{\susp{\ensuremath{\tau_{i}}}}\
      \ensuremath{\mathbin{∣}\Varid{ε}^{\prime}\mathbin{*}\epsilon_{l}^{\prime}}$ \textbf{(5)}.

      We discharge \ensuremath{\overline{x_{i}}}, \ensuremath{\Varid{p}}, and \ensuremath{\Varid{k}} from the context in \textbf{(1)} using
      repeated applications of \cref{lem:substitution}, concluding
      $\ensuremath{\Conid{Γ}\mathbin{⊢}e_{2}\mathbin{:}\Varid{τ}\mathbin{∣}\Varid{ε}^{\prime}\mathbin{*}\epsilon_{l}^{\prime}}$.  We obtain the required derivations for the
      substituted values as follows:
      \begin{itemize}
      \item For \ensuremath{\Varid{k}}, we use the \textsc{T-Abs} and \textsc{T-Handle} rules,
        together with \cref{lem:replacement} and \textbf{(2)} to construct a
        derivation that the expression $\ensuremath{(\Varid{λ}\;\Varid{y}\;\Varid{q}}.\ensuremath{\kw{handle}^{\ell}\;\{\mskip1.5mu \overline{\Conid{C}}\mskip1.5mu\}}\
        \inE{\ensuremath{\Varid{y}\kw{!}}}\ \ensuremath{\Varid{q})}$ has the required type. 
      \item For \ensuremath{\Varid{p}}, we use \textbf{(3)}.
      \item For \ensuremath{\overline{\Varid{x}}}, we use the sub-derivations we get from \textbf{(2)} by decomposing the
        operation call and its context, removing the trailing occurrences of \textsc{T-App}. 
      \end{itemize}
    \end{enumerate}
  \end{case*} 
\end{proof}

\subsection{Progress}

\begin{lemma}[Progress]
  For every expression \ensuremath{\Varid{e}}, if \ensuremath{\Conid{Γ}_{\mathrm{0}}\mathbin{⊢}\Varid{e}\mathbin{:}\Varid{τ}\mathbin{∣}\Varid{ε}\mathbin{*}\epsilon_{l}}, then either there exists
  some \ensuremath{\Varid{e}^{\prime}} such that \ensuremath{\Varid{e} \longmapsto \Varid{e}^{\prime}}, or \ensuremath{\Varid{e}} is a value, where \ensuremath{\Conid{Γ}_{\mathrm{0}}} is an initial
  context that contains bindings for the operations of all known effects in \ensuremath{\Conid{Σ}}.
\end{lemma}

\begin{proof}
  We proceed by rule induction over the derivation for  \ensuremath{\Conid{Γ}_{\mathrm{0}}\mathbin{⊢}\Varid{e}\mathbin{:}\Varid{τ}\mathbin{∣}\Varid{ε}\mathbin{*}\epsilon_{l}}. 

  \begin{case*}[T-Var, T-Abs]
    \ensuremath{\Varid{e}} is a value. 
  \end{case*}

  \begin{case*}[T-App]
    \ensuremath{\Varid{e}\mathrel{=}e_{1}\;e_{2}}, thus from the induction hypothesis for \ensuremath{e_{1}} we know that we can
    either step \ensuremath{e_{1}}, or \ensuremath{e_{1}} is a value, in which case it must be a
    \ensuremath{\Varid{λ}}-abstraction and we can step using the \ensuremath{\Varid{β}}-rule.
  \end{case*}

  \begin{case*}[T-Let]
    \ensuremath{\Varid{e}\mathrel{=}\mathbf{let}\;\Varid{x}\mathrel{=}e_{1}\;\mathbf{in}\;e_{2}}. The induction hypothesis for e₁ tells us that \ensuremath{e_{1}}
    is either a value, or it can be reduced. In the former case, we step using
    the \textsc{Let} rule. Otherwise, we reduce \ensuremath{e_{1}}.
  \end{case*}

  \begin{case*}[T-Inst, T-Gen, T-Resurface]
    \ensuremath{\Varid{e}\mathrel{=}\Varid{e}}, follows from the induction hypothesis. 
  \end{case*}
  
  \begin{case*}[T-Suspend] 
    \ensuremath{\Varid{e}\mathrel{=}}\ \susp{e}, thus \ensuremath{\Varid{e}} is a value. 
  \end{case*}

  \begin{case*}[T-Enact] 
    \ensuremath{\Varid{e}\mathrel{=}\Varid{e}\kw{!}}, thus from the induction hypothesis for \ensuremath{\Varid{e}} we know that either we
    can step \ensuremath{\Varid{e}}, or \ensuremath{\Varid{e}} is value. In the latter case \ensuremath{\Varid{e}} must of the form
    $\susp{e}$ and we step with the \textsc{Enact} rule.
  \end{case*}

  \begin{case*}[T-Handle] 
    \ensuremath{\Varid{e}\mathrel{=}\kw{handle}^{\ell}\;\{\mskip1.5mu \overline{\Conid{C}}\mskip1.5mu\}\;e_{p}\;\Varid{e}}. The induction hypothesis for \ensuremath{e_{p}} tells us we
    can either step \ensuremath{e_{p}}, or \ensuremath{e_{p}} is a value. In the latter case, by the
    induction hypothesis for \ensuremath{\Varid{e}} we can either step \ensuremath{\Varid{e}}, \ensuremath{\Varid{e}} enacts an operation
    of the effect \ensuremath{\Varid{ℓ}}, or \ensuremath{\Varid{e}} is a value, in which case we step either
    with \textsc{Handle} or \textsc{Return}.
  \end{case*}
\end{proof}

\section{Proving Separation of Concerns}

This appendix contains the proof that \emph{handling order is irrelevant for nested handlers with separate concerns}; i.e.,

\begin{equation*}
  \inferrule*
  {  \ensuremath{\Varid{h}\mathbin{:}\mathbin{∀}\Varid{α}}.\
       \ensuremath{\ann{\Varid{ε}\mathbin{*}\mathbin{⟨}\ell_{1},\mathbin{⟨}\ell_{2},\Varid{ε'}\mathbin{⟩}\mathbin{⟩}}{\susp{\Varid{α}}}\mathbin{→}\Conid{P}\mathbin{→}\ann{\mathbin{⟨⟩}\mathbin{*}(\Varid{ε}\mathbin{⊕}\mathbin{⟨}\ell_{1},\mathbin{⟨}\ell_{2},\Varid{ε'}\mathbin{⟩⟩})}{\susp{\Conid{G}\;\Varid{α}}}}
  \\ \ensuremath{\Varid{h}_1\mathbin{:}\mathbin{∀}\Varid{α}\;\Varid{r}\;\Varid{r'}}.\
       \ensuremath{\ann{\mathbin{⟨}\ell_{1},\Varid{r}\mathbin{⟩}\mathbin{*}\Varid{r'}}{\susp{\Varid{α}}}\mathbin{→}\Conid{P}_1\mathbin{→}\ann{\Varid{r}\mathbin{*}\mathbin{⟨}\ell_{1},\Varid{r'}\mathbin{⟩}}{\susp{\Conid{F}_1\;\Varid{α}}}}
  \\ \ensuremath{\Varid{h}_2\mathbin{:}\mathbin{∀}\Varid{α}\;\Varid{r}\;\Varid{r'}}.\
       \ensuremath{\ann{\mathbin{⟨}\ell_{2},\Varid{r}\mathbin{⟩}\mathbin{*}\Varid{r'}}{\susp{\Varid{α}}}\mathbin{→}\Conid{P}_2\mathbin{→}\ann{\Varid{r}\mathbin{*}\mathbin{⟨}\ell_{1},\Varid{r'}\mathbin{⟩}}{\susp{\Conid{F}_2\;\Varid{α}}}}
  \\ \ensuremath{\Varid{e}\mathbin{:}\ann{\mathbin{⟨}\ell_{1},\mathbin{⟨}\ell_{2},\Varid{ε}\mathbin{⟩}\mathbin{⟩}\mathbin{*}\Varid{ε'}}{\susp{\Conid{X}}}}
  \\ \ensuremath{\Conid{H0}\mathbin{:}\Varid{h}^{\Varid{p}}\;(\Varid{h}_1^{\Varid{p}_{1}}\;(\Varid{h}_2^{\Varid{p}_{2}}\;\Varid{e})) \longtwoheadmapsto \Varid{v}_{\mathrm{12}}}
  \and \ensuremath{\Conid{H1}\mathbin{:}\Varid{h}^{\Varid{p}}\;(\Varid{h}_2^{\Varid{p}_{2}}\;(\Varid{h}_1^{\Varid{p}_{1}}\;\Varid{ex})) \longtwoheadmapsto \Varid{v}_{\mathrm{21}}}
  }
  { \ensuremath{\Varid{v}_{\mathrm{12}}\sim\Varid{v}_{\mathrm{21}}} }.
\end{equation*}

To get a more useful induction hypothesis, we generalize \ensuremath{\Conid{H0}} as follows, where \ensuremath{ \longtwoheadmapsto _{\beta}^{*}} represents a sequence of zero or more \ensuremath{\Varid{β}} reductions:

\begin{equation*}
  \inferrule*
  {  \ensuremath{\Varid{h}\mathbin{:}\mathbin{∀}\Varid{α}}.\
       \ensuremath{\ann{\Varid{ε}\mathbin{*}\mathbin{⟨}\ell_{1},\mathbin{⟨}\ell_{2},\Varid{ε'}\mathbin{⟩}\mathbin{⟩}}{\susp{\Varid{α}}}\mathbin{→}\Conid{P}\mathbin{→}\ann{\mathbin{⟨⟩}\mathbin{*}(\Varid{ε}\mathbin{⊕}\mathbin{⟨}\ell_{1},\mathbin{⟨}\ell_{2},\Varid{ε'}\mathbin{⟩⟩})}{\susp{\Conid{G}\;\Varid{α}}}}
  \\ \ensuremath{\Varid{h}_1\mathbin{:}\mathbin{∀}\Varid{α}\;\Varid{r}\;\Varid{r'}}.\
       \ensuremath{\ann{\mathbin{⟨}\ell_{1},\Varid{r}\mathbin{⟩}\mathbin{*}\Varid{r'}}{\susp{\Varid{α}}}\mathbin{→}\Conid{P}_1\mathbin{→}\ann{\Varid{r}\mathbin{*}\mathbin{⟨}\ell_{1},\Varid{r'}\mathbin{⟩}}{\susp{\Conid{F}_1\;\Varid{α}}}}
  \\ \ensuremath{\Varid{h}_2\mathbin{:}\mathbin{∀}\Varid{α}\;\Varid{r}\;\Varid{r'}}.\
       \ensuremath{\ann{\mathbin{⟨}\ell_{2},\Varid{r}\mathbin{⟩}\mathbin{*}\Varid{r'}}{\susp{\Varid{α}}}\mathbin{→}\Conid{P}_2\mathbin{→}\ann{\Varid{r}\mathbin{*}\mathbin{⟨}\ell_{1},\Varid{r'}\mathbin{⟩}}{\susp{\Conid{F}_2\;\Varid{α}}}}
  \\ \ensuremath{\Varid{e}\mathbin{:}\ann{\mathbin{⟨}\ell_{1},\mathbin{⟨}\ell_{2},\Varid{ε}\mathbin{⟩}\mathbin{⟩}\mathbin{*}\Varid{ε'}}{\susp{\Conid{X}}}}
  \\ \ensuremath{\Conid{H0}\mathbin{:}e_{0} \longtwoheadmapsto \Varid{v}_{\mathrm{12}}}
  \and \ensuremath{\Conid{H0}^{\prime}\mathbin{:}e_{0} \longtwoheadmapsto _{\beta}^{*}\Varid{h}^{\Varid{p}}\;(\Varid{h}_1^{\Varid{p}_{1}}\;(\Varid{h}_2^{\Varid{p}_{2}}\;\Varid{e}))}
  \and \ensuremath{\Conid{H1}\mathbin{:}\Varid{h}^{\Varid{p}}\;(\Varid{h}_2^{\Varid{p}_{2}}\;(\Varid{h}_1^{\Varid{p}_{1}}\;\Varid{ex})) \longtwoheadmapsto \Varid{v}_{\mathrm{21}}}
  }
  { \ensuremath{\Varid{v}_{\mathrm{12}}\sim\Varid{v}_{\mathrm{21}}} }.
\end{equation*}

The proof assumes that the \emph{handler op lemma} holds for each handler.
The rule below summarizes the handler op lemma proposition, which assumes that \ensuremath{\Varid{h}} and \ensuremath{\Varid{h}_2} satisfy the handler return lemma:

\begin{quote}
  Each operation \ensuremath{\op{op}^{\ell}\;\Varid{f}} of handler \ensuremath{\Varid{h}_1\mathbin{:}\mathbin{∀}\Varid{α}\;\Varid{r}\;\Varid{r'}} .\ \ensuremath{\ann{\mathbin{⟨}\Varid{ℓ},\Varid{r}\mathbin{⟩}\mathbin{*}\Varid{r'}}{\susp{\Varid{α}}}\mathbin{→}\ann{\Varid{r}\mathbin{*}\mathbin{⟨}\Varid{ℓ},\Varid{r'}\mathbin{⟩}}{\susp{\Conid{F}\;\Varid{α}}}} provides separation of concerns.
  That is,
  \begin{equation*}
    \inferrule*
    {  \ensuremath{\Varid{h}\mathbin{:}\mathbin{∀}\Varid{α}}.\
         \ensuremath{\ann{\Varid{ε}\mathbin{*}\mathbin{⟨}\ell_{1},\mathbin{⟨}\ell_{2},\Varid{ε'}\mathbin{⟩}\mathbin{⟩}}{\susp{\Varid{α}}}\mathbin{→}\Conid{P}\mathbin{→}\ann{\mathbin{⟨⟩}\mathbin{*}(\Varid{ε}\mathbin{⊕}\mathbin{⟨}\ell_{1},\mathbin{⟨}\ell_{2},\Varid{ε'}\mathbin{⟩⟩})}{\susp{\Conid{G}\;\Varid{α}}}}
    \\ \ensuremath{\Varid{h}_1\mathbin{:}\mathbin{∀}\Varid{α}\;\Varid{r}\;\Varid{r'}}.\
         \ensuremath{\ann{\mathbin{⟨}\ell_{1},\Varid{r}\mathbin{⟩}\mathbin{*}\Varid{r'}}{\susp{\Varid{α}}}\mathbin{→}\Conid{P}_1\mathbin{→}\ann{\Varid{r}\mathbin{*}\mathbin{⟨}\ell_{1},\Varid{r'}\mathbin{⟩}}{\susp{\Conid{F}_1\;\Varid{α}}}}
    \\ \ensuremath{\Varid{h}_2\mathbin{:}\mathbin{∀}\Varid{α}\;\Varid{r}\;\Varid{r'}}.\
         \ensuremath{\ann{\mathbin{⟨}\ell_{2},\Varid{r}\mathbin{⟩}\mathbin{*}\Varid{r'}}{\susp{\Varid{α}}}\mathbin{→}\Conid{P}_2\mathbin{→}\ann{\Varid{r}\mathbin{*}\mathbin{⟨}\ell_{1},\Varid{r'}\mathbin{⟩}}{\susp{\Conid{F}_2\;\Varid{α}}}}
    \\ \ensuremath{\Conid{H0}\mathbin{:}\Varid{h}^{\Varid{p}}\;(\Varid{h}_1^{\Varid{p}_{1}}\;(\Varid{h}_2^{\Varid{p}_{2}}\;\Conid{E}[\op{op}^{\ell}\;\Varid{f}\;\overline{\Varid{v}_{\Varid{v}}}\;\overline{\susp{e_{\Varid{m}}}}])) \longmapsto \Varid{e}^{\prime}}
    \and \ensuremath{\Conid{H1}\mathbin{:}\Varid{e}^{\prime} \longtwoheadmapsto \Varid{v}_{\mathrm{12}}}
    \\ \ensuremath{\Conid{H2}\mathbin{:}\Varid{h}^{\Varid{p}}\;(\Varid{h}_2^{\Varid{p}_{2}}\;(\Varid{h}_1^{\Varid{p}_{1}}\;\Conid{E}[\op{op}^{\ell}\;\Varid{f}\;\overline{\Varid{v}_{\Varid{v}}}\;\overline{\{\mskip1.5mu e_{\Varid{m}}\mskip1.5mu\}}])) \longtwoheadmapsto \Varid{v}_{\mathrm{21}}}
    \\ \ensuremath{\Conid{IH}\mathbin{:}} {\left(\begin{array}{l}
       \ensuremath{\mathbin{∀}e_{0}}.\
         \ensuremath{\Varid{e}^{\prime} \longtwoheadmapsto _{\beta}^{*}\Varid{h}^{\Varid{p}}\;(\Varid{h}_1^{\Varid{p}_{1}}\;(\Varid{h}_2^{\Varid{p}_{2}}\;e_{0})) \Longrightarrow }
         \ensuremath{\Varid{h}^{\Varid{p}}\;(\Varid{h}_2^{\Varid{p}_{2}}\;(\Varid{h}_1^{\Varid{p}_{1}}\;e_{0})) \longtwoheadmapsto \Varid{v}_{\mathrm{21}} \Longrightarrow }
         \ensuremath{\Varid{v}_{\mathrm{12}}\sim\Varid{v}_{\mathrm{21}}}
       \end{array}\right)}
    }
    { \ensuremath{\Varid{v}_{\mathrm{12}}\sim\Varid{v}_{\mathrm{21}}} }
  \end{equation*}
\end{quote}

Furthermore, we assume that \emph{handler return lemma} holds for each handler; i.e.,

\begin{quote}
\begin{equation*}
  \inferrule*
  {  \ensuremath{\Varid{h}\mathbin{:}\mathbin{∀}\Varid{α}\;\Varid{r}\;\Varid{r}^{\prime}}.\
       \ensuremath{\ann{\mathbin{⟨}\ell_{1},\Varid{r}\mathbin{⟩}\mathbin{*}\Varid{r'}}{\susp{\Varid{α}}}\mathbin{→}\Conid{P}\mathbin{→}\ann{\Varid{r}\mathbin{*}\mathbin{⟨}\Varid{ℓ},\Varid{r}^{\prime}\mathbin{⟩}}{\susp{\Conid{F}\;\Varid{α}}}}
  \\ \ensuremath{\Varid{p}\mathbin{:}\Conid{P}}
  \\ \ensuremath{\Varid{v}\mathbin{:}\Conid{X}}
  }
  { \ensuremath{\Varid{toBag}_{\Conid{F}}\;(\Varid{h}^{\Varid{p}}\;\susp{\Varid{v}})\mathrel{=}\{\mskip1.5mu \Varid{v}\mskip1.5mu\}} }
\end{equation*}
\end{quote}

Lastly, we assume the following definition of \ensuremath{ \longtwoheadmapsto }.

\begin{mathpar}
\inferrule[Trans]
{  \ensuremath{e_{0} \longmapsto e_{1}}
\\ \ensuremath{e_{1} \longtwoheadmapsto e_{2}} }
{  \ensuremath{e_{0} \longtwoheadmapsto e_{2}} }
\and
\inferrule[Refl]
{}
{  \ensuremath{e_{0} \longtwoheadmapsto e_{0}} }
\end{mathpar}

\subsection{Proof that Handling Order is Irrelevant for Nested Handlers with Separate Concerns}
\label{sec:soc-appendix}

\begin{proof}
  Our goal is to prove

  \begin{equation*}
  \inferrule*
  {  \ensuremath{\Varid{h}\mathbin{:}\mathbin{∀}\Varid{α}}.\
       \ensuremath{\ann{\Varid{ε}\mathbin{*}\mathbin{⟨}\ell_{1},\mathbin{⟨}\ell_{2},\Varid{ε'}\mathbin{⟩}\mathbin{⟩}}{\susp{\Varid{α}}}\mathbin{→}\Conid{P}\mathbin{→}\ann{\mathbin{⟨⟩}\mathbin{*}(\Varid{ε}\mathbin{⊕}\mathbin{⟨}\ell_{1},\mathbin{⟨}\ell_{2},\Varid{ε'}\mathbin{⟩⟩})}{\susp{\Conid{G}\;\Varid{α}}}}
  \\ \ensuremath{\Varid{h}_1\mathbin{:}\mathbin{∀}\Varid{α}\;\Varid{r}\;\Varid{r'}}.\
       \ensuremath{\ann{\mathbin{⟨}\ell_{1},\Varid{r}\mathbin{⟩}\mathbin{*}\Varid{r'}}{\susp{\Varid{α}}}\mathbin{→}\Conid{P}_1\mathbin{→}\ann{\Varid{r}\mathbin{*}\mathbin{⟨}\ell_{1},\Varid{r'}\mathbin{⟩}}{\susp{\Conid{F}_1\;\Varid{α}}}}
  \\ \ensuremath{\Varid{h}_2\mathbin{:}\mathbin{∀}\Varid{α}\;\Varid{r}\;\Varid{r'}}.\
       \ensuremath{\ann{\mathbin{⟨}\ell_{2},\Varid{r}\mathbin{⟩}\mathbin{*}\Varid{r'}}{\susp{\Varid{α}}}\mathbin{→}\Conid{P}_2\mathbin{→}\ann{\Varid{r}\mathbin{*}\mathbin{⟨}\ell_{1},\Varid{r'}\mathbin{⟩}}{\susp{\Conid{F}_2\;\Varid{α}}}}
  \\ \ensuremath{\Varid{e}\mathbin{:}\ann{\mathbin{⟨}\ell_{1},\mathbin{⟨}\ell_{2},\Varid{ε}\mathbin{⟩}\mathbin{⟩}\mathbin{*}\Varid{ε'}}{\susp{\Conid{X}}}}
  \\ \ensuremath{\Conid{H0}\mathbin{:}e_{0} \longtwoheadmapsto \Varid{v}_{\mathrm{12}}}
  \and \ensuremath{\Conid{H0}^{\prime}\mathbin{:}e_{0} \longtwoheadmapsto _{\beta}^{*}\Varid{h}^{\Varid{p}}\;(\Varid{h}_1^{\Varid{p}_{1}}\;(\Varid{h}_2^{\Varid{p}_{2}}\;\Varid{e}))}
  \and \ensuremath{\Conid{H1}\mathbin{:}\Varid{h}^{\Varid{p}}\;(\Varid{h}_2^{\Varid{p}_{2}}\;(\Varid{h}_1^{\Varid{p}_{1}}\;\Varid{e})) \longtwoheadmapsto \Varid{v}_{\mathrm{21}}}
  }
  { \ensuremath{\Varid{v}_{\mathrm{12}}\sim\Varid{v}_{\mathrm{21}}} }.
  \end{equation*}

  The proof is by rule induction on \ensuremath{\Conid{H0}}.

  The base case of the induction proof holds vacuously, because a value \ensuremath{\Varid{v}_{\mathrm{12}}} cannot \ensuremath{\Varid{β}} reduce to \ensuremath{\Varid{h}^{\Varid{p}}\;(\Varid{h}_1^{\Varid{p}_{1}}\;(\Varid{h}_2^{\Varid{p}_{2}}\;\Varid{e}))}.

  The proof obligation for the inductive case is

  \begin{equation*}
    \inferrule*
    {  \ensuremath{\Varid{h}\mathbin{:}\mathbin{∀}\Varid{α}}.\
         \ensuremath{\ann{\Varid{ε}\mathbin{*}\mathbin{⟨}\ell_{1},\mathbin{⟨}\ell_{2},\Varid{ε'}\mathbin{⟩}\mathbin{⟩}}{\susp{\Varid{α}}}\mathbin{→}\Conid{P}\mathbin{→}\ann{\mathbin{⟨⟩}\mathbin{*}(\Varid{ε}\mathbin{⊕}\mathbin{⟨}\ell_{1},\mathbin{⟨}\ell_{2},\Varid{ε'}\mathbin{⟩⟩})}{\susp{\Conid{G}\;\Varid{α}}}}
    \\ \ensuremath{\Varid{h}_1\mathbin{:}\mathbin{∀}\Varid{α}\;\Varid{r}\;\Varid{r'}}.\
         \ensuremath{\ann{\mathbin{⟨}\ell_{1},\Varid{r}\mathbin{⟩}\mathbin{*}\Varid{r'}}{\susp{\Varid{α}}}\mathbin{→}\Conid{P}_1\mathbin{→}\ann{\Varid{r}\mathbin{*}\mathbin{⟨}\ell_{1},\Varid{r'}\mathbin{⟩}}{\susp{\Conid{F}_1\;\Varid{α}}}}
    \\ \ensuremath{\Varid{h}_2\mathbin{:}\mathbin{∀}\Varid{α}\;\Varid{r}\;\Varid{r'}}.\
         \ensuremath{\ann{\mathbin{⟨}\ell_{2},\Varid{r}\mathbin{⟩}\mathbin{*}\Varid{r'}}{\susp{\Varid{α}}}\mathbin{→}\Conid{P}_2\mathbin{→}\ann{\Varid{r}\mathbin{*}\mathbin{⟨}\ell_{1},\Varid{r'}\mathbin{⟩}}{\susp{\Conid{F}_2\;\Varid{α}}}}
    \\ \ensuremath{\Varid{e}\mathbin{:}\ann{\mathbin{⟨}\ell_{1},\mathbin{⟨}\ell_{2},\Varid{ε}\mathbin{⟩}\mathbin{⟩}\mathbin{*}\Varid{ε'}}{\susp{\Conid{X}}}}
    \\ \ensuremath{\Conid{H0}\mathbin{:}e_{0} \longmapsto \Varid{e}^{\prime}}
    \\ \ensuremath{\Conid{H0}^{\prime}\mathbin{:}e_{0} \longtwoheadmapsto _{\beta}^{*}\Varid{h}^{\Varid{p}}\;(\Varid{h}_1^{\Varid{p}_{1}}\;(\Varid{h}_2^{\Varid{p}_{2}}\;\Varid{e}))}
    \\ \ensuremath{\Conid{H1}\mathbin{:}\Varid{e}^{\prime} \longtwoheadmapsto \Varid{v}_{\mathrm{12}}}
    \\ \ensuremath{\Conid{H2}\mathbin{:}\Varid{h}^{\Varid{p}}\;(\Varid{h}_2^{\Varid{p}_{2}}\;(\Varid{h}_1^{\Varid{p}_{1}}\;\Varid{e})) \longtwoheadmapsto \Varid{v}_{\mathrm{21}}}
    \\ \ensuremath{\Conid{IH}\mathbin{:}} {\left(\begin{array}{l}
       \ensuremath{\mathbin{∀}e_{0}}.\
         \ensuremath{\Varid{e}^{\prime} \longtwoheadmapsto _{\beta}^{*}\Varid{h}^{\Varid{p}}\;(\Varid{h}_1^{\Varid{p}_{1}}\;(\Varid{h}_2^{\Varid{p}_{2}}\;e_{0})) \Longrightarrow }
         \ensuremath{\Varid{h}^{\Varid{p}}\;(\Varid{h}_2^{\Varid{p}_{2}}\;(\Varid{h}_1^{\Varid{p}_{1}}\;e_{0})) \longtwoheadmapsto \Varid{v}_{\mathrm{21}} \Longrightarrow }
         \ensuremath{\Varid{v}_{\mathrm{12}}\sim\Varid{v}_{\mathrm{21}}}
       \end{array}\right)}
    }
    { \ensuremath{\Varid{v}_{\mathrm{12}}\sim\Varid{v}_{\mathrm{21}}} }
  \end{equation*}

  Inverting \ensuremath{\Conid{H0}^{\prime}}, there are two cases to consider.

  \begin{case*}[Refl]

  The goal is

  \begin{equation*}
    \inferrule*
    {  \ensuremath{\Varid{h}\mathbin{:}\mathbin{∀}\Varid{α}}.\
         \ensuremath{\ann{\Varid{ε}\mathbin{*}\mathbin{⟨}\ell_{1},\mathbin{⟨}\ell_{2},\Varid{ε'}\mathbin{⟩}\mathbin{⟩}}{\susp{\Varid{α}}}\mathbin{→}\Conid{P}\mathbin{→}\ann{\mathbin{⟨⟩}\mathbin{*}(\Varid{ε}\mathbin{⊕}\mathbin{⟨}\ell_{1},\mathbin{⟨}\ell_{2},\Varid{ε'}\mathbin{⟩⟩})}{\susp{\Conid{G}\;\Varid{α}}}}
    \\ \ensuremath{\Varid{h}_1\mathbin{:}\mathbin{∀}\Varid{α}\;\Varid{r}\;\Varid{r'}}.\
         \ensuremath{\ann{\mathbin{⟨}\ell_{1},\Varid{r}\mathbin{⟩}\mathbin{*}\Varid{r'}}{\susp{\Varid{α}}}\mathbin{→}\Conid{P}_1\mathbin{→}\ann{\Varid{r}\mathbin{*}\mathbin{⟨}\ell_{1},\Varid{r'}\mathbin{⟩}}{\susp{\Conid{F}_1\;\Varid{α}}}}
    \\ \ensuremath{\Varid{h}_2\mathbin{:}\mathbin{∀}\Varid{α}\;\Varid{r}\;\Varid{r'}}.\
         \ensuremath{\ann{\mathbin{⟨}\ell_{2},\Varid{r}\mathbin{⟩}\mathbin{*}\Varid{r'}}{\susp{\Varid{α}}}\mathbin{→}\Conid{P}_2\mathbin{→}\ann{\Varid{r}\mathbin{*}\mathbin{⟨}\ell_{1},\Varid{r'}\mathbin{⟩}}{\susp{\Conid{F}_2\;\Varid{α}}}}
    \\ \ensuremath{\Varid{e}\mathbin{:}\ann{\mathbin{⟨}\ell_{1},\mathbin{⟨}\ell_{2},\Varid{ε}\mathbin{⟩}\mathbin{⟩}\mathbin{*}\Varid{ε'}}{\susp{\Conid{X}}}}
    \\ \ensuremath{\Conid{H0}\mathbin{:}\Varid{h}^{\Varid{p}}\;(\Varid{h}_1^{\Varid{p}_{1}}\;(\Varid{h}_2^{\Varid{p}_{2}}\;\Varid{e})) \longmapsto \Varid{e}^{\prime}}
    \\ \ensuremath{\Conid{H1}\mathbin{:}\Varid{e}^{\prime} \longtwoheadmapsto \Varid{v}_{\mathrm{12}}}
    \\ \ensuremath{\Conid{H2}\mathbin{:}\Varid{h}^{\Varid{p}}\;(\Varid{h}_2^{\Varid{p}_{2}}\;(\Varid{h}_1^{\Varid{p}_{1}}\;\Varid{e})) \longtwoheadmapsto \Varid{v}_{\mathrm{21}}}
    \\ \ensuremath{\Conid{IH}\mathbin{:}} {\left(\begin{array}{l}
       \ensuremath{\mathbin{∀}e_{0}}.\
         \ensuremath{\Varid{e}^{\prime} \longtwoheadmapsto _{\beta}^{*}\Varid{h}^{\Varid{p}}\;(\Varid{h}_1^{\Varid{p}_{1}}\;(\Varid{h}_2^{\Varid{p}_{2}}\;e_{0})) \Longrightarrow }
         \ensuremath{\Varid{h}^{\Varid{p}}\;(\Varid{h}_2^{\Varid{p}_{2}}\;(\Varid{h}_1^{\Varid{p}_{1}}\;e_{0})) \longtwoheadmapsto \Varid{v}_{\mathrm{21}} \Longrightarrow }
         \ensuremath{\Varid{v}_{\mathrm{12}}\sim\Varid{v}_{\mathrm{21}}}
       \end{array}\right)}
    }
    { \ensuremath{\Varid{v}_{\mathrm{12}}\sim\Varid{v}_{\mathrm{21}}} }
  \end{equation*}

  Inverting \ensuremath{\Conid{H0}} using the definition of \ensuremath{ \longmapsto } and \ensuremath{\mathbin{-->}}, the cases for \ensuremath{\Varid{β}} and \ensuremath{\mathbf{let}} follow trivially from the induction hypothesis.
  The \textsc{Return} case follows from the handler return lemmas of \ensuremath{\Varid{h}_1} and \ensuremath{\Varid{h}_2}, and the \textsc{Handle} case follows from the handler op lemmas of \ensuremath{\Varid{h}_1} and \ensuremath{\Varid{h}_2}.

  \end{case*}

  \begin{case*}[Trans]

  The goal is

  \begin{equation*}
  \inferrule*
  {  \ensuremath{\Varid{h}\mathbin{:}\mathbin{∀}\Varid{α}}.\
       \ensuremath{\ann{\Varid{ε}\mathbin{*}\mathbin{⟨}\ell_{1},\mathbin{⟨}\ell_{2},\Varid{ε'}\mathbin{⟩}\mathbin{⟩}}{\susp{\Varid{α}}}\mathbin{→}\Conid{P}\mathbin{→}\ann{\mathbin{⟨⟩}\mathbin{*}(\Varid{ε}\mathbin{⊕}\mathbin{⟨}\ell_{1},\mathbin{⟨}\ell_{2},\Varid{ε'}\mathbin{⟩⟩})}{\susp{\Conid{G}\;\Varid{α}}}}
  \\ \ensuremath{\Varid{h}_1\mathbin{:}\mathbin{∀}\Varid{α}\;\Varid{r}\;\Varid{r'}}.\
       \ensuremath{\ann{\mathbin{⟨}\ell_{1},\Varid{r}\mathbin{⟩}\mathbin{*}\Varid{r'}}{\susp{\Varid{α}}}\mathbin{→}\Conid{P}_1\mathbin{→}\ann{\Varid{r}\mathbin{*}\mathbin{⟨}\ell_{1},\Varid{r'}\mathbin{⟩}}{\susp{\Conid{F}_1\;\Varid{α}}}}
  \\ \ensuremath{\Varid{h}_2\mathbin{:}\mathbin{∀}\Varid{α}\;\Varid{r}\;\Varid{r'}}.\
       \ensuremath{\ann{\mathbin{⟨}\ell_{2},\Varid{r}\mathbin{⟩}\mathbin{*}\Varid{r'}}{\susp{\Varid{α}}}\mathbin{→}\Conid{P}_2\mathbin{→}\ann{\Varid{r}\mathbin{*}\mathbin{⟨}\ell_{1},\Varid{r'}\mathbin{⟩}}{\susp{\Conid{F}_2\;\Varid{α}}}}
  \\ \ensuremath{\Varid{e}\mathbin{:}\ann{\mathbin{⟨}\ell_{1},\mathbin{⟨}\ell_{2},\Varid{ε}\mathbin{⟩}\mathbin{⟩}\mathbin{*}\Varid{ε'}}{\susp{\Conid{X}}}}
  \\ \ensuremath{\Conid{H0}\mathbin{:}e_{0} \longmapsto \Varid{e}^{\prime}}
  \\ \ensuremath{\Conid{H0}^{\prime}\mathbin{:}e_{0} \longmapsto _{\beta}e_{0}^{\prime}}
  \\ \ensuremath{\Conid{H0}^{\prime\prime}\mathbin{:}e_{0}^{\prime} \longtwoheadmapsto _{\beta}^{*}\Varid{h}^{\Varid{p}}\;(\Varid{h}_1^{\Varid{p}_{1}}\;(\Varid{h}_2^{\Varid{p}_{2}}\;\Varid{e}))}
  \\ \ensuremath{\Conid{H1}\mathbin{:}\Varid{e}^{\prime} \longtwoheadmapsto \Varid{v}_{\mathrm{12}}}
  \\ \ensuremath{\Conid{H2}\mathbin{:}\Varid{h}^{\Varid{p}}\;(\Varid{h}_2^{\Varid{p}_{2}}\;(\Varid{h}_1^{\Varid{p}_{1}}\;\Varid{e})) \longtwoheadmapsto \Varid{v}_{\mathrm{21}}}
  \\ \ensuremath{\Conid{IH}\mathbin{:}} {\left(\begin{array}{l}
     \ensuremath{\mathbin{∀}e_{0}}.\
       \ensuremath{\Varid{e}^{\prime} \longtwoheadmapsto _{\beta}^{*}\Varid{h}^{\Varid{p}}\;(\Varid{h}_1^{\Varid{p}_{1}}\;(\Varid{h}_2^{\Varid{p}_{2}}\;e_{0})) \Longrightarrow }
       \ensuremath{\Varid{h}^{\Varid{p}}\;(\Varid{h}_2^{\Varid{p}_{2}}\;(\Varid{h}_1^{\Varid{p}_{1}}\;e_{0})) \longtwoheadmapsto \Varid{v}_{\mathrm{21}} \Longrightarrow }
       \ensuremath{\Varid{v}_{\mathrm{12}}\sim\Varid{v}_{\mathrm{21}}}
     \end{array}\right)}
  }
  { \ensuremath{\Varid{v}_{\mathrm{12}}\sim\Varid{v}_{\mathrm{21}}} }.
  \end{equation*}

  Because \ensuremath{ \longmapsto } is a deterministic relation, we get that \ensuremath{e_{0}^{\prime}\mathrel{=}\Varid{e}^{\prime}}, and that the goal follows from the induction hypothesis.

  \end{case*}

\end{proof}

\subsection{S.o.C. Lemmas for the State Handler}
\label{sec:soc-state-appendix}

\subsubsection{Handler Return Lemma for the State Handler}

\begin{proof}

The goal is to prove that

\begin{equation*}
  \inferrule*
  {  \ensuremath{\Varid{h}_{\Conid{St}}\mathbin{:}\mathbin{∀}\Varid{α}\;\Varid{r}\;\Varid{r}^{\prime}}.\
       \ensuremath{\ann{\mathbin{⟨}\Conid{St},\Varid{r}\mathbin{⟩}\mathbin{*}\Varid{r'}}{\susp{\Varid{α}}}\mathbin{→}\Conid{S}\mathbin{→}\ann{\Varid{r}\mathbin{*}\mathbin{⟨}\Conid{St},\Varid{r}^{\prime}\mathbin{⟩}}{\susp{(\Varid{α},\Conid{S})}}}
  \\ \ensuremath{\Varid{s}\mathbin{:}\Conid{S}}
  \\ \ensuremath{\Varid{v}\mathbin{:}\Conid{X}}
  }
  { \ensuremath{\Varid{toBag}_{(,)\;\Conid{S}}\ (\Varid{h}_{\Conid{St}}^{\Varid{s}}\;\susp{\Varid{v}})\mathrel{=}\{\mskip1.5mu \Varid{v}\mskip1.5mu\}} }
\end{equation*}

where \ensuremath{\Varid{toBag}_{(,)\;\Conid{S}}\ (\Varid{x},\Varid{s})\mathrel{=}\{\mskip1.5mu \Varid{x}\mskip1.5mu\}}

The goal follows trivially from the definition of the return case of the state handler.

\end{proof}

\subsubsection{Handler Op Lemma for the State Handler}

\begin{proof}

\begin{case*}[\ensuremath{\op{op}^{\Conid{St}}\mathrel{=}\op{put}\;\Varid{s}^{\prime}}]

\begin{equation*}
  \inferrule*
  {  \ensuremath{\Varid{h}\mathbin{:}\mathbin{∀}\Varid{α}}.\
       \ensuremath{\ann{\Varid{ε}\mathbin{*}\mathbin{⟨}\Conid{St},\mathbin{⟨}\ell_{2},\Varid{ε'}\mathbin{⟩}\mathbin{⟩}}{\susp{\Varid{α}}}\mathbin{→}\Conid{P}\mathbin{→}\ann{\mathbin{⟨⟩}\mathbin{*}(\Varid{ε}\mathbin{⊕}\mathbin{⟨}\Conid{St},\mathbin{⟨}\ell_{2},\Varid{ε'}\mathbin{⟩⟩})}{\susp{\Conid{G}\;\Varid{α}}}}
  \\ \ensuremath{\Varid{h}_{\Conid{St}}\mathbin{:}\mathbin{∀}\Varid{α}\;\Varid{r}\;\Varid{r'}}.\
     \ensuremath{\ann{\mathbin{⟨}\Conid{St},\Varid{r}\mathbin{⟩}\mathbin{*}\Varid{r'}}{\susp{\Varid{α}}}\mathbin{→}\Conid{S}\mathbin{→}\ann{\Varid{r}\mathbin{*}\mathbin{⟨}\Conid{St},\Varid{r'}\mathbin{⟩}}{\susp{(\Varid{α},\Varid{s})}}}
  \\ \ensuremath{\Varid{h}_2\mathbin{:}\mathbin{∀}\Varid{α}\;\Varid{r}\;\Varid{r'}}.\
       \ensuremath{\ann{\mathbin{⟨}\ell_{2},\Varid{r}\mathbin{⟩}\mathbin{*}\Varid{r'}}{\susp{\Varid{α}}}\mathbin{→}\Conid{P}_2\mathbin{→}\ann{\Varid{r}\mathbin{*}\mathbin{⟨}\Conid{St},\Varid{r'}\mathbin{⟩}}{\susp{\Conid{F}_2\;\Varid{α}}}}
  \\ \ensuremath{\Conid{H0}\mathbin{:}\Varid{h}^{\Varid{p}}\;(\Varid{h}_{\Conid{St}}^{\Varid{s}}\;(\Varid{h}_2^{\Varid{p}_{2}}\;\Conid{E}[\op{put}\;\Varid{s}^{\prime}])) \longmapsto \Varid{e}^{\prime}}
  \and \ensuremath{\Conid{H1}\mathbin{:}\Varid{e}^{\prime} \longtwoheadmapsto \Varid{v}_{\mathrm{12}}}
  \\ \ensuremath{\Conid{H2}\mathbin{:}\Varid{h}^{\Varid{p}}\;(\Varid{h}_2^{\Varid{p}_{2}}\;(\Varid{h}_{\Conid{St}}^{\Varid{s}}\;\Conid{E}[\op{put}\;\Varid{s}^{\prime}])) \longtwoheadmapsto \Varid{v}_{\mathrm{21}}}
  \\ \ensuremath{\Conid{IH}\mathbin{:}} {\left(\begin{array}{l}
     \ensuremath{\mathbin{∀}e_{0}}.\
       \ensuremath{\Varid{e}^{\prime} \longtwoheadmapsto _{\beta}^{*}\Varid{h}^{\Varid{p}}\;(\Varid{h}_{\Conid{St}}^{\Varid{s}}\;(\Varid{h}_2^{\Varid{p}_{2}}\;e_{0})) \Longrightarrow }
       \ensuremath{\Varid{h}^{\Varid{p}}\;(\Varid{h}_2^{\Varid{p}_{2}}\;(\Varid{h}_{\Conid{St}}^{\Varid{s}}\;e_{0})) \longtwoheadmapsto \Varid{v}_{\mathrm{21}} \Longrightarrow }
       \ensuremath{\Varid{v}_{\mathrm{12}}\sim\Varid{v}_{\mathrm{21}}}
     \end{array}\right)}
  }
  { \ensuremath{\Varid{v}_{\mathrm{12}}\sim\Varid{v}_{\mathrm{21}}} }
\end{equation*}

Inversion on \ensuremath{\Conid{H0}}.

\begin{equation*}
  \inferrule*
  {  \ensuremath{\Varid{h}\mathbin{:}\mathbin{∀}\Varid{α}}.\
       \ensuremath{\ann{\Varid{ε}\mathbin{*}\mathbin{⟨}\Conid{St},\mathbin{⟨}\ell_{2},\Varid{ε'}\mathbin{⟩}\mathbin{⟩}}{\susp{\Varid{α}}}\mathbin{→}\Conid{P}\mathbin{→}\ann{\mathbin{⟨⟩}\mathbin{*}(\Varid{ε}\mathbin{⊕}\mathbin{⟨}\Conid{St},\mathbin{⟨}\ell_{2},\Varid{ε'}\mathbin{⟩⟩})}{\susp{\Conid{G}\;\Varid{α}}}}
  \\ \ensuremath{\Varid{h}_{\Conid{St}}\mathbin{:}\mathbin{∀}\Varid{α}\;\Varid{r}\;\Varid{r'}}.\
     \ensuremath{\ann{\mathbin{⟨}\Conid{St},\Varid{r}\mathbin{⟩}\mathbin{*}\Varid{r'}}{\susp{\Varid{α}}}\mathbin{→}\Conid{S}\mathbin{→}\ann{\Varid{r}\mathbin{*}\mathbin{⟨}\Conid{St},\Varid{r'}\mathbin{⟩}}{\susp{(\Varid{α},\Varid{s})}}}
  \\ \ensuremath{\Varid{h}_2\mathbin{:}\mathbin{∀}\Varid{α}\;\Varid{r}\;\Varid{r'}}.\
       \ensuremath{\ann{\mathbin{⟨}\ell_{2},\Varid{r}\mathbin{⟩}\mathbin{*}\Varid{r'}}{\susp{\Varid{α}}}\mathbin{→}\Conid{P}_2\mathbin{→}\ann{\Varid{r}\mathbin{*}\mathbin{⟨}\Conid{St},\Varid{r'}\mathbin{⟩}}{\susp{\Conid{F}_2\;\Varid{α}}}}
  \\ \ensuremath{\Conid{H1}\mathbin{:}\Varid{e}^{\prime}\mathbin{≡}\Varid{h}\;(((\Varid{λ}\;()\;\Varid{s}^{\prime\prime}\!.\!\;\Varid{h}_{\Conid{St}}^{\Varid{s}^{\prime\prime}}\;(\Varid{h}_2^{\Varid{p}_{2}}\;\Conid{E}[()]))\;()\;\Varid{s}^{\prime})) \longtwoheadmapsto \Varid{v}_{\mathrm{12}}}
  \\ \ensuremath{\Conid{H2}\mathbin{:}\Varid{h}^{\Varid{p}}\;(\Varid{h}_2^{\Varid{p}_{2}}\;(\Varid{h}_{\Conid{St}}^{\Varid{s}}\;\Conid{E}[\op{put}\;\Varid{s}^{\prime}])) \longtwoheadmapsto \Varid{v}_{\mathrm{21}}}
  \\ \ensuremath{\Conid{IH}\mathbin{:}} {\left(\begin{array}{l}
     \ensuremath{\mathbin{∀}e_{0}}.\
       \ensuremath{\Varid{e}^{\prime} \longtwoheadmapsto _{\beta}^{*}\Varid{h}^{\Varid{p}}\;(\Varid{h}_{\Conid{St}}^{\Varid{s}}\;(\Varid{h}_2^{\Varid{p}_{2}}\;e_{0})) \Longrightarrow }
       \ensuremath{\Varid{h}^{\Varid{p}}\;(\Varid{h}_2^{\Varid{p}_{2}}\;(\Varid{h}_{\Conid{St}}^{\Varid{s}}\;e_{0})) \longtwoheadmapsto \Varid{v}_{\mathrm{21}} \Longrightarrow }
       \ensuremath{\Varid{v}_{\mathrm{12}}\sim\Varid{v}_{\mathrm{21}}}
     \end{array}\right)}
  }
  { \ensuremath{\Varid{v}_{\mathrm{12}}\sim\Varid{v}_{\mathrm{21}}} }
\end{equation*}

Inversion on \ensuremath{\Conid{H1}} and \ensuremath{\Conid{H2}}; specialize \ensuremath{\Conid{IH}}, using \ensuremath{\Varid{h}\;(((\Varid{λ}\;()\;\Varid{s}^{\prime\prime}\!.\!\;\Varid{h}_{\Conid{St}}^{\Varid{s}^{\prime\prime}}\;(\Varid{h}_2^{\Varid{p}_{2}}\;\Conid{E}[()]))\;()\;\Varid{s}^{\prime})) \longtwoheadmapsto _{\beta}^{*}\Varid{h}\;(\Varid{h}_{\Conid{St}}^{\Varid{s'}}\;(\Varid{h}_2^{\Varid{p}_{2}}\;\Conid{E}[()]))}.

\begin{equation*}
  \inferrule*
  {  \ensuremath{\Varid{h}\mathbin{:}\mathbin{∀}\Varid{α}}.\
       \ensuremath{\ann{\Varid{ε}\mathbin{*}\mathbin{⟨}\Conid{St},\mathbin{⟨}\ell_{2},\Varid{ε'}\mathbin{⟩}\mathbin{⟩}}{\susp{\Varid{α}}}\mathbin{→}\Conid{P}\mathbin{→}\ann{\mathbin{⟨⟩}\mathbin{*}(\Varid{ε}\mathbin{⊕}\mathbin{⟨}\Conid{St},\mathbin{⟨}\ell_{2},\Varid{ε'}\mathbin{⟩⟩})}{\susp{\Conid{G}\;\Varid{α}}}}
  \\ \ensuremath{\Varid{h}_{\Conid{St}}\mathbin{:}\mathbin{∀}\Varid{α}\;\Varid{r}\;\Varid{r'}}.\
     \ensuremath{\ann{\mathbin{⟨}\Conid{St},\Varid{r}\mathbin{⟩}\mathbin{*}\Varid{r'}}{\susp{\Varid{α}}}\mathbin{→}\Conid{S}\mathbin{→}\ann{\Varid{r}\mathbin{*}\mathbin{⟨}\Conid{St},\Varid{r'}\mathbin{⟩}}{\susp{(\Varid{α},\Varid{s})}}}
  \\ \ensuremath{\Varid{h}_2\mathbin{:}\mathbin{∀}\Varid{α}\;\Varid{r}\;\Varid{r'}}.\
       \ensuremath{\ann{\mathbin{⟨}\ell_{2},\Varid{r}\mathbin{⟩}\mathbin{*}\Varid{r'}}{\susp{\Varid{α}}}\mathbin{→}\Conid{P}_2\mathbin{→}\ann{\Varid{r}\mathbin{*}\mathbin{⟨}\Conid{St},\Varid{r'}\mathbin{⟩}}{\susp{\Conid{F}_2\;\Varid{α}}}}
  \\ \ensuremath{\Conid{H1}\mathbin{:}\Varid{h}\;(\Varid{h}_{\Conid{St}}^{\Varid{s'}}\;(\Varid{h}_2^{\Varid{p}_{2}}\;\Conid{E}[()])) \longtwoheadmapsto \Varid{v}_{\mathrm{12}}}
  \\ \ensuremath{\Conid{H2}\mathbin{:}\Varid{h}^{\Varid{p}}\;(\Varid{h}_2^{\Varid{p}_{2}}\;(\Varid{h}_{\Conid{St}}^{\Varid{s}}\;\Conid{E}[()])) \longtwoheadmapsto \Varid{v}_{\mathrm{21}}}
  \\ \ensuremath{\Conid{IH}\mathbin{:}}
       \ensuremath{\Varid{h}^{\Varid{p}}\;(\Varid{h}_2^{\Varid{p}_{2}}\;(\Varid{h}_{\Conid{St}}^{\Varid{s}}\;\Conid{E}[()])) \longtwoheadmapsto \Varid{v}_{\mathrm{21}} \Longrightarrow }
       \ensuremath{\Varid{v}_{\mathrm{12}}\sim\Varid{v}_{\mathrm{21}}}
  }
  { \ensuremath{\Varid{v}_{\mathrm{12}}\sim\Varid{v}_{\mathrm{21}}} }
\end{equation*}

Goal follows from \ensuremath{\Conid{IH}} and \ensuremath{\Conid{H2}}.

\end{case*}

\begin{case*}[\ensuremath{\op{op}^{\Conid{St}}\mathrel{=}\op{get}}]

Analogous to the case for \ensuremath{\op{put}}.

\end{case*}

\end{proof}

\subsection{S.o.C. Lemmas for the Catch Handler}

\subsubsection{Handler Return Lemma for the Catch Handler}

\begin{proof}
The goal is to show that
\begin{equation*}
\inferrule
{ \ensuremath{\Varid{hCatch}\mathbin{:}\mathbin{∀}\Varid{α}\;\Varid{r}\;\Varid{r'}}.\ \ensuremath{\ann{\mathbin{⟨}\Conid{Catch},\Varid{r}\mathbin{⟩}\mathbin{*}\Varid{r'}}{\susp{\Varid{α}}}\mathbin{→}()\mathbin{→}\ann{\Varid{r}\mathbin{*}\mathbin{⟨}\Conid{Catch},\Varid{r'}\mathbin{⟩}}{\susp{(\Conid{Maybe}\;\Varid{α})}}}
\\
  \ensuremath{\Varid{v}\mathbin{:}\Conid{X}} }
{ \ensuremath{\Varid{toBag}_{\Conid{Maybe}}\ (\Varid{hCatch}^{()}\;\susp{\Varid{v}})}
}
\end{equation*}

where

\begin{align*}
  \ensuremath{\Varid{toBag}_{\Conid{Maybe}}\ (\Conid{Just}\;\Varid{x})} &= \ensuremath{\{\mskip1.5mu \Varid{x}\mskip1.5mu\}}
\\
  \ensuremath{\Varid{toBag}_{\Conid{Maybe}}\ \Conid{Nothing}} &= \ensuremath{\{\mskip1.5mu \mskip1.5mu\}}
\end{align*}

The goal follows trivially from the definition of the return case of the catch handler.

\end{proof}

\subsubsection{Handler Op Lemma for the Catch Handler}
\label{app:handler-op-catch}

Assumes that \ensuremath{\Varid{h}} and \ensuremath{\Varid{h}_2} satisfy the handler return lemma.

\begin{proof}

\begin{case*}[\ensuremath{\op{op}^{\Conid{Catch}}\mathrel{=}\op{throw}}]

\begin{equation*}
  \inferrule*
  {  \ensuremath{\Varid{h}\mathbin{:}\mathbin{∀}\Varid{α}}.\
       \ensuremath{\ann{\Varid{ε}\mathbin{*}\mathbin{⟨}\Conid{Catch},\mathbin{⟨}\ell_{2},\Varid{ε'}\mathbin{⟩}\mathbin{⟩}}{\susp{\Varid{α}}}\mathbin{→}\Conid{P}\mathbin{→}\ann{\mathbin{⟨⟩}\mathbin{*}(\Varid{ε}\mathbin{⊕}\mathbin{⟨}\Conid{Catch},\mathbin{⟨}\ell_{2},\Varid{ε'}\mathbin{⟩⟩})}{\susp{\Conid{G}\;\Varid{α}}}}
  \\ \ensuremath{\Varid{hCatch}\mathbin{:}\mathbin{∀}\Varid{α}\;\Varid{r}\;\Varid{r'}}.\
     \ensuremath{\ann{\mathbin{⟨}\Conid{Catch},\Varid{r}\mathbin{⟩}\mathbin{*}\Varid{r'}}{\susp{\Varid{α}}}\mathbin{→}()\mathbin{→}\ann{\Varid{r}\mathbin{*}\mathbin{⟨}\Conid{Catch},\Varid{r'}\mathbin{⟩}}{\susp{(\Conid{Maybe}\;\Varid{α})}}}
  \\ \ensuremath{\Varid{h}_2\mathbin{:}\mathbin{∀}\Varid{α}\;\Varid{r}\;\Varid{r'}}.\
       \ensuremath{\ann{\mathbin{⟨}\ell_{2},\Varid{r}\mathbin{⟩}\mathbin{*}\Varid{r'}}{\susp{\Varid{α}}}\mathbin{→}\Conid{P}_2\mathbin{→}\ann{\Varid{r}\mathbin{*}\mathbin{⟨}\Conid{Catch},\Varid{r'}\mathbin{⟩}}{\susp{\Conid{F}_2\;\Varid{α}}}}
  \\ \ensuremath{\Conid{H0}\mathbin{:}\Varid{h}^{\Varid{p}}\;(\Varid{hCatch}^{()}\;(\Varid{h}_2^{\Varid{p}_{2}}\;\Conid{E}[\op{throw}])) \longmapsto \Varid{e}^{\prime}}
  \and \ensuremath{\Conid{H1}\mathbin{:}\Varid{e}^{\prime} \longtwoheadmapsto \Varid{v}_{\mathrm{12}}}
  \\ \ensuremath{\Conid{H2}\mathbin{:}\Varid{h}^{\Varid{p}}\;(\Varid{h}_2^{\Varid{p}_{2}}\;(\Varid{hCatch}^{()}\;\Conid{E}[\op{throw}])) \longtwoheadmapsto \Varid{v}_{\mathrm{21}}}
  \\ \ensuremath{\Conid{IH}\mathbin{:}} {\left(\begin{array}{l}
     \ensuremath{\mathbin{∀}e_{0}}.\
       \ensuremath{\Varid{e}^{\prime} \longtwoheadmapsto _{\beta}^{*}\Varid{h}^{\Varid{p}}\;(\Varid{hCatch}^{()}\;(\Varid{h}_2^{\Varid{p}_{2}}\;e_{0})) \Longrightarrow }
       \ensuremath{\Varid{h}^{\Varid{p}}\;(\Varid{h}_2^{\Varid{p}_{2}}\;(\Varid{hCatch}^{()}\;e_{0})) \longtwoheadmapsto \Varid{v}_{\mathrm{21}} \Longrightarrow }
       \ensuremath{\Varid{v}_{\mathrm{12}}\sim\Varid{v}_{\mathrm{21}}}
     \end{array}\right)}
  }
  { \ensuremath{\Varid{v}_{\mathrm{12}}\sim\Varid{v}_{\mathrm{21}}} }
\end{equation*}

Inversion on \ensuremath{\Conid{H0}}.

\begin{equation*}
  \inferrule*
  {  \ensuremath{\Varid{h}\mathbin{:}\mathbin{∀}\Varid{α}}.\
       \ensuremath{\ann{\Varid{ε}\mathbin{*}\mathbin{⟨}\Conid{Catch},\mathbin{⟨}\ell_{2},\Varid{ε'}\mathbin{⟩}\mathbin{⟩}}{\susp{\Varid{α}}}\mathbin{→}\Conid{P}\mathbin{→}\ann{\mathbin{⟨⟩}\mathbin{*}(\Varid{ε}\mathbin{⊕}\mathbin{⟨}\Conid{Catch},\mathbin{⟨}\ell_{2},\Varid{ε'}\mathbin{⟩⟩})}{\susp{\Conid{G}\;\Varid{α}}}}
  \\ \ensuremath{\Varid{hCatch}\mathbin{:}\mathbin{∀}\Varid{α}\;\Varid{r}\;\Varid{r'}}.\
     \ensuremath{\ann{\mathbin{⟨}\Conid{Catch},\Varid{r}\mathbin{⟩}\mathbin{*}\Varid{r'}}{\susp{\Varid{α}}}\mathbin{→}()\mathbin{→}\ann{\Varid{r}\mathbin{*}\mathbin{⟨}\Conid{Catch},\Varid{r'}\mathbin{⟩}}{\susp{(\Conid{Maybe}\;\Varid{α})}}}
  \\ \ensuremath{\Varid{h}_2\mathbin{:}\mathbin{∀}\Varid{α}\;\Varid{r}\;\Varid{r'}}.\
       \ensuremath{\ann{\mathbin{⟨}\ell_{2},\Varid{r}\mathbin{⟩}\mathbin{*}\Varid{r'}}{\susp{\Varid{α}}}\mathbin{→}\Conid{P}_2\mathbin{→}\ann{\Varid{r}\mathbin{*}\mathbin{⟨}\Conid{Catch},\Varid{r'}\mathbin{⟩}}{\susp{\Conid{F}_2\;\Varid{α}}}}
  \\ \ensuremath{\Conid{H1}\mathbin{:}\Varid{h}^{\Varid{p}}\;\susp{\Conid{Nothing}} \longtwoheadmapsto \Varid{v}_{\mathrm{12}}}
  \\ \ensuremath{\Conid{H2}\mathbin{:}\Varid{h}^{\Varid{p}}\;(\Varid{h}_2^{\Varid{p}_{2}}\;(\Varid{hCatch}^{()}\;\Conid{E}[\op{throw}])) \longtwoheadmapsto \Varid{v}_{\mathrm{21}}}
  \\ \ensuremath{\Conid{IH}\mathbin{:}} {\left(\begin{array}{l}
     \ensuremath{\mathbin{∀}e_{0}}.\
       \ensuremath{\Varid{e}^{\prime} \longtwoheadmapsto _{\beta}^{*}\Varid{h}^{\Varid{p}}\;(\Varid{hCatch}^{()}\;(\Varid{h}_2^{\Varid{p}_{2}}\;e_{0})) \Longrightarrow }
       \ensuremath{\Varid{h}^{\Varid{p}}\;(\Varid{h}_2^{\Varid{p}_{2}}\;(\Varid{hCatch}^{()}\;e_{0})) \longtwoheadmapsto \Varid{v}_{\mathrm{21}} \Longrightarrow }
       \ensuremath{\Varid{v}_{\mathrm{12}}\sim\Varid{v}_{\mathrm{21}}}
     \end{array}\right)}
  }
  { \ensuremath{\Varid{v}_{\mathrm{12}}\sim\Varid{v}_{\mathrm{21}}} }
\end{equation*}

Inversion on \ensuremath{\Conid{H2}} using definition of \ensuremath{ \longmapsto } and \ensuremath{\mathbin{-->}}.

\begin{equation*}
  \inferrule*
  {  \ensuremath{\Varid{h}\mathbin{:}\mathbin{∀}\Varid{α}}.\
       \ensuremath{\ann{\Varid{ε}\mathbin{*}\mathbin{⟨}\Conid{Catch},\mathbin{⟨}\ell_{2},\Varid{ε'}\mathbin{⟩}\mathbin{⟩}}{\susp{\Varid{α}}}\mathbin{→}\Conid{P}\mathbin{→}\ann{\mathbin{⟨⟩}\mathbin{*}(\Varid{ε}\mathbin{⊕}\mathbin{⟨}\Conid{Catch},\mathbin{⟨}\ell_{2},\Varid{ε'}\mathbin{⟩⟩})}{\susp{\Conid{G}\;\Varid{α}}}}
  \\ \ensuremath{\Varid{hCatch}\mathbin{:}\mathbin{∀}\Varid{α}\;\Varid{r}\;\Varid{r'}}.\
     \ensuremath{\ann{\mathbin{⟨}\Conid{Catch},\Varid{r}\mathbin{⟩}\mathbin{*}\Varid{r'}}{\susp{\Varid{α}}}\mathbin{→}()\mathbin{→}\ann{\Varid{r}\mathbin{*}\mathbin{⟨}\Conid{Catch},\Varid{r'}\mathbin{⟩}}{\susp{(\Conid{Maybe}\;\Varid{α})}}}
  \\ \ensuremath{\Varid{h}_2\mathbin{:}\mathbin{∀}\Varid{α}\;\Varid{r}\;\Varid{r'}}.\
       \ensuremath{\ann{\mathbin{⟨}\ell_{2},\Varid{r}\mathbin{⟩}\mathbin{*}\Varid{r'}}{\susp{\Varid{α}}}\mathbin{→}\Conid{P}_2\mathbin{→}\ann{\Varid{r}\mathbin{*}\mathbin{⟨}\Conid{Catch},\Varid{r'}\mathbin{⟩}}{\susp{\Conid{F}_2\;\Varid{α}}}}
  \\ \ensuremath{\Conid{H1}\mathbin{:}\Varid{h}^{\Varid{p}}\;\susp{\Conid{Nothing}} \longtwoheadmapsto \Varid{v}_{\mathrm{12}}}
  \\ \ensuremath{\Conid{H2}\mathbin{:}\Varid{h}^{\Varid{p}}\;(\Varid{h}_2^{\Varid{p}_{2}}\;\Conid{Nothing}) \longtwoheadmapsto \Varid{v}_{\mathrm{21}}}
  \\ \ensuremath{\Conid{IH}\mathbin{:}} {\left(\begin{array}{l}
     \ensuremath{\mathbin{∀}e_{0}}.\
       \ensuremath{\Varid{e}^{\prime} \longtwoheadmapsto _{\beta}^{*}\Varid{h}^{\Varid{p}}\;(\Varid{hCatch}^{()}\;(\Varid{h}_2^{\Varid{p}_{2}}\;e_{0})) \Longrightarrow }
       \ensuremath{\Varid{h}^{\Varid{p}}\;(\Varid{h}_2^{\Varid{p}_{2}}\;(\Varid{hCatch}^{()}\;e_{0})) \longtwoheadmapsto \Varid{v}_{\mathrm{21}} \Longrightarrow }
       \ensuremath{\Varid{v}_{\mathrm{12}}\sim\Varid{v}_{\mathrm{21}}}
     \end{array}\right)}
  }
  { \ensuremath{\Varid{v}_{\mathrm{12}}\sim\Varid{v}_{\mathrm{21}}} }
\end{equation*}

The goal follows from the fact that \ensuremath{\Varid{h}} and \ensuremath{\Varid{h}_2} satisfy the handler return lemma.

\end{case*}

\begin{case*}[\ensuremath{\op{op}^{\Conid{Catch}}\mathrel{=}\Varid{try\char95 catch}\;\Varid{m}_{1}\;\Varid{m}_{2}}]

The goal is

\begin{equation*}
  \inferrule*
  {  \ensuremath{\Varid{h}\mathbin{:}\mathbin{∀}\Varid{α}}.\
       \ensuremath{\ann{\Varid{ε}\mathbin{*}\mathbin{⟨}\Conid{Catch},\mathbin{⟨}\ell_{2},\Varid{ε'}\mathbin{⟩}\mathbin{⟩}}{\susp{\Varid{α}}}\mathbin{→}\Conid{P}\mathbin{→}\ann{\mathbin{⟨⟩}\mathbin{*}(\Varid{ε}\mathbin{⊕}\mathbin{⟨}\Conid{Catch},\mathbin{⟨}\ell_{2},\Varid{ε'}\mathbin{⟩⟩})}{\susp{\Conid{G}\;\Varid{α}}}}
  \\ \ensuremath{\Varid{hCatch}\mathbin{:}\mathbin{∀}\Varid{α}\;\Varid{r}\;\Varid{r'}}.\
     \ensuremath{\ann{\mathbin{⟨}\Conid{Catch},\Varid{r}\mathbin{⟩}\mathbin{*}\Varid{r'}}{\susp{\Varid{α}}}\mathbin{→}()\mathbin{→}\ann{\Varid{r}\mathbin{*}\mathbin{⟨}\Conid{Catch},\Varid{r'}\mathbin{⟩}}{\susp{(\Conid{Maybe}\;\Varid{α})}}}
  \\ \ensuremath{\Varid{h}_2\mathbin{:}\mathbin{∀}\Varid{α}\;\Varid{r}\;\Varid{r'}}.\
       \ensuremath{\ann{\mathbin{⟨}\ell_{2},\Varid{r}\mathbin{⟩}\mathbin{*}\Varid{r'}}{\susp{\Varid{α}}}\mathbin{→}\Conid{P}_2\mathbin{→}\ann{\Varid{r}\mathbin{*}\mathbin{⟨}\Conid{Catch},\Varid{r'}\mathbin{⟩}}{\susp{\Conid{F}_2\;\Varid{α}}}}
  \\ \ensuremath{\Conid{H0}\mathbin{:}\Varid{h}^{\Varid{p}}\;(\Varid{hCatch}^{()}\;(\Varid{h}_2^{\Varid{p}_{2}}\;\Conid{E}[(\op{catch}\;\Varid{m}_{1}\;\Varid{m}_{2})])) \longmapsto \Varid{e}^{\prime}}
  \and \ensuremath{\Conid{H1}\mathbin{:}\Varid{e}^{\prime} \longtwoheadmapsto \Varid{v}_{\mathrm{12}}}
  \\ \ensuremath{\Conid{H2}\mathbin{:}\Varid{h}^{\Varid{p}}\;(\Varid{h}_2^{\Varid{p}_{2}}\;(\Varid{hCatch}^{()}\;\Conid{E}[(\op{catch}\;\Varid{m}_{1}\;\Varid{m}_{2})])) \longtwoheadmapsto \Varid{v}_{\mathrm{21}}}
  \\ \ensuremath{\Conid{IH}\mathbin{:}} {\left(\begin{array}{l}
     \ensuremath{\mathbin{∀}e_{0}}.\
       \ensuremath{\Varid{e}^{\prime} \longtwoheadmapsto _{\beta}^{*}\Varid{h}^{\Varid{p}}\;(\Varid{hCatch}^{()}\;(\Varid{h}_2^{\Varid{p}_{2}}\;e_{0})) \Longrightarrow }
       \ensuremath{\Varid{h}^{\Varid{p}}\;(\Varid{h}_2^{\Varid{p}_{2}}\;(\Varid{hCatch}^{()}\;e_{0})) \longtwoheadmapsto \Varid{v}_{\mathrm{21}} \Longrightarrow }
       \ensuremath{\Varid{v}_{\mathrm{12}}\sim\Varid{v}_{\mathrm{21}}}
     \end{array}\right)}
  }
  { \ensuremath{\Varid{v}_{\mathrm{12}}\sim\Varid{v}_{\mathrm{21}}} }
\end{equation*}

Inversion on \ensuremath{\Conid{H0}} and \ensuremath{\Conid{H2}}.

\begin{equation*}
  \inferrule*
  {  \ensuremath{\Varid{h}\mathbin{:}\mathbin{∀}\Varid{α}}.\
       \ensuremath{\ann{\Varid{ε}\mathbin{*}\mathbin{⟨}\Conid{Catch},\mathbin{⟨}\ell_{2},\Varid{ε'}\mathbin{⟩}\mathbin{⟩}}{\susp{\Varid{α}}}\mathbin{→}\Conid{P}\mathbin{→}\ann{\mathbin{⟨⟩}\mathbin{*}(\Varid{ε}\mathbin{⊕}\mathbin{⟨}\Conid{Catch},\mathbin{⟨}\ell_{2},\Varid{ε'}\mathbin{⟩⟩})}{\susp{\Conid{G}\;\Varid{α}}}}
  \\ \ensuremath{\Varid{hCatch}\mathbin{:}\mathbin{∀}\Varid{α}\;\Varid{r}\;\Varid{r'}}.\
     \ensuremath{\ann{\mathbin{⟨}\Conid{Catch},\Varid{r}\mathbin{⟩}\mathbin{*}\Varid{r'}}{\susp{\Varid{α}}}\mathbin{→}()\mathbin{→}\ann{\Varid{r}\mathbin{*}\mathbin{⟨}\Conid{Catch},\Varid{r'}\mathbin{⟩}}{\susp{(\Conid{Maybe}\;\Varid{α})}}}
  \\ \ensuremath{\Varid{h}_2\mathbin{:}\mathbin{∀}\Varid{α}\;\Varid{r}\;\Varid{r'}}.\
       \ensuremath{\ann{\mathbin{⟨}\ell_{2},\Varid{r}\mathbin{⟩}\mathbin{*}\Varid{r'}}{\susp{\Varid{α}}}\mathbin{→}\Conid{P}_2\mathbin{→}\ann{\Varid{r}\mathbin{*}\mathbin{⟨}\Conid{Catch},\Varid{r'}\mathbin{⟩}}{\susp{\Conid{F}_2\;\Varid{α}}}}
  \\ \ensuremath{\Conid{H1}\mathbin{:}\Varid{e}^{\prime}\mathbin{≡}\Varid{h}^{\Varid{p}}\;((\Varid{λ}\;\Varid{x}\;\Varid{q}\;\!.\!\;\Varid{hCatch}^{\Varid{q}}\;(\Varid{h}_2^{\Varid{p}_{2}}\;\Conid{E}[\Varid{x}\mathbin{!}]))\;\susp{\begin{array}{l} \kw{match}(\Varid{hCatch}^{()}\;\Varid{m}_{1})\\\;\mid\Conid{Nothing}\mathbin{→}\Varid{m}_{2}\mathbin{!}\\\;\mid(\Conid{Just}\;\Varid{x})\mathbin{→}\Varid{x}\end{array}}\;()) \longtwoheadmapsto \Varid{v}_{\mathrm{12}}}
  \\ \ensuremath{\Conid{H2}\mathbin{:}\Varid{h}^{\Varid{p}}\;(\Varid{h}_2^{\Varid{p}_{2}}\;((\Varid{λ}\;\Varid{x}\;\Varid{q}\;\!.\!\;(\Varid{hCatch}^{\Varid{q}}\;\Conid{E}[\Varid{x}\mathbin{!}]))\;\susp{\begin{array}{l} \kw{match}(\Varid{hCatch}^{()}\;\Varid{m}_{1})\\\;\mid\Conid{Nothing}\mathbin{→}\Varid{m}_{2}\mathbin{!}\\\;\mid(\Conid{Just}\;\Varid{x})\mathbin{→}\Varid{x}\end{array}})\;()) \longtwoheadmapsto \Varid{v}_{\mathrm{21}}}
  \\ \ensuremath{\Conid{IH}\mathbin{:}} {\left(\begin{array}{l}
     \ensuremath{\mathbin{∀}e_{0}}.\
       \ensuremath{\Varid{e}^{\prime} \longtwoheadmapsto _{\beta}^{*}\Varid{h}^{\Varid{p}}\;(\Varid{hCatch}^{()}\;(\Varid{h}_2^{\Varid{p}_{2}}\;e_{0})) \Longrightarrow }
       \ensuremath{\Varid{h}^{\Varid{p}}\;(\Varid{h}_2^{\Varid{p}_{2}}\;(\Varid{hCatch}^{()}\;e_{0})) \longtwoheadmapsto \Varid{v}_{\mathrm{21}} \Longrightarrow }
       \ensuremath{\Varid{v}_{\mathrm{12}}\sim\Varid{v}_{\mathrm{21}}}
     \end{array}\right)}
  }
  { \ensuremath{\Varid{v}_{\mathrm{12}}\sim\Varid{v}_{\mathrm{21}}} }
\end{equation*}

Since

\begin{align*}
  &\ensuremath{\Varid{h}^{\Varid{p}}\;((\Varid{λ}\;\Varid{x}\;\Varid{q}\;\!.\!\;\Varid{hCatch}^{\Varid{q}}\;\susp{\Varid{h}_2^{\Varid{p}_{2}}\ \Conid{E}[\Varid{x}\mathbin{!}]})\;\susp{\begin{array}{l} \kw{match}(\Varid{hCatch}^{()}\;\Varid{m}_{1})\\\;\mid\Conid{Nothing}\mathbin{→}\Varid{m}_{2}\mathbin{!}\\\;\mid(\Conid{Just}\;\Varid{x})\mathbin{→}\Varid{x}\end{array}}\;())}
\\
  &
  \ensuremath{ \longtwoheadmapsto _{\beta}^{*}\Varid{h}^{\Varid{p}}\;(\Varid{hCatch}^{()}\;(\Varid{h}_2^{\Varid{p}_{2}}\;\Conid{E}[\susp{\begin{array}{l} \kw{match}(\Varid{hCatch}^{()}\;\Varid{m}_{1})\\\;\mid\Conid{Nothing}\mathbin{→}\Varid{m}_{2}\mathbin{!}\\\;\mid(\Conid{Just}\;\Varid{x})\mathbin{→}\Varid{x}\end{array}}\mathbin{!}]))}
\end{align*}

and

\begin{align*}
  &\ensuremath{\Varid{h}^{\Varid{p}}\;(\Varid{h}_2^{\Varid{p}_{2}}\;((\Varid{λ}\;\Varid{x}\;\Varid{q}\;\!.\!\;\Varid{hCatch}^{\Varid{q}}\;\Conid{E}[\Varid{x}\mathbin{!}])\;\susp{\begin{array}{l} \kw{match}(\Varid{hCatch}^{()}\;\Varid{m}_{1})\\\;\mid\Conid{Nothing}\mathbin{→}\Varid{m}_{2}\mathbin{!}\\\;\mid(\Conid{Just}\;\Varid{x})\mathbin{→}\Varid{x}\end{array}}\;()))}
\\
  &
  \ensuremath{ \longtwoheadmapsto _{\beta}^{*}\Varid{h}^{\Varid{p}}\;(\Varid{h}_2^{\Varid{p}_{2}}\;(\Varid{hCatch}^{()}\;\Conid{E}[\susp{\begin{array}{l} \kw{match}(\Varid{hCatch}^{()}\;\Varid{m}_{1})\\\;\mid\Conid{Nothing}\mathbin{→}\Varid{m}_{2}\mathbin{!}\\\;\mid(\Conid{Just}\;\Varid{x})\mathbin{→}\Varid{x}\end{array}}\mathbin{!}]))}
\end{align*}

the IH applies, and the goal follows.

\end{case*}

\end{proof}

%
%

\end{document}